\documentclass[showpacs,twocolumn,floatfix,superscriptaddress,notitlepage]{revtex4-2}
\usepackage{graphicx}
\usepackage{dcolumn}
\usepackage{bm}
\usepackage[colorlinks=true,urlcolor=blue,citecolor=blue,linkcolor=blue]{hyperref}
\usepackage{suffix}

\usepackage{xr}
\usepackage{multirow}
\usepackage{minitoc}
\usepackage{xfrac}
\usepackage{amsmath}
\usepackage{amssymb}
\usepackage{braket}
\usepackage{physics}
\usepackage{amsthm}
\usepackage{natbib}
\usepackage{mathtools}
\usepackage{diagbox}
\usepackage[utf8]{inputenc}
\usepackage[english]{babel}
\usepackage{scalerel}[2014/03/10]
\usepackage{stackengine}
\usepackage{algorithm}
\usepackage[noend]{algpseudocode} 
\usepackage{tikz}
\usetikzlibrary{quantikz}
\usepackage{subfigure}
\newcommand{\algmargin}{\the\ALG@thistlm}
\makeatother
\newlength{\whilewidth}
\algdef{SE}[parWHILE]{parWhile}{EndparWhile}[1]
  {\parbox[t]{\dimexpr\linewidth-\algmargin}{%
     \hangindent\whilewidth\strut\algorithmicwhile\ #1\ \algorithmicdo\strut}}{\algorithmicend\ \algorithmicwhile}%
\algnewcommand{\parState}[1]{\State%
  \parbox[t]{\dimexpr\linewidth-\algmargin}{\strut #1\strut}}

\newtheorem{theorem}{Theorem}[]
\newtheorem*{remark}{Remark}
\newtheorem{corollary}{Corollary}[]
\newtheorem{lemma}[]{Lemma}
\newtheorem{prop}{Proposition}

\newtheorem{definition}{Definition}

\theoremstyle{definition}
\newtheorem{proofpart}{Part}
\newtheorem{proofcase}{Case}
\makeatletter
\@addtoreset{proofpart}{theorem}
\@addtoreset{proofcase}{theorem}
\makeatother
\renewcommand{\exp}[1]{\text{exp}\left( #1 \right)}
\newcommand{\g}[1]{{g \left( #1 \right)}}

\begin{document}

\title{Simulating Noisy Variational Quantum Algorithms: A Polynomial Approach}

\author{Yuguo Shao}
\affiliation{Yau Mathematical Sciences Center and Department of Mathematics, Tsinghua University,  Beijing 100084, China}
\author{Fuchuan Wei}%
\affiliation{Yau Mathematical Sciences Center and Department of Mathematics, Tsinghua University,  Beijing 100084, China}
\author{Song Cheng}%
\email{chengsong@bimsa.cn}
\affiliation{Yanqi Lake Beijing Institute of Mathematical Sciences and Applications, Beijing 100407, China}
\author{Zhengwei Liu}%
 \email{liuzhengwei@mail.tsinghua.edu.cn}
\affiliation{Yau Mathematical Sciences Center and Department of Mathematics, Tsinghua University,  Beijing 100084, China}%
\affiliation{Yanqi Lake Beijing Institute of Mathematical Sciences and Applications, Beijing 100407, China}


\begin{abstract}
Large-scale variational quantum algorithms are widely recognized as a potential pathway to achieve practical quantum advantages. 
However, the presence of quantum noise might suppress and undermine these advantages, which blurs the boundaries of classical simulability. 
To gain further clarity on this matter, we present a novel polynomial-scale method based on the path integral of observable's back-propagation on Pauli paths (OBPPP). 
This method efficiently approximates expectation values of operators in variational quantum algorithms with bounded truncation error in the presence of single-qubit Pauli noise. Theoretically, we rigorously prove:
1) For a constant minimal non-zero noise rate $\gamma$, OBPPP's time and space complexity exhibit a polynomial relationship with the number of qubits $n$, the circuit depth $L$. 
2) For variable $\gamma$, in scenarios where more than two non-zero noise factors exist, the complexity remains $\mathrm{Poly}\left(n,L\right)$ if $\gamma$ exceeds $\sfrac{1}{\log{L}}$, but grows exponential with $L$ when $\gamma$ falls below $\sfrac{1}{L}$. 
Numerically, we conduct classical simulations of IBM's zero-noise extrapolated experimental results on the 127-qubit Eagle processor [Nature \textbf{618}, 500 (2023)].  
Our method attains higher accuracy and faster runtime compared to the quantum device.
Furthermore, our approach allows us to simulate noisy outcomes, enabling accurate reproduction of IBM's unmitigated results that directly correspond to raw experimental observations.
Our research reveals the vital role of noise in classical simulations and the derived method is general in computing the expected value for a broad class of quantum circuits and can be applied in the verification of quantum computers. 
\end{abstract}

\maketitle

\section{Introduction}
\label{sec:Introduction}
In the current Noisy Intermediate-Scale Quantum (NISQ) era~\cite{preskill2018quantum,bharti2022noisy,chen2022complexity}, Variational Quantum Algorithms (VQAs)~\cite{Cerezo2021variational,mcclean2016theory,tilly2022variational} play a significant role in diverse fields like combinatorial optimization~\cite{farhi2014quantum,moll2018quantum}, quantum chemistry~\cite{peruzzo2014variational, kandala2017hardware,li2022toward}, quantum machine learning~\cite{beer2020training,huang2021experimental,havlivcek2019supervised,mitarai2018quantum}, quantum circuit compilation~\cite{khatri2019quantum}, and quantum error correction~\cite{johnson2017qvector,xu2021variational}, etc.

In practice, NISQ devices are inevitably affected by noises. These noises would decoherent quantum systems and cause quantum states to collapse, thereby limiting the quantum advantages \cite{stilck2021limitations,wu2023complexity,yan2023limitations,de2023limitations,marshall2020characterizing,gao2018efficient,ben2013quantum}.
On the other hand, noise potentially enables the simulability of complex quantum algorithms by classical methods~\cite{hangleiter2022computational,noh2020efficient,bremner2017achieving,aharonov2022polynomial,yung2017can,fefferman2023effect,mele2024noise,singkanipa2024beyond}.
For instance, in noiseless circuits, Random Circuit Sampling (RCS) tasks have been proven to be difficult to simulate classically \cite{bremner2011classical, aaronson2011computational, bouland2019complexity, movassagh2018efficient, movassagh2019quantum}. However, a polynomial-time algorithm for simulating noisy RCS has been established in the presence of depolarizing noises \cite{aharonov2022polynomial}. For general cases, noisy simulation algorithms based on tensor networks also exhibit decreasing computational complexity as the noise rate increases \cite{noh2020efficient, cheng2021mpdo,zhou2020limits,ayral2023density}.

\begin{figure}[htbp]
  \includegraphics[width=240px]{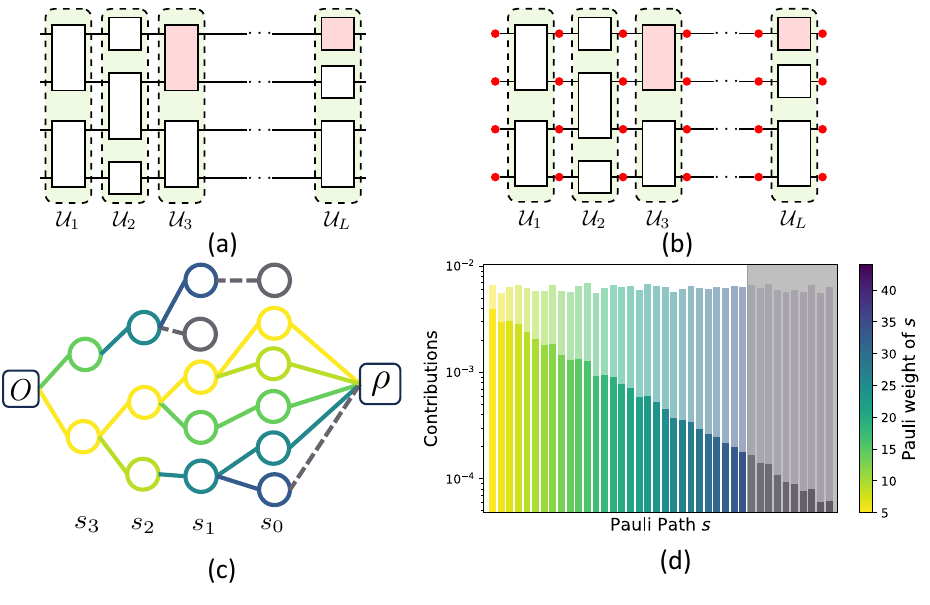}
  \caption{(a) Noiseless parametrized quantum circuits consists of Pauli rotation gates (white boxes) and selected Clifford gates (pink boxes). (b) Single-qubit Pauli noise applied independently to each qubit (red points). (c) A diagrammatic illustration of our method. The Pauli operator back-propagates from the observable through the circuit and bifurcates during the propagation.
  (d) Under the presence of noise, the contributions of Pauli paths are suppressed exponentially with their Pauli weight. Translucent bars represent noiseless contributions, while opaque bars show the contributions after applying single-qubit Pauli noise.
  Contributions from Pauli paths with high Pauli weights are truncated and represented by dashed lines in (c) and the shaded area in (d).}
  \label{fig:first}
\end{figure}
In this work, we propose OBPPP, a novel polynomial-scale method for approximating expectation values in noisy VQAs using a parameterized quantum circuit with Pauli rotation gates and ${{\mathrm{H},\mathrm{S},\mathrm{CNOT}}}$. We leverage the path integral on the Pauli basis, 
which could also be viewed as the Fourier transformation on quantum circuits~\cite{aharonov2022polynomial,gao2018efficient,wang2019quantifying}.
Adopting the Pauli basis offers two main advantages. Firstly, if the system is sparse under this basis, OBPPP could leverage it to accelerate computations. Secondly, Pauli noise can heavily suppress contributions from high-weight Pauli paths, thereby limiting truncation errors.

Fig.~\ref{fig:first} shows the quantum circuits OBPPP is applied, outlining its basic principles, as well as the relationship between noise and algorithm efficiency.

Compared to related simulators on quantum circuits~\cite{noh2020efficient, cheng2021mpdo,bravyi2021classical, napp2022efficient,aharonov2022polynomial,mele2024noise}, OBPPP is not constrained by geometric structure requirements and is less affected by circuit depth. Moreover, it does not require the locality of observables. 
Our approach simulates many commonly utilized noisy VQAs with a polynomial time complexity, and could be served as a benchmark for assessing the capabilities of NISQ computation.

\section{Notations and Prerequisites}
\label{sec:Notations and assumptions}
In typical VQAs, the cost function is determined as:
\begin{equation}\label{eq:define_L}
\mathcal{L}(\bm{\theta})=\Tr{O\mathcal{U}(\bm{\theta})\rho \mathcal{U}^\dagger(\bm{\theta})},
\end{equation}
where $\rho$ is the density matrix of the $n$-qubit input state, and $O$ is the observable, $\mathcal{U}(\bm{\theta})=\mathcal{U}_L(\bm{\theta}_L)\mathcal{U}_{L-1}(\bm{\theta}_{L-1})\cdots \mathcal{U}_1(\bm{\theta}_1)$ is a parameterized quantum circuit, which is composed with $L$ layers unitary transformation $\mathcal{U}_i(\bm{\theta}_i)$ and $\bm{\theta}=(\bm{\theta}_1,\cdots,\bm{\theta}_L)$.

The $\mathcal{U}_i(\bm{\theta}_i)$ in each layer consists of $R_i$ rotation gates and $C_i$ Clifford gates that act on mutually disjoint qubits, $\bm{\theta}_i=(\theta_{i,1},\cdots,\theta_{i,R_i})$ denote the parameter vector of the $i$-th layer.
Specifically, the $j$-th rotation gate in the $i$-th layer is denoted as $U_{i, j}(\theta_{i,j})=\exp{-i \frac{\theta_{i,j}}{2} \sigma_{i,j}}$, where $j \in \{1, \cdots , R_i\}$, $\theta_{i,j}$ is the variational parameter and $\sigma_{i, j}\in \{\mathbb{I}, X,Y,Z\}^{\otimes n}$.
Similarly, the $k$-th Clifford gate in the $i$-th layer is denoted as $V_{i,k}$, where $k \in \{1, \cdots , C_i\}$. $V_{i,k}\in\{\mathrm{H}(a),\mathrm{S}(a),\mathrm{CNOT}(a,b)\}$, where $a,b $ refers to the index of qubit where the gate acts on.

The set of Pauli words for all $U_{i,j}(\theta_{i,j})$ is $\{\sigma_{i,j}\}$.
We denote the set $\{\overline{\sigma}_{i,j}\}$ as all $\sigma_{i,j}$ after Clifford gates transformation, i.e.~$\overline{\sigma}_{i,j}= \mathcal{V}_{L} \cdots \mathcal{V}_{i} \sigma_{i,j} \mathcal{V}_{i}^\dagger \cdots \mathcal{V}_{L}^\dagger$, where $\mathcal{V}_{i} = \prod_{k=1}^{C_i} V_{i,k}$ is the unitary transformation corresponds to the tensor product of all Clifford gates in the $i$-th layer.

To ensure the validity of Lemma~\ref{lemma:MSE_l}, we require an easily achievable prerequisite: the set $\{\overline{\sigma}_{i,j}\}$ could generate $\{\mathbb{I}, X, Y, Z\}^{\otimes n}$ up to a phase of $\{e^{i\psi}|\psi=0,\frac{\pi}{2},\pi,\frac{3\pi}{2}\}$, formulated as
\begin{equation}\label{eq:generate}
  \langle \{\overline{\sigma}_{i,j}\}\rangle/\left(\langle \{\overline{\sigma}_{i,j}\}\rangle\cap\langle i\mathbb{I}^{\otimes n}\rangle\right)=\{\mathbb{I},X,Y,Z\}^{\otimes n},
\end{equation}
here $\langle \{\overline{\sigma}_{i,j}\} \rangle$ refers to the Pauli subgroup that is generated by set $\{\overline{\sigma}_{i,j}\}$, which means elements in $\langle \{\overline{\sigma}_{i,j}\}\rangle$ can be expressed as the finite product of elements in $\{\overline{\sigma}_{i,j}\}$.
To demonstrate the condition is indeed easily met, considering there are a layer of $R_X$ gates and a layer of $R_Z$ gates acting on each qubit in the circuit at the last two layers, then $\{X_i, Z_i\}_{i=1,\cdots,n}$ is contained in $\{\overline{\sigma}_{i,j}\}$. These $\{ X_i,Z_i \}_{i=1,\cdots,n}$ are enough to generate $\{\mathbb{I},X,Y,Z\}^{\otimes n}$. In fact, this sufficient condition can be further weakened, as shown in Suppl. Mat.~{VIII}. 

We also demand $\rho$ and $O$ to be sparse. More precisely, we require the number of the non-zero elements of $\rho=\sum_{a,b}\rho_{a,b}\ketbra{a}{b}$ is polynomially related to the number of qubits $n$, denoted as $\mathrm{Poly}(n)$, where $\ket{a}$ and $\ket{b}$ are computational basis states. The number of all Pauli words $\{\sigma\}$ which linearly compose $O$ is also restricted to be $\mathrm{Poly}(n)$. This constraint is naturally satisfied by widespread VQA frameworks such as the Variational Quantum Eigensolver (VQE) and the Quantum Approximate Optimization Algorithm (QAOA)~\cite{peruzzo2014variational,farhi2014quantum,beer2020training,mitarai2018quantum}.

In this work, the single-qubit Pauli noise $\mathcal{N}$ is assumed to occur independently before each layer and the final observation operator $O$. This noise channel $\mathcal{N}$ is modeled as
$\mathcal{N}(\phi)=(1-p_x-p_y-p_z)\phi+ p_x X\phi X+ p_y Y\phi Y+ p_z Z\phi Z$,
where $\phi$ is a single-qubit density matrix and $p_x,p_y,p_z$ are the probabilities of $X,Y,Z$ error occuring, respectively. 

We define $\widehat{\mathcal{L}}$ as the cost function under this noise:
\begin{equation}\label{eq:define_NF}
\widehat{\mathcal{L}}(\bm{\theta})=\Tr{ O \mathcal{N}^{\otimes n}(\mathcal{U}_L \mathcal{N}^{\otimes n}(\cdots\mathcal{U}_1\mathcal{N}^{\otimes n}(\rho) \mathcal{U}_1^\dagger\cdots) \mathcal{U}_L^\dagger)}.
\end{equation}

\section{Simulation Method}
\label{sec:Method}
The main idea of our approach is to express $\mathcal{L}$ as the path integral of the matrix algebra, in which we select Pauli operators as the basis.

A Pauli path is a sequence $s=(s_0,\cdots,s_L)\in \bm{P}^{L+1}_n$, where $\bm{P}_n=\{\sfrac{\mathbb{I}}{\sqrt{2}},\sfrac{X}{\sqrt{2}},\sfrac{Y}{\sqrt{2}},\sfrac{Z}{\sqrt{2}}\}^{\otimes n} $ represents the set of all normalized $n$-qubit Pauli words. 
As detailed shown in Suppl. Mat.~{III}, the noiseless cost function can be expressed as the sum of the contributions of all Pauli paths, given by
$\mathcal{L}(\bm{\theta})=\sum_{s\in \bm{P}^{L+1}_n} f(\bm{\theta},s,O,\rho)$
~\footnote{
The superoperator $\mathcal{S}_i$, as defined in~\cite{wood2011tensor}, is denoted as the operator $\mathcal{\overline{U}}_i\otimes\mathcal{U}_i$, and the symbol $| \cdot \rangle\!\rangle$ indicates the vectorization of a matrix. The expected value $\mathcal{L}(\bm{\theta})$ is given by the expression:
$\langle\!\langle H| \prod_{i=1}^{L} \mathcal{S}_i |\rho \rangle\!\rangle=\sum_{s\in \bm{P}^{L+1}_n}\langle\!\langle H| s_L \rangle\!\rangle \left(\prod_{i=1}^{L} \langle\!\langle s_i| \mathcal{S}_i| s_{i-1} \rangle\!\rangle \right) \langle\!\langle s_0 | \rho \rangle\!\rangle$.
Here, each term of the summation is an alternative representation of the expression in Eq.~\eqref{eq:f}.\label{footnote:1}
},
where $f(\bm{\theta},s,O,\rho)$ denotes the contribution of one particular Pauli path $s=(s_0,\cdots,s_L)\in \bm{P}^{L+1}_n$:
\begin{equation}\label{eq:f}
\begin{aligned}
f(\bm{\theta},s,O,\rho)=&\Tr{Os_L}\left(\prod_{i=1}^{L}\Tr{s_i\mathcal{U}_i s_{i-1}\mathcal{U}_i^\dagger}\right)\Tr{s_0\rho}.
\end{aligned}
\end{equation}
In the following discussion, we aim to establish that the time complexity for computing each $f(\bm{\theta},s, O,\rho)$ is $\order{nL}+\mathrm{Poly}(n)$.

Firstly, the term $\Tr{Os_L}$ requires the Pauli word $s_L$ must be in $O$, otherwise $f(\bm{\theta},s, O,\rho) = 0$, thereby resulting in an cost of $\order{n}$. Similarly, the input term $\Tr{s_0\rho}$ can be achieved with time complexity of $\mathrm{Poly}(n)$, facilitated by the polynomial-size non-zero elements in $\rho$. More details are provided in Suppl. Mat.~{II}.

Furthermore, for the calculation of the $i$-th layer term $\Tr{s_i\mathcal{U}_i s_{i-1}\mathcal{U}_i^\dagger}$, we propose the following proposition:

\begin{widetext}
\begin{prop}\label{prop:f_ele}
The time and space complexity of calculate the $i$-th layer term in $f$ is of $\order{n}$ by the equality:
\begin{equation}\label{eq:i-layer_terms}
\begin{aligned}
\Tr{s_i\mathcal{U}_i s_{i-1}\mathcal{U}_i^\dagger}&=\Tr{\left(s_i s_{i-1}\right)\big|_{I_i}}\prod_{k=1}^{C_i}\Tr{\left(s_iV_{i,k} s_{i-1}V_{i,k}^{\dagger}\right)\big|_\g{V_{i,k}}}\prod_{\sigma_{i,j}\in C(i,s_{i-1})}\Tr{\left(s_i s_{i-1}\right)\big|_\g{\sigma_{i,j}}}\\
&\prod_{\sigma_{i,j'}\in AC(i,s_{i-1})} \Bigg\{ \Tr{\left(s_i s_{i-1}\right)\big|_\g{\sigma_{i,j'}}}\cos{\theta_{i,j'}}- \Tr{\left(is_i \sigma_{i,j'}s_{i-1}\right)\big|_\g{\sigma_{i,j'}}}\sin{\theta_{i,j'}}  \Bigg\}.
\end{aligned}
\end{equation}
We define $g:\{\mathbb{I}, X, Y, Z\}^{\otimes n}\cup\{\mathrm{CNOT}_{a,b},\mathrm{H}_a,\mathrm{S}_a\}\rightarrow2^{\{1,\cdots, n\}}$ as a map from a unitary operator to the indices of qubits where the unitary operator's action is non-identity. Here $2^{\{1,\cdots, n\}}$ represents all subsets of $\{1,\cdots, n\}$.
For simplification, we divide the indices of $n$ qubits in the $i$-th layer into three sets based on the type of gates applied to them. These sets are denoted as the symbol $\big|_{I_i}$, $\big|_\g{V_{i,k}}$ and $\big|_\g{\sigma_{i,j}}$, corresponding to the identity gate, the Clifford gate and the Pauli rotation gate, respectively. Additionally, the sets $C(i,s_{i-1})$ and $AC(i,s_{i-1})$ denote the sets of Pauli words in $\{\sigma_{i,j}\}_{j=1}^{R_i}$ that commute and anti-commute with $s_{i-1}$, respectively.
\end{prop}
\end{widetext}

\begin{remark}\label{remark:f_ele}
By utilizing the orthonormality property of Pauli words, we can establish the following relations for $\Tr{s_i\mathcal{U}_i s_{i-1}\mathcal{U}_i^\dagger} \neq 0$:

\resizebox{0.95\columnwidth}{!}{
\begin{tabular}{|c|c|c|}
  \hline
   Terms & Relations &Factor in $f$\\
  \hline
  ${I_i}$ &$s_{i-1}\big|_{I_i} = s_{i}\big|_{I_i}$& 1 \\
  \hline
  $V_{i,k}$ &$s_{i-1}\big|_\g{V_{i,k}}= \pm V_{i,k}^{\dagger} s_{i}V_{i,k}\big|_\g{V_{i,k}}$& $\pm 1$ \\
  \hline
  $C(i,s_{i-1})=C(i,s_{i})$&$s_{i}|_\g{\sigma_{i,j}}=s_{i-1}|_\g{\sigma_{i,j}}$& 1 \\
  \hline
  \multirow{2}{*}{$AC(i,s_{i-1})=AC(i,s_{i})$}
  &$s_{i}|_\g{\sigma_{i,j}}=s_{i-1}|_\g{\sigma_{i,j}}$& $\cos{\theta_{i,j}}$ \\
  \cline{2-3}
  &$s_{i}|_\g{\sigma_{i,j}}=\pm i \sigma_{i,j} s_{i-1}|_\g{\sigma_{i,j}}$& $\mp \sin{\theta_{i,j}}$ \\
  \hline
\end{tabular}}

\end{remark}
A crucial point is that each item in $AC(i,s_{i-1})$ splits a Pauli path into two, whereas others do not. This determined the dynamics of the observable's back-propagation on Pauli paths as shown in Fig.~\ref{fig:first}(c). In OBPPP, we initially select all $s_L$ which are included in $O$. For each instance of $s_L$, we enumerate all potential $s_{L-1}$, resulting in at most $2^\abs{s_L}$ configurations of $s_{L-1}$. This enumeration process is iteratively applied back to $s_0$.

This complexity can be further suppressed by the presence of single-qubit Pauli noise. The following lemma is an extension of the lemma proposed in \cite{aharonov2022polynomial}.
\begin{lemma}\label{lemma:f_noisy}
Let $\hat{f}(\bm{\theta},s,O,\rho)$ be the contribution of a Pauli path $s=(s_0,\cdots,s_L)\in \bm{P}^{L+1}_n$ in the noisy cost function $\widehat{\mathcal{L}}(\bm{\theta})$. In the presence of the single-qubit Pauli noise, the relationship between the noiseless contribution $f$ and $\hat{f}$ can be characterized as follows:
\begin{equation}\label{eq:noisy_f_term_decrease}
  \begin{aligned}
    \hat{f}(\bm{\theta},s,O,\rho)=&(1-2(p_y+p_z))^{\abs{s}_X}(1-2(p_x+p_z))^{\abs{s}_Y}\\
    &(1-2(p_x+p_y))^{\abs{s}_Z}f(\bm{\theta},s,O,\rho),
  \end{aligned}
\end{equation}
where $\abs{s}_P=\sum_{i}\abs{s_i}_P$ represents the aggregate count of Pauli operator $P \in \{\frac{X}{\sqrt{2}},\frac{Y}{\sqrt{2}},\frac{Z}{\sqrt{2}}\}$ among all $(L+1)n$ individual Pauli operators, $\abs{s_i}_P$ denotes the number of element $P$ in $s_i$.
\end{lemma}
Lemma~\ref{lemma:f_noisy} states that the path integral in the Pauli basis provides a convenient approach for quantifying the impact of noise. In essence, by estimating all noiseless contributions $f(\bm{\theta},s, O,\rho)$, it is sufficient to evaluate the noisy cost function $\widehat{\mathcal{L}}$.

To estimate $\widehat{\mathcal{L}}$, our discussion is divided into two cases based on the number of non-zero factor in noise:
\begin{itemize}
  \item Case 1: At least two non-zero elements in $\{p_x,p_y,p_z\}$, We define the Pauli weight as $\abs{s}:=\sum_{P\in \{X,Y,Z\}}\abs{s}_P$, which commonly referred as the Hamming weight. An example of case 1 is the depolarizing channel.
  \item Case 2: Only one element in $\{p_x,p_y,p_z\}$ non-zero~(without loss of generality, $p_z\neq 0$), Pauli weight is defined as $\abs{s}:=\sum_{P\in \{X,Y\}}\abs{s}_P$. The dephasing channel is an example of case 2.
\end{itemize}

OBPPP calculates all contributions of the Pauli paths with $\abs{s}\leq M$ to approximate $\widehat{\mathcal{L}}$. 
Here, let $\widetilde{\mathcal{L}}(\bm{\theta}):=\sum_{\abs{s}\leq M} \hat{f}(\bm{\theta},s,O,\rho)$~\footnote{Similar to \cite{Note1}, in noisy case, the different is to replace the superoperator $\mathcal{S}_i$ from noiseless operators $\mathcal{\overline{U}}_i\otimes\mathcal{U}_i$ to noisy operators $(\mathcal{\overline{U}}_i\otimes\mathcal{U}_i) \cdot\mathcal{N}^{\otimes n}$} represent the approximate noisy cost function. To estimate the complexity of calculating $\widetilde{\mathcal{L}}$ with bounded truncation error, we introduce the following lemma and theorem.

\begin{lemma}\label{lemma:MSE_l}
  Suppose Eq.~\eqref{eq:generate} is satisfied, for $\forall\nu > 0$, given $M\geq\frac{1}{4\gamma}\ln{\frac{\norm{O}_\infty^2}{\nu}}$, the mean-square error $\mathbb{E}_{\bm{\theta}}\abs{\widetilde{\mathcal{L}}-\widehat{\mathcal{L}}}^2$ is upper bounded by $\nu$, where $\gamma:=\min\{p|{p \in \{p_x,p_y,p_z\},p\neq 0}\}$.
\end{lemma}

\begin{theorem}\label{thm:main}
Suppose Eq.~\eqref{eq:generate} is satisfied and $O,\rho$ are sparse, for a fixed $\gamma$, given arbitrary truncation error $\varepsilon$, there exists a polynomial-scale classical algorithm to determine the approximated noisy cost function $\widetilde{\mathcal{L}}$, which satisfies $\abs{\widetilde{\mathcal{L}}-\widehat{\mathcal{L}}} \leq \varepsilon$ with a probability of at least $1-\delta$ over all possible parameters $\bm{\theta}$. The time complexity is $\mathrm{Poly}(n) \order{L} \bigg(\frac{\norm{O}_\infty}{\varepsilon \sqrt{\delta}} \bigg)^{\order{\sfrac{1}{\gamma}}}$ for Case 1 and $\mathrm{Poly}(n)  \order{(nL)^{\frac{1}{2\gamma} \ln{\frac{\norm{O}_\infty}{\varepsilon \sqrt{\delta}}}+1}}$ for Case 2. The space complexity is $\order{\mathrm{Poly}(n)+nL}$.
\end{theorem}

For Case 2, the time complexity remains $\mathrm{Poly}(n,L)$ once given $\delta$ and finite $\sfrac{\norm{O}_{\infty}}{\varepsilon}$. For detailed proof, see Suppl. Mat.~{VII and IX}.

To investigate how variable $\gamma$ affects time complexity with increasing depth $L$ in Case 1, we propose the following proposition.
\begin{prop}
\label{prop:lambda_and_L}
For Case 1, suppose Eq.~\eqref{eq:generate} is satisfied and $\norm{O}_\infty$ is fixed.
To estimate $\widetilde{\mathcal{L}}(\bm{\theta})$ with $\mathbb{E}_{\bm{\theta}}\abs{\widetilde{\mathcal{L}}-\widehat{\mathcal{L}}}^2$ less than a sufficiently small constant,
  we have 
  
  1. If $\gamma=\Omega(\frac{1}{\log{L}})$, there exists a classical algorithm that can complete the computation in time $\mathrm{Poly}\left(n,L\right)$. 
  
  2. If $\gamma=\order{\frac{1}{L}}$, there exists a situation where our method exhibits exponential time complexity with respect to $L$.
\end{prop}

\section{Numerical experiments}
\label{sec:Numerical}
In our theoretical analysis, we estimated the worst-case computational complexity, which is significantly higher than the complexity observed during actual computations. For further exploration, see Suppl. Mat.~{XI}.
\begin{figure}[htbp]
\includegraphics[width=0.5\textwidth]{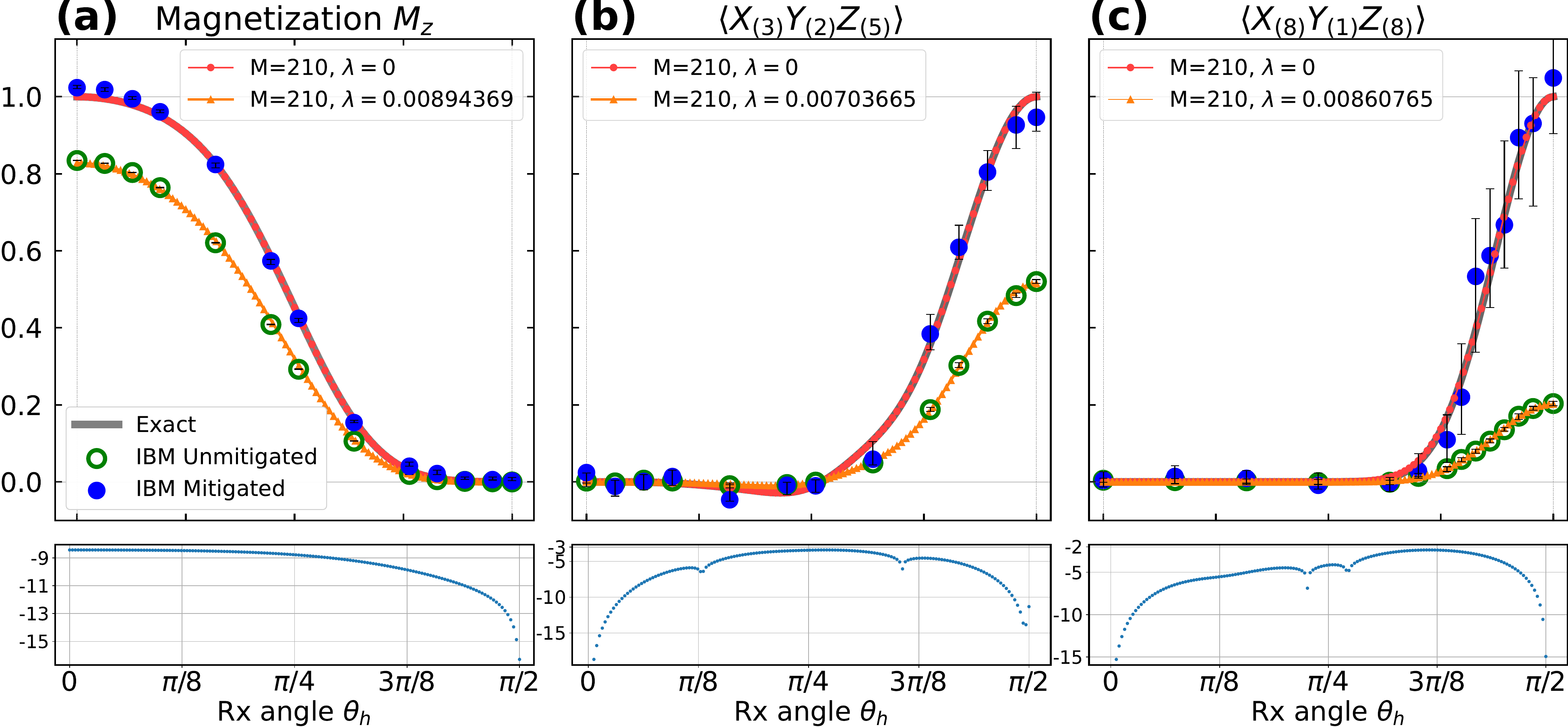}
\includegraphics[width=0.5\textwidth]{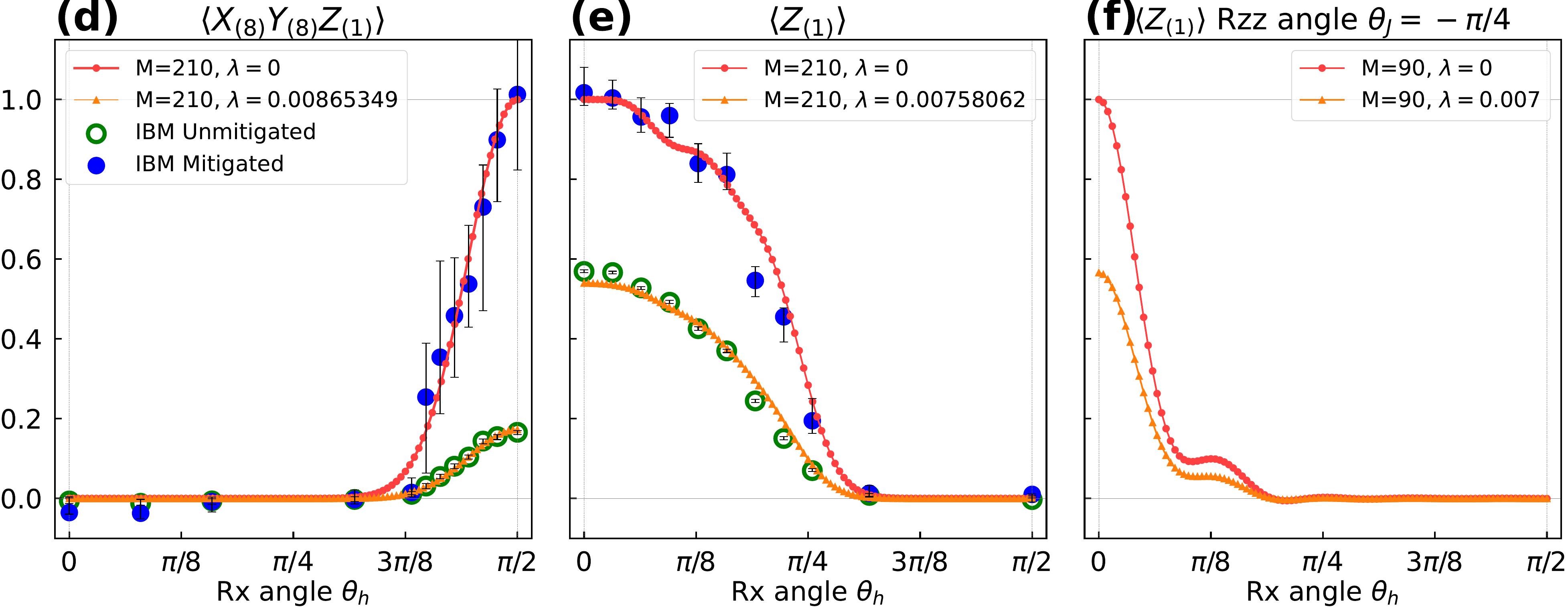}
\caption{
Classical simulation of different Pauli operators on IBM's 127-qubit Eagle processor. The green loops and blue circles in (a)-(e) correspond to direct experimental observations and error-mitigated results of different Pauli operators as reported in Ref.~\cite{kim2023evidence}. (a)-(c) simulate a circuit with $5$ Trotter steps, (d) simulates a circuit with $5$ Trotter steps and an additional layer of $R_X$ gates, and (e)-(f) simulate a circuit with $20$ Trotter steps. The rotation angle $\theta_J$ of the $R_{ZZ}$ gate in (a)-(e) is set to $\theta_J=-\pi/2$, while (f) is taken from Ref.~\cite{anand2023classical} with $\theta_J=-\pi/4$. The bottom subplots in (a)-(c) show the absolute difference between the simulated results of OBPPP and the exact results from Ref.~\cite{kim2023evidence}. In each figure, the red-dot line represents the output of OBPPP, which is an analytic expression of a trigonometric polynomial with respect to $\theta_h$. The orange-triangle line denotes the expectation from the suppressed path contributions under certain noise. The runtimes of (a)-(f) on two Xeon 6330 CPUs (28 cores per chip) are within 13 seconds, 146 seconds, 29 seconds, 137 seconds, 262 seconds, and 57 seconds, respectively.}\label{fig:IBM}
\end{figure}

In this section, we validate the efficiency of OBPPP method in practical applications by performing a classical simulation of IBM's 127-qubit Eagle processor~\cite{kim2023evidence}.

We utilize a depth-first search strategy on 56 cores of Xeon CPU to generate a list of valid paths. Notably, those Pauli path contributions are actually represented as analytical expressions of $\theta_h$, enabling us to obtain all corresponding results for $\theta_h$ across the entire continuous interval in a single computation.

Using Lemma~\ref{lemma:f_noisy}, we can directly compute the operator expectation values of the noisy circuits based on the analytical expressions of contributions.
This allows us to directly fit the raw experimental data before error mitigation.
To the best of our knowledge, this is currently the only classical algorithm capable of achieving this.

Figure~\ref{fig:IBM} presents six simulations: (a)-(e) from Ref.~\cite{kim2023evidence} and (f) from Ref.~\cite{anand2023classical}. For (a)-(c), where exact solutions exist, we found the OBPPP (M=210) outperforms quantum devices in precision and speed~\cite{beguvsic2023fast, kim2023evidence}. Without exact solutions for (d)-(e), OBPPP matched IBM's mitigated outcomes closely. In the unique scenario of (f), lacking both exact and experimental benchmarks, OBPPP simulations surpassed quantum chips in speed. For (a)-(e), we optimized noise rates ($p_x=p_y=p_z=\sfrac{\lambda}{4}$) using the least squares method,  so as to fit the expectation values of noisy circuits. A strong agreement was observed between OBPPP and IBM's unmitigated results, with average deviations below $0.002$ for (a)-(d) and below $0.008$ for (e). The optimal $\lambda$, ranging from $0.007$ to $0.009$, also in agreement with IBM's reported error rates.

In comparison with other recent classical simulation algorithms~\cite{tindall2023efficient,beguvsic2023fast}, our method possesses two main advantages: the ability to obtain an analytical expression for $\theta_h$ and the capability to infer the expected values of noisy circuit outcomes. Compared to the tensor network method~\cite{tindall2023efficient}, OBPPP demonstrates higher accuracy in (b) and (c). 
Throughout cases (a)-(e), OBPPP also exhibits faster execution times, especially in deeper circuits~\cite{tindall2023efficient}. The runtime of OBPPP is also much shorter than that of sparse Pauli dynamics (SPD)~\cite{beguvsic2023fast}. Besides, a notable difference is that OBPPP is less affected when computing $\theta_J=\sfrac{-\pi}{4}$ in (f). Moreover, OBPPP can deliver more accurate results than SPD.

In Suppl. Mat. XIII. we have also compared with an experiment on trapped-ion system~\cite{pagano2020quantum} and got highly consistent results.

\section{Conclusions and Discussions}
\label{sec:Discussion}
In this work, we have introduced OBPPP, a novel polynomial-scale method for approximating expectation values in VQAs under the single-qubit Pauli noise. This method is based on the truncated path integral on the Pauli basis. 

In theory, we have proven the method's time complexity $\mathrm{Poly}\left(n, L\right)$, and its space complexity $\order{\mathrm{Poly}(n)+nL}$.
We have also proven that in case 1 when $\gamma$ exceeds $\sfrac{1}{\log{L}}$, the computational complexity remains $\mathrm{Poly}\left(n,L\right)$.
These results highlight noise's vital role in shaping classical simulation feasibility.

Numerically, we have successfully performed classical simulations on IBM's Eagle processor in a shorter runtime than quantum hardware, while achieving more accurate expectations. Compared with other classical simulators, our method also perform faster runtime. By leveraging Lemma~\ref{lemma:f_noisy}, we obtained varying values of $\widetilde{\mathcal{L}}$ for different $\lambda$, leading to a strong fit with the unmitigated raw experimental data.

Our approach eliminates geometric constraints on quantum devices, allowing interactions with qubits in any position and facilitating multi-qubit rotation gates. Unlike previous methods with $\order{1}$ and $\Omega(\log{n})$ depth, our method imposes no assumptions on circuit depth or structure randomness (e.g. 2-design), or the prior distribution of the circuit output(e.g. anti-concentration).

Meanwhile, setting $p_x,p_y,p_z$ to $0$ in the analytic expressions derived from OBPPP introduces a novel idea of global~(rather than handling errors point-by-point) error mitigation~(for Pauli noise), offering new insight for efficiently fitting noiseless results.

 It is important to note that OBPPP relies on the following prerequisites: (1) Restricting quantum circuit gates to ${\mathrm{H}, \mathrm{S}, \mathrm{CNOT}}$ and single-parameter Pauli rotation gates. (2) Ensuring that the Clifford-transformed set ${\overline{\sigma}_{i,j}}$ generates all Pauli words $\{\mathbb{I},X,Y,Z\}^{\otimes n}$. (3) Sparsity of $\rho$ (computational basis) and $O$ (Pauli basis).

Our research focuses primarily on the effects of single-qubit Pauli noise, which includes the most common channels like depolarizing and dephasing noise. Similar approach also works for local unital noise. Yet, understanding the impact of other noise forms, such as general non-unital and correlated noises, remains challenging~~(more discussion on general noise models and related research is shown in Suppl. Mat.~{XI}).
Current findings lack rigorous proof about the noise's effect on model training performance. Moreover, numerical experiments present significant opportunities for optimization in existing algorithms. For example, as shown in Fig. Fig.~\ref{fig:IBM}(f), reaching $M=90$ is far from the limit of classical computers. We are merely limited by the memory of our current devices. This can be easily improved with better devices and algorithms.       In conclusion, numerous relevant studies require further exploration.
\begin{acknowledgments}
We thank Xun Gao, Fan Lu, Ningfeng Wang, Zhaohui Wei, Yusen Wu, Zishuo Zhao and Qin-Cheng Zheng for valuable discussions.
S.C was supported by the National Science Foundation of China (Grant No. 12004205). 
Z.L was supported by NKPs (Grant No. 2020YFA0713000).
Y.S, F.W and Z.L were supported by BMSTC and ACZSP (Grant No. Z221100002722017). 
S.C and Z.L were supported by Beijing Natural Science Foundation (Grant No. Z220002).

\end{acknowledgments}

\nocite{*}

\bibliography{aps_preprint}

\clearpage
\widetext

\appendix
\section*{Supplement Material}
\renewcommand{\thesection}{\Roman{section}}
\renewcommand{\appendixname}{Supplement Material}


\section{Notations}

In typical VQAs, the cost function is determined by the following expected value:
\begin{equation}\label{ap:eq:define_L}
\mathcal{L}(\bm{\theta})=\Tr{O\mathcal{U}(\bm{\theta})\rho \mathcal{U}^\dagger(\bm{\theta})},
\end{equation}
where $\rho$ is the density matrix of the $n$-qubit input state, and $O$ is the observable which is represented as a linear combination of Pauli operators, $\mathcal{U}(\bm{\theta})$ is a parameterized quantum circuit, which is composed with $L$ layers unitary transformation $\mathcal{U}_i(\bm{\theta}_i)$, $\mathcal{U}(\bm{\theta})=\mathcal{U}_L(\bm{\theta}_L)\mathcal{U}_{L-1}(\bm{\theta}_{L-1})\cdots \mathcal{U}_1(\bm{\theta}_1)$ and $\bm{\theta}=(\bm{\theta}_1,\cdots,\bm{\theta}_L)$. 

Notations $X$, $Y$, $Z$ represent the Pauli matrices $\begin{pmatrix}
  0 & 1\\
  1 & 0
\end{pmatrix},\;
\begin{pmatrix}
  0 & -i\\
  i & 0
\end{pmatrix}\; \mathrm{and}\;
\begin{pmatrix}
  1 & 0\\
  0 & -1
\end{pmatrix}$, respectively.

$\mathrm{H}(a),\mathrm{S}(a)$, and $ \mathrm{CNOT}(a,b)$ are defined as
\begin{equation}
  \begin{aligned}
\mathrm{H}(a)&=\mathbb{I}\otimes\cdots\otimes \mathbb{I}\otimes \mathrm{H_a}\otimes \mathbb{I}\otimes\cdots\otimes \mathbb{I}; \\
\mathrm{S}(a)&=\mathbb{I}\otimes\cdots\otimes \mathbb{I}\otimes \mathrm{S_a}\otimes \mathbb{I}\otimes\cdots\otimes \mathbb{I} ;\\
\mathrm{CNOT}(a,b)&=\mathbb{I}\otimes\cdots\otimes \mathbb{I}\otimes \mathrm{CNOT}_{a,b}\otimes \mathbb{I}\otimes\cdots\otimes \mathbb{I},
  \end{aligned}
\end{equation}
where $\mathrm{H}_a$ and $\mathrm{S}_a$ represent the Hadamard and phase gate acting on the $a$-th qubit. $\mathrm{CNOT}_{a,b}$ is given by $\ket{x}_a\ket{y}_b\rightarrow \ket{x}_a\ket{x\oplus y}_b$, where $\ket{\cdot}_a$ represent computational basis in the $a$-th qubit.

\section{Preparation of data}
\label{ap:pre_data}

\subsection{The input state $\rho$}

By assumption, the input state $\rho$ in the algorithm has sparsity:
\begin{equation}
\rho=\sum_{a,b} \rho_{a,b}\ketbra{a}{b},
\end{equation}
where $\ket{a}$ and $\ket{b}$ are computational basis states and there are $\mathrm{Poly}(n)$ size of non-zero $\rho_{a,b}$.

For each element $\rho_{a,b}\ketbra{a}{b}$ in $\rho$, $\Tr{s_0  (\rho_{a,b}\ketbra{a}{b})}$ can be calculated by
\begin{equation}\label{ap:eq:component_of_input}
\Tr{s_0  (\rho_{a,b}\ketbra{a}{b})}=\rho_{a,b} \bra{b}s_0 \ket{a}=\rho_{a,b} \prod_{j=1}^n \bra{b}_j (s_0)|_j \ket{a}_j,
\end{equation}
where $|_j$ is the limitation that limit the operator on $j$-th qubit and $\ket{\cdot}_j$ denotes $j$-th component of $\ket{\cdot}$.

By Eq.~\eqref{ap:eq:component_of_input}, $\Tr{s_0  (\rho_{a,b}\ketbra{a}{b})}$ can be calculated with time (space) complexity $\order{n}$. Then, by the sparsity assumption, $\Tr{s_0 \rho}$ can be calculated with time (space) complexity $\mathrm{Poly}(n)$.

\subsection{The observable $O$}

By assumption, the observable $O$ is a linear combination of Pauli words, and there are $\mathrm{Poly}(n)$ size of Pauli words in $O$ with non-zero coefficients.

We store $O$ in a tree data structure.
Each node in the tree is assigned a Pauli operator. 
The leaf nodes of the tree correspond to a unique Pauli word and store the corresponding coefficient value of $O$.
As an illustration, consider the observable $O=1 X_0+1 Z_1+0.5 X_0X_1$, which can be represented by the tree depicted as Fig.~\ref{fig:trie}.

\begin{figure}[htbp]
\begin{tikzpicture}
  [level distance=10mm,
   every node/.style={fill=red!60,circle,inner sep=1pt},
   level 1/.style={sibling distance=20mm,nodes={fill=red!45}},
   level 2/.style={sibling distance=10mm,nodes={fill=red!30}},
   level 3/.style={dashed,sibling distance=5mm,nodes={fill=red!0}}],
  \node[minimum size=0.01cm] { }
     child {node {$I_0$}
       child {node {$Z_1$}
       child {node {$1$}}
       }
     }
     child {node {$X_0$}
       child {node {$I_1$}
       child {node {$1$}}}
       child {node {$X_1$}
       child {node {$0.5$}}}
     };
\end{tikzpicture}
\caption{The tree representation of observable $O=1 X_0+1 Z_1+0.5 X_0X_1$.}\label{fig:trie}
\end{figure}

As a result, we can determine $\Tr{O s_L}$ utilizing the tree data structure with time complexity of $\order{n}$ and store the tree with space complexity of $\mathrm{Poly}(n)$.

\section{Pauli path integral}\label{ap:path_integral}
In the method, we used the Feynman path integral in the
Pauli basis to express the cost function $\mathcal{L}(\bm{\theta})=\Tr{O\mathcal{U}(\bm{\theta})\rho \mathcal{U}^\dagger(\bm{\theta})}$ as:
\begin{equation}\label{ap:eq:pauli_path_integral}
  \mathcal{L}(\bm{\theta})=\sum_{s_0,\cdots,s_L \in \bm{P}_n} f(\bm{\theta},s,O,\rho),
\end{equation}
where
\begin{equation}\label{ap:eq:f}
  f(\bm{\theta},s,O,\rho)=\Tr{Os_L}\left(\prod_{i=1}^{L}\Tr{s_i\mathcal{U}_i s_{i-1}\mathcal{U}_i^\dagger}\right)\Tr{s_0\rho}.
\end{equation}

First, to facilitate readers who are not familiar with tensor network diagrams~\cite{wood2011tensor}, below are some basic operation rules of tensor network diagrams.

A vector $\ket{v}$ in a tensor network is a first-order tensor, represented by a node with a line on the left. Similarly, a linear operator $A$ in a tensor network is a second-order tensor, represented by a node with a line on the left and right side. They can be represented as:
\begin{equation}
  \ket{v}
  =
  \begin{tikzpicture}[baseline=(current bounding box.center)]
    \node[draw] (s0) at (0,0) {$v$};
    \draw [thick] (-0.5,0)--(s0);
  \end{tikzpicture}
  \quad \mathrm{and} \quad
  A
  =
  \begin{tikzpicture}[baseline=(current bounding box.center)]
    \node[draw] (s0) at (0,0) {$A$};
    \draw [thick] (-0.5,0)--(s0) (s0)--(0.5,0);
  \end{tikzpicture}
\end{equation}

A transposition of a matrix $A$ can be expressed as a rotation of the corresponding tensor:
\begin{equation}
  \begin{tikzpicture}[baseline=(current bounding box.center)]
    \node[draw] (s0) at (0,0) {$A^T$};
    \draw [thick] (-0.5,0)--(s0) (s0)--(0.5,0);
  \end{tikzpicture}
  =
  \begin{tikzpicture}[baseline=(current bounding box.center)]
    \node[draw] (s0) at (0,0) {$A$};
    \draw[thick] (-0.5,0)--(s0) (s0)--(0.5,0);
    \draw[thick] (-0.5,0.5) arc(90:270:0.25);
    \draw[thick] (0.5,-0.5) arc(-90:90:0.25);
    \draw[thick] (-0.5,0.5)--(1,0.5) (0.5,-0.5)--(-1,-0.5);
  \end{tikzpicture}
  =
  \begin{tikzpicture}[baseline=(current bounding box.center)]
    \node[draw] (s0) at (0,0) {$A$};
    \draw[thick] (-0.5,0)--(s0) (s0)--(0.5,0);
    \draw[thick] (-0.5,0) arc(90:270:0.25);
    \draw[thick] (0.5,0) arc(-90:90:0.25);
    \draw[thick] (-0.5,-0.5)--(1,-0.5) (0.5,0.5)--(-1,0.5);
  \end{tikzpicture}
\end{equation}

The matrix multiplication of two matrices $A$ and $B$ can be expressed as connecting the corresponding wires of the tensors representing the matrices:
\begin{equation}
  \begin{tikzpicture}[baseline=(current bounding box.center)]
    \node[draw] (s1) at (0,0) {$AB$};
    \draw [thick] (-0.75,0)--(s1) (s1)--(0.75,0);
  \end{tikzpicture}
  =
  \begin{tikzpicture}[baseline=(current bounding box.center)]
    \node[draw] (s0) at (0,0) {$A$};
    \node[draw] (s1) at (0.75,0) {$B$};
    \draw [thick] (-0.5,0)--(s0) (s0)--(s1) (s1)--(1.25,0);
  \end{tikzpicture}
\end{equation}

The tensor product of two tensors $A$ and $B$ can be expressed as vertical juxtaposition:
\begin{equation}
  \begin{tikzpicture}[baseline=(current bounding box.center)]
    \node[draw,minimum height=1.25cm] (s1) at (0,0) {$A\otimes B$};
    \draw [thick] (-0.75,0.375)-|(s1.west) (-0.75,-0.375)-|(s1.west) (s1.east)|-(0.75,0.375) (s1.east)|-(0.75,-0.375);
  \end{tikzpicture}
  =
  \begin{tikzpicture}[baseline=(current bounding box.center)]
    \node[draw] (s0) at (0,0) {$A$};
    \node[draw] (s1) at (0,-0.75) {$B$};
    \draw [thick] (-0.5,0)--(s0)--(0.5,0)  (-0.5,-0.75)--(s1)--(0.5,-0.75);
  \end{tikzpicture}
\end{equation}

The trace of an operator $A$ is depicted by connecting the corresponding left and right wires of a linear operator:
\begin{equation}
  \Tr{A}
  =
  \begin{tikzpicture}[baseline=(current bounding box.center)]
    \node[draw] (s0) at (0,0) {$A$};
    \draw [thick] (-0.5,0)--(s0) (s0)--(0.5,0);
    \draw[thick] (-0.5,0.5) arc(90:270:0.25);
    \draw[thick] (0.5,0) arc(-90:90:0.25);
    \draw [thick] (-0.5,0.5)--(0.5,0.5);
  \end{tikzpicture}
\end{equation}

To verify the validity of Eq.~\eqref{ap:eq:pauli_path_integral}, we first express $f(\bm{\theta},s,O,\rho)$ as tensor network diagrams:
\begin{equation}
  \begin{aligned}
  f(\bm{\theta},s,O,\rho)=&\Tr{Os_L}\left(\prod_{i=1}^{L}\Tr{s_i\mathcal{U}_i s_{i-1}\mathcal{U}_i^\dagger}\right)\Tr{s_0\rho}\\
  =&\;
      \begin{tikzpicture}[baseline=(current bounding box.center)]
        \coordinate(l)at(-0.25,0.5){};\coordinate(r)at(1,0.5){};
        \node[rectangle,draw] (H) at (0,0) {$O$};
        \node[draw,shape=circle,inner sep=1pt] (sL) at (0.75,0) {$s_L$};
        \draw [thick] (H) -- (sL) (l) -- (r);
        \draw[thick] (-0.25,0.5) arc(90:270:0.25);
        \draw[thick] (1,0) arc(-90:90:0.25);
        \end{tikzpicture}
    \; \prod_{i=1}^{L}
      \begin{tikzpicture}[baseline=(current bounding box.center)]
        \coordinate(l)at(-0.25,0.5){};\coordinate(r)at(2.55,0.5){};
        \node[draw,shape=circle,inner sep=1pt] (sj) at (0.,0) {$s_i$};
        \node[rectangle,draw] (Uj) at (0.75,0) {$\mathcal{U}_i$};
        \node[draw,shape=circle,inner sep=-1pt] (sj1) at (1.5,0) {\scriptsize $s_{i-1}$};
        \node[rectangle,draw] (Ujt) at (2.25,0) {$\mathcal{U}_i^\dagger$};
        \draw [thick] (sj) -- (Uj) --(sj1)--(Ujt) (l) -- (r);
        \draw[thick] (-0.25,0.5) arc(90:270:0.25);
        \draw[thick] (2.55,0) arc(-90:90:0.25);
      \end{tikzpicture}
    \;
      \begin{tikzpicture}[baseline=(current bounding box.center)]
        \coordinate(l)at(-0.25,0.5){};\coordinate(r)at(1,0.5){};
        \node[rectangle,draw] (rho) at (0.75,0) {$\rho$};
        \node[draw,shape=circle,inner sep=1pt] (s0) at (0.,0) {$s_0$};
        \draw [thick] (rho) -- (s0) (l) -- (r);
        \draw[thick] (-0.25,0.5) arc(90:270:0.25);
        \draw[thick] (1,0) arc(-90:90:0.25);
      \end{tikzpicture}\\
      =&\;
      \begin{tikzpicture}[baseline=(current bounding box.center)]
        \coordinate(l)at(-0.25,0.5){};\coordinate(r)at(1,0.5){};
        \node[rectangle,draw] (H) at (0,0) {$O$};
        \node[draw,shape=circle,inner sep=1pt] (sL) at (0.75,0) {$s_L$};
        \draw [thick] (H) -- (sL) (l) -- (r);
        \draw[thick] (-0.25,0.5) arc(90:270:0.25);
        \draw[thick] (1,0) arc(-90:90:0.25);
        \end{tikzpicture}
    \; \prod_{i=1}^{L}
      \begin{tikzpicture}[baseline=(current bounding box.center)]
        \coordinate(l)at(-0.25,0.5){};\coordinate(r)at(1.8,0.5){};
        \node[draw,shape=circle,inner sep=1pt] (sj) at (0.,0) {$s_i$};
        \node[rectangle,draw] (Uj) at (0.75,0) {$\mathcal{U}_i$};
        \node[draw,shape=circle,inner sep=-1pt] (sj1) at (1.5,0) {\scriptsize $s_{i-1}$};
        \node[rectangle,draw] (Ujt) at (0.75,0.75) {$\overline{\mathcal{U}}_i$};
        \draw [thick] (sj) -- (Uj) --(sj1) (l) --(Ujt)-- (r);
        \draw[thick] (-0.25,0.5) arc(90:270:0.25);
        \draw[thick] (1.8,0) arc(-90:90:0.25);
      \end{tikzpicture}
    \;
      \begin{tikzpicture}[baseline=(current bounding box.center)]
        \coordinate(l)at(-0.25,0.5){};\coordinate(r)at(1,0.5){};
        \node[rectangle,draw] (rho) at (0.75,0) {$\rho$};
        \node[draw,shape=circle,inner sep=1pt] (s0) at (0.,0) {$s_0$};
        \draw [thick] (rho) -- (s0) (l) -- (r);
        \draw[thick] (-0.25,0.5) arc(90:270:0.25);
        \draw[thick] (1,0) arc(-90:90:0.25);
      \end{tikzpicture} \quad .
  \end{aligned}
  \end{equation}

Using this form, the right side of Eq.~\eqref{ap:eq:pauli_path_integral} can be expressed as
  \begin{equation}
    \begin{aligned}
      \sum_{s_0,\cdots,s_L \in \bm{P}_n} f(\bm{\theta},s,O,\rho)=&
      \sum_{s_0,\cdots,s_L \in \bm{P}_n}
      \begin{tikzpicture}[baseline=(current bounding box.center)]
        \coordinate(l)at(-0.25,0.5){};\coordinate(r)at(1,0.5){};
        \node[rectangle,draw] (H) at (0,0) {$O$};
        \node[draw,shape=circle,inner sep=1pt] (sL) at (0.75,0) {$s_L$};
        \draw [thick] (H) -- (sL) (l) -- (r);
        \draw[thick] (-0.25,0.5) arc(90:270:0.25);
        \draw[thick] (1,0) arc(-90:90:0.25);
        \end{tikzpicture}\;
      \begin{tikzpicture}[baseline=(current bounding box.center)]
        \coordinate(l)at(-0.25,0.5){};\coordinate(r)at(1.8,0.5){};
        \node[draw,shape=circle,inner sep=1pt] (sj) at (0.,0) {$s_L$};
        \node[rectangle,draw] (Uj) at (0.75,0) {$\mathcal{U}_L$};
        \node[draw,shape=circle,inner sep=-1pt] (sj1) at (1.5,0) {\scriptsize $s_{L-1}$};
        \node[rectangle,draw] (Ujt) at (0.75,0.75) {$\overline{\mathcal{U}}_L$};
        \draw [thick] (sj) -- (Uj) --(sj1) (l) --(Ujt)-- (r);
        \draw[thick] (-0.25,0.5) arc(90:270:0.25);
        \draw[thick] (1.8,0) arc(-90:90:0.25);
      \end{tikzpicture}
    \cdots
      \begin{tikzpicture}[baseline=(current bounding box.center)]
        \coordinate(l)at(-0.25,0.5){};\coordinate(r)at(1.8,0.5){};
        \node[draw,shape=circle,inner sep=1pt] (sj) at (0.,0) {$s_1$};
        \node[rectangle,draw] (Uj) at (0.75,0) {$\mathcal{U}_1$};
        \node[draw,shape=circle,inner sep=1.5pt] (sj1) at (1.5,0) {$s_{0}$};
        \node[rectangle,draw] (Ujt) at (0.75,0.75) {$\overline{\mathcal{U}}_1$};
        \draw [thick] (sj) -- (Uj) --(sj1) (l) --(Ujt)-- (r);
        \draw[thick] (-0.25,0.5) arc(90:270:0.25);
        \draw[thick] (1.8,0) arc(-90:90:0.25);
      \end{tikzpicture}\;
      \begin{tikzpicture}[baseline=(current bounding box.center)]
        \coordinate(l)at(-0.25,0.5){};\coordinate(r)at(1,0.5){};
        \node[rectangle,draw] (rho) at (0.75,0) {$\rho$};
        \node[draw,shape=circle,inner sep=1pt] (s0) at (0.,0) {$s_0$};
        \draw [thick] (rho) -- (s0) (l) -- (r);
        \draw[thick] (-0.25,0.5) arc(90:270:0.25);
        \draw[thick] (1,0) arc(-90:90:0.25);
      \end{tikzpicture}\\
      =&\;
      \begin{tikzpicture}[baseline=(current bounding box.center)]
        \coordinate(l)at(-0.25,0.75){};\coordinate(r)at(3.9,0.75){};
        \node[rectangle,draw] (H) at (0,0) {$O$};
        \node[rectangle,draw] (Ul) at (0.75,0) {$\mathcal{U}_L$};
        \node[rectangle,draw] (Ujt) at (0.75,0.75) {$\overline{\mathcal{U}}_L$};
        \node[rectangle,draw] (U1) at (3,0) {$\mathcal{U}_1$};
        \node[rectangle,draw] (U1t) at (3,0.75) {$\overline{\mathcal{U}}_1$};
        \node[] (cdots_down) at (2,0) {$\cdots$};
        \node[] (cdots_up) at (2,0.75) {$\cdots$};
        \node[rectangle,draw] (rho) at (3.65,0) {$\rho$};
        \draw [thick] (H) -- (Ul)-- (cdots_down)--(U1)--(rho) (l) -- (Ujt) -- (cdots_up)--(U1t)--(r);
        \draw[thick] (-0.25,0.75) arc(90:270:0.75/2);
        \draw[thick] (3.9,0) arc(-90:90:0.75/2);
        \end{tikzpicture}\\
        =&\;
        \begin{tikzpicture}[baseline=(current bounding box.center)]
          \coordinate(l)at(-0.25,0.5){};\coordinate(r)at(2.55,0.5){};
          \node[rectangle,draw] (H) at (0,0) {$O$};
          \node[rectangle,draw] (U) at (0.75,0) {$\mathcal{U}$};
          \node[rectangle,draw] (rho) at (1.5,0) {$\rho$};
          \node[rectangle,draw] (Ud) at (2.25,0) {$\mathcal{U}^\dagger$};
          \draw [thick] (H)--(U)--(rho)--(Ud) (l) -- (r);
          \draw[thick] (-0.25,0.5) arc(90:270:0.25);
          \draw[thick] (2.55,0) arc(-90:90:0.25);
        \end{tikzpicture}\\
        =&\Tr{O\mathcal{U}(\bm{\theta})\rho \mathcal{U}^\dagger(\bm{\theta})}\\
        =&\mathcal{L}(\bm{\theta}),
    \end{aligned}
  \end{equation}
where second equality is obtained by the property of normalized orthonormal basis $\{\bm{P}_n\}$, allowing an operator
to be expressed as $T=\sum_{s_i \in \bm{P}_n} Tr(T s_i) s_i$:
\begin{equation}
  \sum_{s\in\bm{P}_n}
  \begin{tikzpicture}[baseline=(current bounding box.center)]
    \node[draw,shape=circle,inner sep=1pt] (s1) at (1.75,0) {$s$};
    \node[draw,shape=circle,inner sep=1pt] (s0) at (0.5,0) {$s$};
    \node[rectangle,draw,inner sep=2pt] (O) at (0,0) {$T$};
    \draw[thick] (-0.25,0.5) arc(90:270:0.25);
    \draw [thick] (O)--(s0) (-0.25,0.5)--(0.75,0.5);
    \draw [thick] (2.25,0)--(s1) (1.5,0.5)--(2.25,0.5);
    \draw[thick] (1.5,0.5) arc(90:270:0.25);
    \draw[thick] (0.75,0) arc(-90:90:0.25);
  \end{tikzpicture}
  =
  \begin{tikzpicture}[baseline=(current bounding box.center)]
    \draw [thick] (O)--(1,0) (-0.25,0.5)--(1,0.5);
    \node[rectangle,draw,inner sep=2pt] (O) at (0,0) {$T$};
    \draw[thick] (-0.25,0.5) arc(90:270:0.25);
    \node[] at (1.2,0){$.$};
  \end{tikzpicture}
\end{equation}

In the presence of single-qubit Pauli noise, the noisy cost function $\widehat{\mathcal{L}}$ can be expressed as: 
\begin{equation}\label{ap:eq:tn:cost_function_noisy}
  \begin{aligned}
    \widehat{\mathcal{L}}(\bm{\theta})&=\Tr{ O \mathcal{N}^{\otimes n}(\mathcal{U}_L \mathcal{N}^{\otimes n}(\cdots\mathcal{U}_1\mathcal{N}^{\otimes n}(\rho) \mathcal{U}_1^\dagger\cdots) \mathcal{U}_L^\dagger)}\\
    &=\;\begin{tikzpicture}[baseline=(current bounding box.center)]
      \coordinate(l)at(-0.25,0.75){};\coordinate(r)at(8.25,0.75){};
      \node[rectangle,draw] (H) at (0,0) {$O$};
      \node[rectangle,draw,minimum height=40] (Nl) at (1,0.4) {$\mathcal{N}^{\otimes n}$};
      \node[rectangle,draw] (Ul) at (2,0) {$\mathcal{U}_L$};
      \node[rectangle,draw] (Ujt) at (2,0.75) {$\overline{\mathcal{U}}_L$};
      \node[rectangle,draw,minimum height=40] (Nl1) at (3,0.4) {$\mathcal{N}^{\otimes n}$};
      \node[rectangle,draw,minimum height=40] (N1) at (5,0.4) {$\mathcal{N}^{\otimes n}$};
      \node[rectangle,draw,minimum height=40] (N0) at (7,0.4) {$\mathcal{N}^{\otimes n}$};
      \node[rectangle,draw] (U1) at (6,0) {$\mathcal{U}_1$};
      \node[rectangle,draw] (U1t) at (6,0.75) {$\overline{\mathcal{U}}_1$};
      \node[] (cdots_down) at (4,0) {$\cdots$};
      \node[] (cdots_up) at (4,0.75) {$\cdots$};
      \node[rectangle,draw] (rho) at (8,0) {$\rho$};
      \draw [thick] (H)--(H-| Nl.west) (Ul-|Nl.east)--(Ul)--(Ul-| Nl1.west) (cdots_down-|Nl1.east)--(cdots_down)--(cdots_down-| N1.west) (U1-|N1.east)--(U1)--(U1-| N0.west) (rho-|N0.east)--(rho) (l)--(l-| Nl.west) (Ujt-| Nl.east)--(Ujt)--(Ujt-| Nl1.west) (cdots_up-|Nl1.east)-- (cdots_up)--(cdots_up-| N1.west) (U1t-|N1.east)--(U1t)--(U1t-| N0.west) (r-|N0.east)--(r);
      \draw[thick] (-0.25,0.75) arc(90:270:0.75/2);
      \draw[thick] (8.25,0) arc(-90:90:0.75/2);
      \end{tikzpicture}\\
      &=\sum_{s_0,\cdots,s_L \in \bm{P}_n}
      \begin{tikzpicture}[baseline=(current bounding box.center)]
        \coordinate(l)at(-0.25,0.5){};\coordinate(r)at(2.25,0.5){};
        \node[rectangle,draw] (H) at (0,0) {$O$};
        \node[rectangle,draw,minimum height=40] (N) at (1,0.4) {$\mathcal{N}^{\otimes n}$};
        \node[draw,shape=circle,inner sep=1pt] (sL) at (2,0) {$s_L$};
        \draw [thick] (H)--(H-| N.west) (sL-|N.east)--(sL) (l)--(l-| N.west) (r-|N.east)--(r);
        \draw[thick] (-0.25,0.5) arc(90:270:0.25);
        \draw[thick] (2.25,0) arc(-90:90:0.25);
        \end{tikzpicture}\;
      \begin{tikzpicture}[baseline=(current bounding box.center)]
        \coordinate(l)at(-0.25,0.5){};\coordinate(r)at(3.05,0.5){};
        \node[draw,shape=circle,inner sep=1pt] (sj) at (0.,0) {$s_L$};
        \node[rectangle,draw] (Uj) at (0.75,0) {$\mathcal{U}_L$};
        \node[draw,shape=circle,inner sep=-1pt] (sj1) at (2.75,0) {\scriptsize $s_{L-1}$};
        \node[rectangle,draw] (Ujt) at (0.75,0.75) {$\overline{\mathcal{U}}_L$};
        \node[rectangle,draw,minimum height=40] (N) at (1.75,0.4) {$\mathcal{N}^{\otimes n}$};
        \draw [thick] (sj)--(Uj)--(Uj-| N.west) (sj1-|N.east)--(sj1) (l) --(l-|Ujt.west) (Ujt)--(Ujt-| N.west) (r-|N.east)-- (r);
        \draw[thick] (-0.25,0.5) arc(90:270:0.25);
        \draw[thick] (3.05,0) arc(-90:90:0.25);
      \end{tikzpicture}
    \cdots
      \begin{tikzpicture}[baseline=(current bounding box.center)]
        \coordinate(l)at(-0.25,0.5){};\coordinate(r)at(3,0.5){};
        \node[draw,shape=circle,inner sep=1pt] (sj) at (0.,0) {$s_1$};
        \node[rectangle,draw] (Uj) at (0.75,0) {$\mathcal{U}_1$};
        \node[draw,shape=circle,inner sep=1.5pt] (sj1) at (2.75,0) {$s_{0}$};
        \node[rectangle,draw] (Ujt) at (0.75,0.75) {$\overline{\mathcal{U}}_1$};
        \node[rectangle,draw,minimum height=40] (N) at (1.75,0.4) {$\mathcal{N}^{\otimes n}$};
        \draw [thick] (sj)--(Uj)--(Uj-| N.west) (sj1-|N.east)--(sj1) (l) --(l-|Ujt.west) (Ujt)--(Ujt-| N.west) (r-|N.east)-- (r);
        \draw[thick] (-0.25,0.5) arc(90:270:0.25);
        \draw[thick] (3,0) arc(-90:90:0.25);
      \end{tikzpicture}\;
      \begin{tikzpicture}[baseline=(current bounding box.center)]
        \coordinate(l)at(-0.25,0.5){};\coordinate(r)at(1.25,0.5){};
        \node[rectangle,draw] (rho) at (1,0) {$\rho$};
        \node[rectangle,draw,minimum height=40,opacity=0] (N) at (1,0.4) {$\mathcal{N}^{\otimes n}$};
        \node[draw,shape=circle,inner sep=1pt] (s0) at (0.,0) {$s_0$};
        \draw [thick] (rho)--(s0) (l)--(r);
        \draw[thick] (-0.25,0.5) arc(90:270:0.25);
        \draw[thick] (1.25,0) arc(-90:90:0.25);
      \end{tikzpicture}\\
      &=\Tr{O\mathcal{N}^{\otimes n}(s_L)}\Tr{s_L\mathcal{U}_L \mathcal{N}^{\otimes n}(s_{L-1})\mathcal{U}_L^\dagger}\cdots\Tr{s_1\mathcal{U}_1 \mathcal{N}^{\otimes n}(s_{0})\mathcal{U}_1^\dagger}\Tr{s_0\rho}.
  \end{aligned}
\end{equation}

Thus, we can express the noisy cost function $\widehat{\mathcal{L}}$ as the sum of the contributions of all Pauli paths, given by
\begin{equation}\label{ap:eq:Npauli_path_integral}
\mathcal{\hat{L}}(\bm{\theta})=\sum_{s\in \bm{P}^{L+1}_n} \hat{f}(\bm{\theta},s,O,\rho),
\end{equation}
where
\begin{equation}
  \hat{f}(\bm{\theta},s,O,\rho)=\Tr{O\mathcal{N}^{\otimes n}(s_L)}
  \left(\prod_{i=1}^{L}\Tr{s_i\mathcal{U}_i \mathcal{N}^{\otimes n}(s_{i-1})\mathcal{U}_i^\dagger}\right)\Tr{s_0\rho}.
\end{equation}

\section{Algorithm}\label{ap:algorithm}

The scenario is categorized based on the number of non-zero factors in $\{p_x,p_y,p_z\}$ of Pauli noise:
\begin{itemize}
  \item \textbf{Case 1}: When there are at least two non-zero elements in $\{p_x,p_y,p_z\}$, the truncation weight is $\abs{s}=\sum_{P\in \{X,Y,Z\}}\abs{s}_P$, commonly referred to as the Hamming weight. An example of this is the depolarizing channel.
  \item \textbf{Case 2}: When there is only one element in $\{p_x,p_y,p_z\}$ non-zero~(without loss of generality $p_x\neq 0$), truncation weight is $\abs{s}=\sum_{P\in \{Y,Z\}}\abs{s}_P$. An illustration of this case would be the dephasing channel.
\end{itemize}

Our algorithm calculates the contributions of the Pauli paths with $\abs{s}\leq M$ to provide an approximation of noisy cost function $\widehat{\mathcal{L}}$. 
The approximate noisy cost function can be expressed as:
\begin{equation}
\widetilde{\mathcal{L}}(\bm{\theta})=\sum_{\abs{s}\leq M} \hat{f}(\bm{\theta},s,O,\rho)=\sum_{\abs{s}\leq M} (1-2(p_y+p_z))^{\abs{s}_X}(1-2(p_x+p_z))^{\abs{s}_Y}(1-2(p_x+p_y))^{\abs{s}_Z}f(\bm{\theta},s,O,\rho).
\end{equation}

Generally, it is a significant challenge to evaluate all the Pauli paths with  weight less than $M$.
But owing to the remark of Proposition~\ref{prop:f_ele}, most Pauli paths have zero contribution in the path integrals.
We introduce the following method for enumerating all Pauli paths with non-zero contributions and $\abs{s}\leq M$.

The key idea of this method is based on the sparsity of $O$ and the observation in Proposition~\ref{prop:f_ele} that for any Pauli path $s$ with $f(\bm{\theta},s,O,\rho)\neq 0$, if one of $s_{i-1}|_\g{\sigma_{i,j}}$ or $s_{i}|_\g{\sigma_{i,j}}$ is fixed, then the other one has at most two cases, which holds for all $i,j$.
Likewise, if $s_{i-1}|_\g{V_{i,k}}$ (or $s_{i-1}|_{I_{i}}$) or $s_{i}|_\g{V_{i,k}}$ (or $s_{i}|_{I_{i}}$) is fixed, the other one has only one case, which holds for all $i,k$.

\subsection{Analysis of the algorithm for \textbf{Case 1}}
In \textbf{Case 1}, for any Pauli path $s=(s_0,\cdots,s_L)\in \bm{P}^{L+1}_n$, the truncation weight is defined as $\abs{s}=\sum_{P\in \{X,Y,Z\}}\abs{s}_P$, commonly referred to as the Hamming weight $\abs{s}_H=\abs{s_L}_H+\cdots+\abs{s_L}_H$, where $\abs{s_i}_H$ is the count of non-identify operate in Pauli word $s_i$.

Moreover, for Pauli path $s$ that has a non-zero contribution, $\abs{s_i}_H > 0$ for $i=0,\cdots,L$ is required, otherwise at some layer $s_i$ will be trivial, which will lead to $\Tr{s_{i+1} \mathcal{U}_i s_i \mathcal{U}_i^\dagger}=\Tr{s_{i+1} s_i}$. To avoid zero contribution, there must have $s_L=\cdots=s_{i+1}=s_i=\left(\frac{\mathbb{I}}{\sqrt{2}}\right)^{\otimes n}$.
Without loss of generality, we can set $\Tr{O}=0$ (or replace $O$ by $O-\Tr{O} I$), which leads to $\Tr{O s_L}=0$. Thus, for Pauli path $s$ with Hamming wight $\abs{s}_H\leq M$ and non-zero contribution, there must be $\abs{s_{L}}_H+\cdots+\abs{s_{L-i}}_H\leq M-(L-i)$.

The back-propagation process for searching Pauli paths is as follows:
\begin{enumerate}
\item We begin by selecting $s_L$. In order to ensure that $\Tr{Os_L}\neq0$, $s_L$ can only be selected from Pauli words in $O$. Owing to the assumption that observable $O$ is a linear combination of Pauli words in Polynomial size of $n$, there are at most $\mathrm{Poly}(n)$ cases for $s_L$.
The time and space complexity of enumerating $s_L$ are $\order{\mathrm{Poly}(n)}$.

\item For each case of $s_L$, the next step is to explore all potential $s_{L-1}$. 
There are at most $\abs{s_L}_H$ non-identity elements in $\{s_{L}|_\g{\sigma_{L,j}}\} \cup \{s_{L}|_\g{V_{L,k}}\} \cup \{s_{L}|_{I_{L}}\}$, corresponding to at most $\abs{s_L}_H$ non-identity elements in $\{s_{L-1}|_\g{\sigma_{L,j}}\} \cup \{s_{L-1}|_\g{V_{L,k}}\} \cup \{s_{L-1}|_{I_{L}}\}$.
Furthermore, each element has at most two potential candidates, resulting in at most $2^\abs{s_L}_H$ cases for $s_{L-1}$.
In addition, we need to eliminate the cases in which $\abs{s_{L}}_H+\abs{s_{L-1}}_H> M-(L-1)$. The time complexity of enumerating $s_{L-1}$ for a given $s_L$ is $\order{n2^{\abs{s_L}_H}}$ and the space complexity is $\order{n}$.

\item Similarly, repeat step $(2)$ to enumerate $s_{L-2}$ for each case of $s_{L-1}$ and eliminate candidates with $\abs{s_{L}}_H+\abs{s_{L-1}}_H+\abs{s_{L-2}}_H> M-(L-2)$.
We obtain up to $2^{\abs{s_{L-1}}_H}$ cases for $s_{L-2}$, with time complexity $\order{n2^{\abs{s_{L-1}}_H}}$ for a given $s_{L-1}$ and space complexity $\order{n}$.
Repeating this process, we can enumerate all $s_{L-2},\cdots,s_{0}$.
\end{enumerate}

From the above discussion, given any $s_L$, the number of different Pauli paths output is at most $2^{\abs{s_1}_H+\cdots+\abs{s_{L}}_H}\leq 2^{M}$. Alternatively, for a given $s_L$, we consider all Pauli paths $s$ with $s_L$ as the end element, $\abs{s}_H\leq M$ and $f(\bm{\theta},s,O,\rho)\neq 0$.
It can be considered as a tree starting from $s_L$ and the new branching will only occur when $s_{i}|_\g{\sigma_{i,j}}$ is not identity and $\sigma_{i,j} \in AC(i,s_{i})$. While the number of non-identity elements in $\{s_{i}|_\g{\sigma_{i,j}}\}$ is at most $M$, the number of possible cases is at most $2^{M}$.
Thus, to compute all contributions of the Pauli path with $\abs{s}_H\leq M$, we need to calculate at most $\mathrm{Poly}(n) 2^{M}$ different Pauli paths.

In step $(1)$ the time cost is $\order{\mathrm{Poly}(n)}$. In step $(2)$, considering all cases of $s_L$, the time cost is $\sum_{s_L}n2^{\abs{s_L}_H}$. In step $(3)$, considering all case of $s_{L-1}$, the time cost is $\sum_{s_L}\sum_{s_{L-1}(s_L)}n2^{\abs{s_{L-1}(s_L)}_H}$, where $s_{L-1}(s_L)$ denotes the output of step $(2)$ corresponding to a given $s_L$. 
Similar results hold for $s_{L-2},\cdots,s_{0}$.

Thus, the time complexity of the above process is
\begin{equation}
  \begin{aligned}
    \order{\mathrm{Poly}(n)+\sum_{s_L}n2^{\abs{s_L}_H}+\sum_{s_L}\sum_{s_{L-1}(s_L)}n2^{\abs{s_{L-1}(s_L)}_H} +\sum_{s_L}\sum_{s_{L-1}(s_L)}\sum_{s_{L-2}(s_L,s_{L-1})}n2^{\abs{s_{L-2}(s_L,s_{L-1})}_H}+  \cdots} \\
  \leq \order{\mathrm{Poly}(n)n L 2^{M}}=\mathrm{Poly}(n) \order{L} 2^{M}.
  \end{aligned}
\end{equation}
Here the inequality holds, because
\begin{equation}
  \begin{aligned}
    \sum_{s_L}\cdots\sum_{s_{i-1}(s_L,\cdots,s_i)}n2^{\abs{s_{i-1}(s_L,\cdots,s_i)}_H}
    &=\sum_{s_L}\cdots\sum_{s_{i}(s_L,\cdots,s_{i+1})}\sum_{s_{i-1}(s_L,\cdots,s_i)}n2^{\abs{s_{i-1}(s_L,\cdots,s_i)}_H}\\
    (\mathrm{By} \; \abs{s_{i-1}}_H+\cdots +\abs{s_L}_H\leq M)\qquad
    &\leq  \sum_{s_L}\cdots\sum_{s_{i}(s_L,\cdots,s_{i+1})}\sum_{s_{i-1}(s_L,\cdots,s_i)}n2^{M-(\abs{s_{i}(s_L,\cdots,s_{i+1})}_H+\cdots+\abs{s_L}_H)}\\
    (\mathrm{By} \; \#s_{i-1}\leq 2^{\abs{s_i}_H})\qquad\qquad\quad
    &\leq \sum_{s_L}\cdots\sum_{s_{i}(s_L,\cdots,s_{i+1})}n2^{M+\abs{s_{i}(s_L,\cdots,s_{i+1})}-\left(\abs{s_{i}(s_L,\cdots,s_{i+1})}_H+\cdots+\abs{s_L}_H\right)}\\
    &=\sum_{s_L}\cdots\sum_{s_{i}(s_L,\cdots,s_{i+1})}n2^{M-(\abs{s_{i+1}(s_L,\cdots,s_{i+2})}_H+\cdots+\abs{s_L}_H)}\\
    &\qquad\vdots\\
    &\leq \sum_{s_L}\sum_{s_{L-1}(s_L)} n2^{M-\abs{s_L}_H}\\
    &\leq \sum_{s_L} n2^M=\mathrm{Poly}(n)n2^M.
  \end{aligned}
\end{equation}

The space complexity of the above process is
\begin{equation}\label{ap:eq:space_cost}
\order{\mathrm{Poly}(n)}+\sum_{i=1}^{L}\left(\order{n}\right)\leq\order{\mathrm{Poly}(n)+nL}.
\end{equation}

After obtaining candidates of Pauli path $s=(s_0,\cdots,s_L)$, the next step is computing its contribution $\hat{f}(\bm{\theta},s,O,\rho)$. For each Pauli path $s$, it is possible to determine $f(\bm{\theta},s,O,\rho)$ with time complexity $\order{nL}+\mathrm{Poly}(n)$ using Eq.~\eqref{ap:eq:f} and Proposition~\ref{prop:f_ele}. 
Thus, the overall time cost for computing $\widetilde{\mathcal{L}}$ in \textbf{Case 1} is about 

\begin{equation}\label{ap:eq:T}
  \left( \order{nL} + \mathrm{Poly}(n) \right)  \mathrm{Poly}(n) 2^{M}+\mathrm{Poly}(n) \order{L} 2^{M} = \mathrm{Poly}(n) \order{L} 2^{M}.
\end{equation}

The process of our algorithm is summarized in Algorithm~\ref{ALGORITHM_HS}.

\begin{algorithm}[H]
\caption{OBPPP algorithm for estimating the truncated cost function in Case 1}\label{ALGORITHM_HS}
\begin{algorithmic}
  \State Set $\widetilde{\mathcal{L}}=0$.
  \State Enumerate $s_L$ as all Pauli words with non-zero coefficient in $O$
  \For{candidates of $s_{L}$}
  \State{According to the $L$-th layer of the circuit, generate candidates of $s_{L-1}$.}
  \State{Eliminate cases with $\abs{s_{L}}_H+\abs{s_{L-1}}_H > M-(L-1)$.}
    \For{candidates of $s_{L-1}$}
      \State{According to the $(L-1)$-th layer of the circuit, generate candidates of $s_{L-2}$.}
      \State{Eliminate cases with $\abs{s_{L}}_H+\abs{s_{L-1}}_H+\abs{s_{L-2}}_H> M-(L-2)$.}
      \State{\vdots}
      \For{candidates of $s_1$}
      \State{According to the $1$-th layer of the circuit, generate candidates of $s_{0}$.}
      \State{Eliminate cases with $\abs{s_{L}}_H+\cdots+\abs{s_{0}}_H> M$.}
      \For{candidates of $s_{0}$}
      \State{Set Pauli path $s=(s_0,\cdots,s_L)$}
          \State Update $\widetilde{\mathcal{L}}=\widetilde{\mathcal{L}}+(1-2(p_y+p_z))^{\abs{s}_X}(1-2(p_x+p_z))^{\abs{s}_Y}(1-2(p_x+p_y))^{\abs{s}_Z}f(\bm{\theta},s,O,\rho)$
      \EndFor
      \EndFor
  \EndFor
  \EndFor
  
  \State Output the approximate cost function $\widetilde{\mathcal{L}}$.
\end{algorithmic}
\end{algorithm}

\subsection{Analysis of the algorithm for \textbf{Case 2}}
\label{ap:algorithm:case2}
Without loss of generality, we assume $p_x$ is the only non-zero factor in $\{p_x,p_y,p_z\}$. The strategy in this case is to truncate the number of other operators less than $M$, formalize as $\abs{s}_Y+\abs{s}_Z=\abs{s_L}_Y+\abs{s_L}_Z+\cdots+\abs{s_0}_Y+\abs{s_0}_Z\leq M$

The algorithm for \textbf{Case 2} is similar to the algorithm for \textbf{Case 1}. The difference is that the cases in $(i+1)$-th step eliminated are $\abs{s_{L}}_Y+\abs{s_{L}}_Z+\abs{s_{L-1}}_Y+\abs{s_{L-1}}_Z \cdots \abs{s_{L-i}}_Y+\abs{s_{L-i}}_Z > M$.

For each case of $s_L$, there are at most $\abs{s_L}$ non-identity elements in $\{s_{L}|_\g{\sigma_{L,j}}\} \cup \{s_{L}|_\g{V_{L,k}}\} \cup \{s_{L}|_{I_{L}}\}$, corresponding to at most $\abs{s_L}$ non-identity elements in $\{s_{L-1}|_\g{\sigma_{L,j}}\} \cup \{s_{L-1}|_\g{V_{L,k}}\} \cup \{s_{L-1}|_{I_{L}}\}$.
Furthermore, each element has at most two potential candidates, resulting in at most $2^\abs{s_L}$ cases for $s_{L-1}$.
This back-propagation process can be represented as a tree starting from $s_L$ and the new branching will only occur when $s_{i}|_\g{\sigma_{i,j}}$ is not identity and $\sigma_{i,j} \in AC(i,s_{i})$, shown in Fig.~\ref{fig:bp_tree}. And the bifurcation points are explained in Fig.~\ref{fig:bp_tree_ele}, where $s_i|_{g(\sigma_{i,j})}$ denotes the candidates of $s_i$ back-propagated to $j$-th gate in $i$-th layer, $s_{i-1}|_{g(\sigma_{i,j})}$ and $s'_{i-1}|_{g(\sigma_{i,j})}$ denotes two candidates of $s_{i-1}$ restricted to $g(\sigma_{i,j})$ equal to $s_i|_{g(\sigma_{i,j})}$ and $i\sigma_{i,j}s_i|_{g(\sigma_{i,j})}$, respectively.

Notice that the tripartite of bifurcation points cannot all consist of tensors of $X$ and $I$, otherwise $s_{i}|_\g{\sigma_{i,j}}$ commutes with $\sigma_{i,j}$. This means that $\abs{s_{i-1}}_Y+\abs{s_{i-1}}_Z$ and $\abs{s'_{i-1}}_Y+\abs{s'_{i-1}}_Z$ can not equal $0$ at the same time. 

\begin{figure}[htbp]
  \begin{tikzpicture}
    [level distance=4mm,level/.style={sibling distance=480/2^#1}],
    \node(a) {}
      child{
        child{
          child{child{child{child{child} child{child}} child{child{child child}}} child{child{child child{child child}} child{child{child child} child{child child}}}} 
          child{child{child{child{child} child} child{child{child child} child{child child}}} child{child{child{child child} child{child child}} child{child{child} child{child}}}}
          } 
        child{
          child{child{child{child{child} child{child}} child{child{child} child{child}}} child{child{child{child child} child{child child}} child{child{child} child{child}}}} 
          child{child{child{child{child child} child{child child}} child{child{child child} child{child child}}} child{child{child{child} child{child}} child{child{child} child{child}}}}
          } 
        };
    \draw[dashed,red] ($(a)+(-5,-2.5mm)$)--($(a)+(5,-2.5mm)$);
    \node at ($(a)+(5.5,-2.5mm)$) {$s_L$};
    \draw[dashed,red] ($(a)+(-5,-18mm)$)--($(a)+(5,-18mm)$);
    \node at ($(a)+(5.5,-18mm)$) {$s_{L-1}$};
    \draw[dashed,red] ($(a)+(-5,-26mm)$)--($(a)+(5,-26mm)$);
    \node at ($(a)+(5.5,-26mm)$) {$s_{L-2}$};
    \node at ($(a)+(0,-30mm)$) {$\vdots$};
    \node at ($(a)+(-3,-30mm)$) {$\vdots$};
    \node at ($(a)+(3,-30mm)$) {$\vdots$};
  \end{tikzpicture}
  \caption{The tree representation of back-propagation process.}\label{fig:bp_tree}
  \end{figure}

\begin{figure}[htbp]
  \begin{tikzpicture}[level distance=10mm,sibling distance=20mm,
    level 3/.style={}],
    \node(a){$s_i|_{g(\sigma_{i,j})}$}
       child {node[draw] {$U_{i,j}(\theta_{i,j})$}
         child {node {$s_{i-1}|_{g(\sigma_{i,j})}$} child {node {$s_i|_{g(\sigma_{i,j})}$}  edge from parent[double]}}
         child {node {$s'_{i-1}|_{g(\sigma_{i,j})}$} child {node {$i \sigma_{i,j}s_i|_{g(\sigma_{i,j})}$} edge from parent[double]}}
       };
  \end{tikzpicture}
  \caption{The explanation of bifurcation points in Fig.\ref{fig:bp_tree}.}\label{fig:bp_tree_ele}
  \end{figure}

Next, we use induction to place an upper bound on the number of Pauli paths that satisfy $\abs{s}_Y+\abs{s}_Z\leq M$, with non-zero contribution and start at a candidate $s_L$:

\begin{enumerate}

\item  If we restrict $\abs{s}_Y+\abs{s}_Z\leq 0$, then there is no bifurcation points in the back-propagation process. This assumption leads to a single Pauli path that starts at a candidate $s_L$. The time complexity of this case is upper bound by $\order{nL}$.

\item  If we restrict $\abs{s}_Y+\abs{s}_Z\leq 1$ and there is a bifurcation point in the back-propagation process located at $j$-th gate in $i$-th layer, then element $s'_{i-1}$ shown in Fig.~\ref{fig:bp_tree_ele} satisfies $\abs{s'_{i-1}}_Y+\abs{s'_{i-1}}_Z\geq 1$. Thus, $s'$ is reduced to the (1) case which satisfies $\abs{s'_{i-2}}_Y+\abs{s'_{i-2}}_Z+\cdots \abs{s'_{0}}_Y+\abs{s'_{0}}_Z \leq 0$, and there is no bifurcation points in the process that follows, illustrated in Fig.~\ref{fig:bp_tree_ele:2}. Besides this, there are at most $nL$ positions for bifurcation points. This assumption leads to at most $nL$ Pauli paths that start at a $s_L$. The time complexity of this case is upper bound by $\order{(nL)^2}$.

\begin{figure}[htbp]
  \begin{tikzpicture}
    [level distance=4mm,level/.style={sibling distance=120/#1}],
    \node(a) {}
      child{
        child{ child{child{child child} child{child}} child{child{child}}
          } 
          child{child{child{child}}}
        };
    \node at ($(a)+(0,-20mm)$) {$\vdots$};
    \node at ($(a)+(-2.5,-20mm)$) {$\vdots$};
    \node at ($(a)+(1.5,-20mm)$) {$\vdots$};
  \end{tikzpicture}
  \caption{The tree representation of back-propagation process while $\abs{s}_Y+\abs{s}_Z\leq 1$.}\label{fig:bp_tree_ele:2}
  \end{figure}

\item  If we restrict $\abs{s}_Y+\abs{s}_Z\leq 2$ and there is a bifurcation point in the back-propagation process located at $j$-th gate in $i$-th layer, then element $s'_{i-1}$ shown in Fig.~\ref{fig:bp_tree_ele} satisfies $\abs{s'_{i-1}}_Y+\abs{s'_{i-1}}_Z\geq 1$. Thus, $s'$ is reduced to the (2) case which satisfies $\abs{s'_{i-2}}_Y+\abs{s'_{i-2}}_Z+\cdots \abs{s'_{0}}_Y+\abs{s'_{0}}_Z \leq 1$. By discussed in (2), there are at most $nL$ paths for $s'$ in the process that follows. Thus there are at most $(nL)^2$ paths for $s$ that start at a $s_L$. The time complexity of this case is upper bound by $\order{(nL)^3}$.

\item Based on the above generalization, it is assumed that if we restrict $\abs{s}_Y+\abs{s}_Z\leq M-1$ then there is at most $(nL)^{M-1}$ paths for $s$ that start at a $s_L$ and the time complexity of this case is upper bound by $\order{(nL)^M}$.

For the case of $\abs{s}_Y+\abs{s}_Z\leq M$, if there is a bifurcation point in the back-propagation process located at $j$-th gate in $i$-th layer, then element $s'_{i-1}$ shown in Fig.~\ref{fig:bp_tree_ele} satisfies $\abs{s'_{i-1}}_Y+\abs{s'_{i-1}}_Z\geq 1$. Similarly, $s'$ is reduced to the case of $\abs{s'_{i-2}}_Y+\abs{s'_{i-2}}_Z+\cdots \abs{s'_{0}}_Y+\abs{s'_{0}}_Z \leq M-1$. By assumption, there are at most $(nL)^{M-1}$ paths for $s'$ in the process that follows.
Thus, there are at most $nL*(nL)^{M-1}=(nL)^{M}$ paths for $s$ that start at a $s_L$.
The time complexity of this case is upper bound by $\order{(nL)^{M+1}}$.

\end{enumerate}

Thus, to compute all contributions of the Pauli path with $\abs{s}_Y+\abs{s}_Z\leq M$, we need to calculate at most $\mathrm{Poly}(n) (nL)^{M}$ different Pauli paths. The time complexity of the enumerate process is $\mathrm{Poly}(n)\order{(nL)^{M+1}}$ and the space complexity is $\order{\mathrm{Poly}(n)+nL}$.

Similarly to Eq.~\ref{ap:eq:T}, for each Pauli path $s$, it is possible to determine $f(\bm{\theta},s,O,\rho)$ with time complexity $\order{nL}+\mathrm{Poly}(n)$ using Eq.~\eqref{ap:eq:f} and Proposition~\ref{prop:f_ele}. 
Thus, the overall time cost for computing $\widetilde{\mathcal{L}}$ in \textbf{Case 2} is about 

\begin{equation}\label{ap:eq:T2}
  \left( \order{nL} + \mathrm{Poly}(n) \right)  \mathrm{Poly}(n) (nL)^{M}+\mathrm{Poly}(n)\order{(nL)^{M+1}} = \mathrm{Poly}(n)\order{(nL)^{M+1}}.
\end{equation}

\section{Proof of Lemma~\ref{lemma:f_noisy}}
\label{sec:ap:lemma:f_noisy}
In Lemma~\ref{lemma:f_noisy}, we expressed the contribution of a Pauli path $s=(s_0,\cdots,s_L)\in \bm{P}^{L+1}_n$ in cost function $L(\bm{\theta})$ as:
\begin{equation}
  f(\bm{\theta},s,O,\rho)=\Tr{Os_L}\left(\prod_{i=1}^{L}\Tr{s_i\mathcal{U}_i s_{i-1}\mathcal{U}_i^\dagger}\right)\Tr{s_0\rho}.
\end{equation}
By Eq.~\eqref{ap:eq:Npauli_path_integral}, the contribution of a Pauli path $s=(s_0,\cdots,s_L)\in \bm{P}^{L+1}_n$ in noisy cost function $\widehat{\mathcal{L}}(\bm{\theta})$ can be expressed as:
\begin{equation}
  \hat{f}(\bm{\theta},s,O,\rho)=\Tr{O\mathcal{N}^{\otimes n}(s_L)}\left(\prod_{i=1}^{L}\Tr{s_i\mathcal{U}_i \mathcal{N}^{\otimes n}(s_{i-1})\mathcal{U}_i^\dagger}\right)\Tr{s_0\rho}.
\end{equation}

We generalize the assumption that $\mathcal{N}$ represents a single qubit Pauli error channel \begin{equation}\label{eq:define_N_Pauli}
  \mathcal{N}(\phi)=(1-p_x-p_y-p_z)\phi+ p_x X\phi X+ p_y Y\phi Y+ p_z Z\phi Z ,
\end{equation}
where $\phi$ is a single-qubit density matrix and $p_x,p_y,p_z$ represent the probabilities of $X,Y,Z$ error occuring, respectively. 

For $i=1,\cdots,L$, we have
\begin{equation}
  \mathcal{N}^{\otimes n}(s_i)=\mathcal{N}(s_i|_1)\otimes \cdots \otimes \mathcal{N}(s_i|_n),
\end{equation}
where notation $|_j$ represents that limit the operator on $j$-th qubit.
Simple calculations show that $\mathcal{N}(\frac{\mathbb{I}}{\sqrt{2}})=\frac{\mathbb{I}}{\sqrt{2}}$, $\mathcal{N}(\frac{X}{\sqrt{2}})=(1-2(p_y+p_z))\frac{X}{\sqrt{2}}$, $\mathcal{N}(\frac{Y}{\sqrt{2}})=(1-2(p_x+p_z))\frac{Y}{\sqrt{2}}$ and $\mathcal{N}(\frac{Z}{\sqrt{2}})=(1-2(p_x+p_y))\frac{Z}{\sqrt{2}}$.
Thus, for any $s_i\in \bm{P}_n=\{\frac{\mathbb{I}}{\sqrt{2}},\frac{X}{\sqrt{2}},\frac{Y}{\sqrt{2}},\frac{Z}{\sqrt{2}}\}^{\otimes n}$, we have
\begin{equation}
  \mathcal{N}^{\otimes n}(s_i)=(1-2(p_y+p_z))^{\abs{s_i}_X}(1-2(p_x+p_z))^{\abs{s_i}_Y}(1-2(p_x+p_y))^{\abs{s_i}_Z}s_i,
\end{equation}
where $\abs{s_i}_X,\abs{s_i}_Y,\abs{s_i}_Z$ denote the number of operates $\frac{X}{\sqrt{2}},\frac{Y}{\sqrt{2}},\frac{Z}{\sqrt{2}}$ in $s_i$, respectively. So we get 
\begin{equation}\label{ap:eq:noisy_f_term_decrease:pauli}
  \hat{f}(\bm{\theta},s,O,\rho)=(1-2(p_y+p_z))^{\abs{s}_X}(1-2(p_x+p_z))^{\abs{s}_Y}(1-2(p_x+p_y))^{\abs{s}_Z}f(\bm{\theta},s,O,\rho),
\end{equation}
where $\abs{s}_P=\abs{s_0}_P+\cdots+\abs{s_L}_P$ for $P \in \{\frac{X}{\sqrt{2}},\frac{Y}{\sqrt{2}},\frac{Z}{\sqrt{2}}\}$.

These complete the proof of Lemma~\ref{lemma:f_noisy}.

In particular, the depolarizing channel $\mathcal{N}(\phi)=(1-\lambda) \phi + \frac{\lambda}{2} \text{Tr}(\phi)\mathbb{I}$  is obtained when $p_x = p_y = p_z = \frac{\lambda}{4}$. Eq.~\eqref{ap:eq:noisy_f_term_decrease:pauli} is transformed to: 
\begin{equation}
  \hat{f}(\bm{\theta},s,O,\rho)=(1-\lambda)^{\abs{s}_X+\abs{s}_Y+\abs{s}_Z}f(\bm{\theta},s,O,\rho)=(1-\lambda)^\abs{s}f(\bm{\theta},s,O,\rho).
\end{equation}

\section{Proof of Proposition~\ref{prop:f_ele}}\label{ap:prop:f_ele}

In Proposition~\ref{prop:f_ele}, we claimed that the elements corresponding to each layer in $f$ can be calculated as the following rules with time and space cost $\order{n}$:

\begin{equation}\label{ap:eq:gate_ele}
  \begin{aligned}
  \Tr{s_i\mathcal{U}_i s_{i-1}\mathcal{U}_i^\dagger}&=\Tr{\left(s_i s_{i-1}\right)\big|_{I_i}}\prod_{k=1}^{C_i}\Tr{\left(s_iV_{i,k} s_{i-1}V_{i,k}^{\dagger}\right)\big|_\g{V_{i,k}}}\prod_{\sigma_{i,j}\in C(i,s_{i-1})}\Tr{\left(s_i s_{i-1}\right)\big|_\g{\sigma_{i,j}}}\\
  &\prod_{\sigma_{i,j'}\in AC(i,s_{i-1})} \Bigg\{ \Tr{\left(s_i s_{i-1}\right)\big|_\g{\sigma_{i,j'}}}\cos{\theta_{i,j'}}- \Tr{\left(is_i \sigma_{i,j'}s_{i-1}\right)\big|_\g{\sigma_{i,j'}}}\sin{\theta_{i,j'}}  \Bigg\}.
\end{aligned}
\end{equation}
Here the set $C(i,s_{i-1})$ and $AC(i,s_{i-1})$ denote the sets of Pauli words in $\{\sigma_{i,j}|j=1,\cdots,R_i\}$ commute and anti-commute with $s_{i-1}$, respectively. 
The symbol $\big|_\g{\sigma_{i,j}}$ denotes that limit the operator on the qubits of $\sigma_{i,j}$ non-trivially
applied, $\big|_\g{V_{i,k}}$ denotes that limit the operator on the qubits with Clifford gates $V_{i,k}$ non-trivially applied and $\big|_{I_i}$ denotes the limitation on the qubits without gates applied in $i$-th layer.

\begin{proof}

In our setting, $\mathcal{U}_i$ is composed by a series of gates $U_{i,1},\cdots,U_{i,R_i}$ and $V_{i,1},\cdots,V_{i,C_i}$ without operating twice on each qubit. 
Pauli rotation gates $U_{i,j}(\theta_{i,j})$ have form
\begin{equation}
  U_{i,j}(\theta_{i,j})=\exp{-i \frac{\theta_{i,j}}{2} \sigma_{i,j}}.
\end{equation}

Clifford gates $V_{i,k}$ can be chosen from $\{\mathrm{H}(a),\mathrm{S}(a),\mathrm{CNOT}(a,b) \}_ {a\neq b}$.
Then
\begin{equation}
  \begin{aligned}
    \Tr{s_i\mathcal{U}_i s_{i-1}\mathcal{U}_i^\dagger}=\Tr{\left(s_i s_{i-1}\right)\big|_{I_i}}
    \prod_{k=1}^{C_i}\Tr{\left(s_iV_{i,k} s_{i-1}V_{i,k}^{\dagger}\right)\big|_\g{V_{i,k}}}
    \prod_{j=1}^{R_i}\Tr{\left(s_i U_{i,j}(\theta_{i,j}) s_{i-1} U_{i,j}^\dagger(\theta_{i,j})\right) \big|_\g{\sigma_{i,j}}}.
  \end{aligned}
\end{equation}

The exponent of a Hermitian operate $X$ is defined as Taylar expansion $\exp{X}=\sum_{k=0}^\infty\frac{X^k}{k!}$. When calculating rotation on Pauli words, the square of any Pauli word is identity $\sigma^2=\mathbb{I}$, thus we have
\begin{equation}
  \exp{-i \frac{\theta}{2} \sigma}=\sum_{k=0}^\infty\frac{(-i \frac{\theta}{2}\sigma)^k}{k!}=\sum_{k=0}^\infty\frac{(-1)^k (\frac{\theta}{2})^{2k}}{(2k)!} \mathbb{I}- i\frac{(-1)^k (\frac{\theta}{2})^{2k+1}}{(2k+1)!}\sigma=\cos{\frac{\theta}{2}}\mathbb{I}-i \sin{\frac{\theta}{2}}\sigma.
\end{equation}

Therefore, according to the exchange relation of $\sigma$ and another Pauli word $\sigma'$, we have
\begin{equation}
\begin{aligned}
&\sigma'\exp{-i \frac{\theta}{2} \sigma}=\exp{-i \frac{\theta}{2} \sigma}\sigma',\text{ if }\sigma\text{ commutes with }\sigma'.\\
&\sigma'\exp{-i \frac{\theta}{2} \sigma}=\exp{i \frac{\theta}{2} \sigma}\sigma',\text{ if }\sigma\text{ anti-commutes with }\sigma'.\\
\end{aligned}
\end{equation}
So we can divide $\{\sigma_{i,j}\}$ into two case. If $\sigma_{i,j}$ commutes with $s_{i-1}$, we have

\begin{equation}
  \begin{aligned}
  \Tr{\left(s_i U_{i,j}(\theta_{i,j}) s_{i-1} U_{i,j}^\dagger(\theta_{i,j})\right)\big|_\g{\sigma_{i,j}}}
  &=\Tr{(s_i \exp{-i \frac{\theta_{i,j}}{2}\sigma_{i,j} }  \exp{i \frac{\theta_{i,j}}{2}\sigma_{i,j} }s_{i-1}) \big|_\g{\sigma_{i,j}}}\\
  &=\Tr{(s_i \exp{-i \frac{\theta_{i,j}}{2}\sigma_{i,j} } \exp{i \frac{\theta_{i,j}}{2}\sigma_{i,j} }s_{i-1}) \big|_{\g{\sigma_{i,j}}}}\\
  &=\Tr{(s_i s_{i-1})\big|_{\g{\sigma_{i,j}}}}.
\end{aligned}
\end{equation}
While $\sigma_{i,j}$ anti-commutes with $s_{i-1}$, we have
\begin{equation}
  \begin{aligned}
  \Tr{\left(s_i U_{i,j}(\theta_{i,j}) s_{i-1} U_{i,j}^\dagger(\theta_{i,j})\right)\big|_\g{\sigma_{i,j}}}
  &=\Tr{ \left(s_i\exp{-i \frac{\theta_{i,j}}{2}\sigma_{i,j} } s_{i-1}\exp{i \frac{\theta_{i,j}}{2}\sigma_{i,j}}\right) \big|_\g{\sigma_{i,j}}}\\
  &=\Tr{\left( s_i  \exp{-i\theta_{i,j}\sigma_{i,j}} s_{i-1}\right) \big|_\g{\sigma_{i,j}}}\\
  &=\Tr{ \left(s_i (\cos{\theta_{i,j}}\mathbb{I}-i\sin{\theta_{i,j}}\sigma_{i,j}) s_{i-1} \right) \big|_\g{\sigma_{i,j}}}\\
  &=\Tr{\left(s_i s_{i-1}\right)\big|_\g{\sigma_{i,j}}}\cos{\theta_{i,j}}-\Tr{\left(i s_i \sigma_{i,j}s_{i-1}\right)\big|_\g{\sigma_{i,j}}}\sin{\theta_{i,j}}.
\end{aligned}
\end{equation}
These complete the proof of Eq.~\eqref{ap:eq:gate_ele}.
\end{proof}

\begin{remark}
  According to Eq.~\eqref{ap:eq:gate_ele}, we can obtain the following rules for the elements corresponding to each layer in $f$, for which $\Tr{s_i\mathcal{U}_i s_{i-1}\mathcal{U}_i^\dagger} \neq 0$:
\begin{enumerate}
  \item For qubits without non-trivial gates applied, we must have $s_{i-1}\big|_{I_i} = s_{i}\big|_{I_i}$.
  \item For Clifford gates, the only case for $\Tr{s_i\mathcal{U}_i s_{i-1}\mathcal{U}_i^\dagger} \neq 0$ is $s_{i-1}\big|_\g{V_{i,k}}= \pm V_{i,k}^{\dagger} s_{i}V_{i,k}\big|_\g{V_{i,k}}$, where the sign $\pm$ depends on the sign of Pauli word $V_{i,k}^{\dagger} s_{i}V_{i,k}\big|_\g{V_{i,k}}$.
  \item For Pauli rotation gates, if $\sigma_{i,j}$ commutes with $s_{i-1}$, $s_{i}|_\g{\sigma_{i,j}}=s_{i-1}|_\g{\sigma_{i,j}}$.
  \item If $\sigma_{i,j}$ anti-commutes with $s_{i-1}$, we encounter two cases that $\Tr{s_i\mathcal{U}_i s_{i-1}\mathcal{U}_i^\dagger} \neq 0$, $s_{i}|_\g{\sigma_{i,j}}=s_{i-1}|_\g{\sigma_{i,j}}$ with a factor of $\cos{\theta_{i,j}}$; or $s_{i}|_\g{\sigma_{i,j}}=\pm i \sigma_{i,j} s_{i-1}|_\g{\sigma_{i,j}}$ with a factor of $\mp \sin{\theta_{i,j}}$, where the sign $\pm$ depends on the sign of Pauli word $i \sigma_{i,j} s_{i-1}|_\g{\sigma_{i,j}}$.
  \item If $\Tr{s_i\mathcal{U}_i s_{i-1}\mathcal{U}_i^\dagger} \neq 0$ holds, there are $C(i,s_{i-1})=C(i,s_{i})$ and $AC(i,s_{i-1})=AC(i,s_{i})$.
  \end{enumerate}
\end{remark}

\section{Proof of Lemma~\ref{lemma:MSE_l}}\label{ap:lemma:MSE_l}

Before providing the proof, we first present the following lemma to handle the cross-terms, whose detailed proof is in Supplement Material~\ref{ap:lemma:cross_items} and \cite{gao2018efficient}.

\begin{lemma}\label{lemma:cross_items}
Suppose Eq.~\eqref{ap:eq:generate} is satisfied, for any distinct Pauli paths $s,s^{\prime}\in \bm{P}^{L+1}_n$, we have 
\begin{equation}\label{ap:eq:E_cross_equals_0}
\mathbb{E}_{\bm{\theta}}f(\bm{\theta},s,O,\rho)f(\bm{\theta},s^{\prime},O,\rho)=0.
\end{equation}
\end{lemma}

After obtain Eq.~\eqref{ap:eq:E_cross_equals_0}, we can prove Lemma~\ref{lemma:MSE_l} as follow:
\begin{proof}[Proof of Lemma~\ref{lemma:MSE_l}]
By Lemma~\ref{lemma:f_noisy} and Lemma~\ref{lemma:cross_items}, an estimation of the Mean-Square Error (MSE) between $\widetilde{\mathcal{L}}$ and $\widehat{\mathcal{L}}$ can be derived as follow:
\begin{equation}\label{ap:eq:diff}
\begin{aligned}
\mathbb{E}_{\bm{\theta}}\abs{\widetilde{\mathcal{L}}-\widehat{\mathcal{L}}}^2 &=\mathbb{E}_{\bm{\theta}}\big[\sum_{\abs{s}> M}\hat{f}(\bm{\theta},s,O,\rho)\big]^2\\
&\overset{\eqref{ap:eq:E_cross_equals_0}}{=}\mathbb{E}_{\bm{\theta}} \sum_{\abs{s}> M}\hat{f}(\bm{\theta},s,O,\rho)^2\\
&\overset{\eqref{ap:eq:noisy_f_term_decrease}}{=}\mathbb{E}_{\bm{\theta}} \sum_{\abs{s}> M}(1-2(p_y+p_z))^{2\abs{s}_X}(1-2(p_x+p_z))^{2\abs{s}_Y}(1-2(p_x+p_y))^{2\abs{s}_Z}{f}(\bm{\theta},s,O,\rho)^2,
\end{aligned}
\end{equation}
where the last equality holds by
\begin{equation}\label{ap:eq:noisy_f_term_decrease}
  \hat{f}(\bm{\theta},s,O,\rho)=(1-2(p_y+p_z))^{\abs{s}_X}(1-2(p_x+p_z))^{\abs{s}_Y}(1-2(p_x+p_y))^{\abs{s}_Z}f(\bm{\theta},s,O,\rho).
\end{equation}

In finite-size systems, the expected values could be bounded by a finite $\norm{O}_\infty$, note that in the noiseless setting
\begin{equation}
\abs{\sum_s f(\bm{\theta},s,O,\rho)}=\abs{\Tr{OU(\bm{\theta})\rho U^\dagger(\bm{\theta})}}\leq \norm{O}_\infty.
\end{equation}
Thus, we have
\begin{equation}\label{ap:eq:normbound}
\mathbb{E}_{\bm{\theta}} \sum_{s}f(\bm{\theta},s,O,\rho)^2\overset{\eqref{ap:eq:E_cross_equals_0}}{=}\mathbb{E}_{\bm{\theta}}\big[\sum_{s}f(\bm{\theta},s,O,\rho)\big]^2 \leq \norm{O}_\infty^2.
\end{equation}

We define $\gamma=\min\{p|{p \in \{p_x,p_y,p_z\},p\neq 0}\}$ as the minimal non-zero factor in $\{p_x,p_y,p_z\}$. The truncation weight has the following two definitions:
\begin{itemize}
  \item \textbf{Case 1}: When there are at least two non-zero elements in $\{p_x,p_y,p_z\}$, the truncation weight is $\abs{s}=\sum_{P\in \{X,Y,Z\}}\abs{s}_P$. 
  \item \textbf{Case 2}: When there is only one element in $\{p_x,p_y,p_z\}$ non-zero~(without loss of generality $p_x\neq 0$), truncation weight is $\abs{s}=\sum_{P\in \{Y,Z\}}\abs{s}_P$.
\end{itemize}

Combining Eq.~\ref{ap:eq:normbound} with Eq.~\eqref{ap:eq:diff}, we obtain the following inequality for both case 1 and case 2:
\begin{equation}\label{ap:eq:expectation_errors}
  \mathbb{E}_{\bm{\theta}}\abs{\widetilde{\mathcal{L}}-\widehat{\mathcal{L}}}^2 \leq (1-2\gamma)^{2M}\norm{O}_\infty^2\leq\exp{-4\gamma M} \norm{O}_\infty^2.
\end{equation}

Here the first inequality holds because $f(\bm{\theta},s,O,\rho)\in \mathbb{R}$. To elucidate that $f(\bm{\theta},s,O,\rho)\in \mathbb{R}$, several observations can be made from Eq.~\eqref{ap:eq:f} and Eq.~\eqref{ap:eq:gate_ele}:
\begin{itemize}
  \item The Hermiticity of operators $O$ and $\rho$ ensures $\Tr{Os_L}$ and $\Tr{s_0\rho}$ are real. 
  \item The terms $C(i,s_{i-1})$ and ${I_i}$ in the $\Tr{s_i\mathcal{U}_i s_{i-1}\mathcal{U}_i^\dagger}$ correspond to the inner product of Pauli words, which are real. 
  \item The realness of the $V_{i,k}$ Clifford term can be verified by exhaustively applying $\{\mathrm{H},\mathrm{S}, \mathrm{CNOT}\}$ to Pauli matrices $\{\mathbb{I},{X}, {Y}, {Z}\}$. 
  \item The term $AC(i,s_{i-1})$ is also real due to the product property of Pauli matrices \footnote{The Pauli matrices, $\sigma_1=X,\sigma_2=Y$ and $\sigma_3=Z$, satisfy the relation $\sigma_i \sigma_j=i\epsilon_{ijk}\sigma_k$, where $\epsilon_{ijk}$ is the Levi-Civita symbol and Einstein summation notation is used. Moreover, $s_{i-1}$ and $\sigma_{i,j'}$ anti-commute. Thus $i\sigma_{i,j'}s_{i-1}$ is a Pauli Word with a sign $\pm$.}.
\end{itemize}

Thus for $\forall\nu > 0$, if we have 
\begin{equation}
  M\geq\frac{1}{4\gamma}\ln{\frac{\norm{O}_\infty^2}{\nu}},
\end{equation}
then by Eq.~\eqref{ap:eq:expectation_errors}, $\mathbb{E}_{\bm{\theta}}\abs{\widetilde{\mathcal{L}}-\widehat{\mathcal{L}}}^2\leq \nu$ holds. This completes the proof of Lemma~\ref{lemma:MSE_l}.
\end{proof}

\begin{remark}
  In particular, for the depolarizing channel $p_x = p_y = p_z = \frac{\lambda}{4}$, Eq.~\ref{ap:eq:noisy_f_term_decrease} is substituted by 
  $\hat{f}(\bm{\theta},s,O,\rho)=(1-\lambda)^\abs{s}f(\bm{\theta},s,O,\rho)$.
  Analogous to Eq.~\ref{ap:eq:expectation_errors}, it follows that
  \begin{equation}
    \mathbb{E}_{\bm{\theta}}\abs{\widetilde{\mathcal{L}}-\widehat{\mathcal{L}}}^2 \leq (1-\lambda)^{2M}\norm{O}_\infty^2\leq\exp{-2\gamma M} \norm{O}_\infty^2.
  \end{equation}
  Therefore, for depolarizing channel 
  \begin{equation}\label{ap:eq:depolarizing_M}
    M\geq\frac{1}{2\lambda}\ln{\frac{\norm{O}_\infty^2}{\nu}}
  \end{equation}
  is sufficient to guarantee $\mathbb{E}_{\bm{\theta}}\abs{\widetilde{\mathcal{L}}-\widehat{\mathcal{L}}}^2\leq \nu$.
\end{remark}

Combining with Eq.~\eqref{ap:eq:T}, the subsequent corollary provides the time complexity of obtaining the approximated noisy cost function $\widetilde{\mathcal{L}}$ in our method:
\begin{corollary}\label{cor:MSE_T}
  Suppose Eq.~\eqref{ap:eq:generate} is satisfied, for $\forall\nu > 0$ and a fixed error rate $\gamma$, the time complexity of obtaining $\widetilde{\mathcal{L}}$ with mean-square error $\mathbb{E}_{\bm{\theta}}\abs{\widetilde{\mathcal{L}}-\widehat{\mathcal{L}}}^2\leq \nu$ are $\mathrm{Poly}(n) \order{L} \left( \frac{\norm{O}_\infty}{\sqrt{\nu}} \right)^{\order{\frac{1}{\gamma}}}=\mathrm{Poly}(n,L,\frac{1}{\sqrt{\nu}},\norm{O}_\infty)$ for case 1 and  a quasi-polynomial relation $\mathrm{Poly}(n)\order{(nL)^{\ln{\frac{\norm{O}_\infty}{\sqrt{\nu}}}}}^{\order{\frac{1}{\gamma}}}$ for case 2.
\end{corollary}

\begin{proof}
  By Lemma~\ref{lemma:MSE_l}, we can set $M=\frac{1}{4\gamma}\ln{\frac{\norm{O}_\infty^2}{\nu}}$, while the mean-square error(MSE) $\mathbb{E}_{\bm{\theta}}\abs{\widetilde{\mathcal{L}}-\widehat{\mathcal{L}}}^2$ is below $\nu$.

Using the algorithm in Supplement Material~\ref{ap:algorithm}, the total time cost for obtaining $\mathcal{\widetilde{L}}$ in case 1 is about:
\begin{equation}
\begin{aligned}
\mathrm{Poly}(n) \order{L} 2^{M}
&=
\mathrm{Poly}(n) \order{L} \left( \frac{\norm{O}_\infty}{\sqrt{\nu}} \right)^{\order{\frac{1}{\gamma}}}\\
&=\mathrm{Poly}(n,L,\frac{1}{\sqrt{\nu}},\norm{O}_\infty).
\end{aligned}
\end{equation}

And in case 2 is about:
\begin{equation}
  \begin{aligned}
  \mathrm{Poly}(n)\order{(nL)^{M+1}}
  &=
  \mathrm{Poly}(n)\order{(nL)^{\frac{1}{2\gamma}\ln{\frac{\norm{O}_\infty}{\sqrt{\nu}}}+1}}\\
  &=\mathrm{Poly}(n)\order{(nL)^{\ln{\frac{\norm{O}_\infty}{\sqrt{\nu}}}}}^{\order{\frac{1}{\gamma}}}.
\end{aligned}
\end{equation}

\end{proof}

\section{Proof of Lemma~\ref{lemma:cross_items}}
\label{ap:lemma:cross_items}
In the Lemma~\ref{lemma:cross_items}, we claimed that if the set of Pauli words $\{\overline{\sigma}_{i,j}\}$ can generate $\{\mathbb{I},X,Y,Z\}^{\otimes n}$ up to phase
\begin{equation}\label{ap:eq:generate}
  \langle \{\overline{\sigma}_{i,j}\}\rangle/\left(\langle \{\overline{\sigma}_{i,j}\}\rangle\cap\langle i\mathbb{I}^{\otimes n}\rangle\right)=\{\mathbb{I},X,Y,Z\}^{\otimes n},
\end{equation}
then for any distinct Pauli paths $s$ and $s^{\prime}$ we have 
\begin{equation}
\mathbb{E}_{\bm{\theta}}f(\bm{\theta},s,O,\rho)f(\bm{\theta},s^{\prime},O,\rho)=0,
\end{equation}
where $\overline{\sigma}_{i,j}$ is defined as:
\begin{equation}\label{ap:eq:effected_Pauli}
  \overline{\sigma}_{i,j}= \mathcal{V}_{L} \cdots \mathcal{V}_{i} \sigma_{i,j} \mathcal{V}_{i}^\dagger\cdots \mathcal{V}_{L}^\dagger.
\end{equation}

In order to prove this claim, we first define the ``split" relation between sets of Pauli words:
\begin{definition}
  There are two sets of Pauli words $A$ and $B$. We define $B$ to be split by $A$ if there exist no two distinct elements in $B$ that exhibit identical anti-commute/commute relation with each element in $A$.
\end{definition}
\begin{remark}
    The name ``split" is used because if $A$ can split $B$, then any element in $B$ can be uniquely determined by characterizing its exchange relation with each element in $A$. In a sense, $A$ separates each element in $B$ into independent parts by characterizing their exchange relationship.
\end{remark}

Before the discussion, we introduce the following lemma:
\begin{lemma}\label{ap:lemma:3terms_exchange}
    Assume $\mathcal{P}$, $\sigma_a$, $\sigma_b$ and $\sigma_c$ are Pauli words. If $\mathcal{P}\sigma_a$ and $\mathcal{P}\sigma_b$ have the same commute or anti-commute relation with $\sigma_c$. Then $\sigma_a$ and $\sigma_b$ have the same commute or anti-commute relation with $\sigma_c$.

    \begin{proof}
        First, we assume $\mathcal{P}\sigma_a$ and $\mathcal{P}\sigma_b$ commute with $\sigma_c$, for $i=a,b$ it can be expressed as 
        \begin{equation}\label{ap:eq:3terms_exchange}
            \mathcal{P}\sigma_i\sigma_c=\sigma_c\mathcal{P}\sigma_i=\pm\mathcal{P}\sigma_c\sigma_i,
        \end{equation}
        where the sign $\pm$ is set to $+$ if and only if $\mathcal{P}$ commutes with $\sigma_c$.
        
        This leads to
        \begin{equation}
        \sigma_i\sigma_c=\pm\sigma_c\sigma_i,
        \end{equation}
        where the sign $\pm$ is same as Eq.~\eqref{ap:eq:3terms_exchange}, for $i=a,b$.

        In a similar way to discuss the case of $\mathcal{P}\sigma_a$ and $\mathcal{P}\sigma_b$ anti-commute with $\sigma_c$, we obtain $\sigma_a$ and $\sigma_b$ have the same commute or anti-commute relation with $\sigma_c$.
    \end{proof}
\end{lemma}

We will demonstrate that if $\{\overline{\sigma}_{i,j}\}$ can split Pauli word set $\{\sigma\}$ which linearly compose $O$, then a similar conclusion can be established.
\begin{lemma}\label{ap:lemma:split}
Suppose the set of Pauli words $\{\overline{\sigma}_{i,j}\}$ can split the Pauli word set $\{\sigma\}$ of $O$. Then for any distinct Pauli paths $s\neq s^{\prime}\in \bm{P}^{L+1}_n$, we have 
\begin{equation}
\mathbb{E}_{\bm{\theta}}f(\bm{\theta},s,O,\rho)f(\bm{\theta},s^{\prime},O,\rho)=0.
\end{equation}

\begin{proof}    
Note that
\begin{equation}\label{ap:eq:E_cross_terms}
\begin{aligned}
&\mathbb{E}_{\bm{\theta}}f(\bm{\theta},s,O,\rho)f(\bm{\theta},s^{\prime},O,\rho)\\
&=\Tr{Os_L}\Tr{Os'_L}\left(\prod_{i=1}^{L}\mathbb{E}_{\bm{\theta}_i}\Tr{s_i\mathcal{U}_i s_{i-1}\mathcal{U}_i^\dagger}\Tr{s'_i\mathcal{U}_i s'_{i-1}\mathcal{U}_i^\dagger}\right)\Tr{s_0\rho}\Tr{s'_0\rho},
\end{aligned}
\end{equation}
Thus $\mathbb{E}_{\bm{\theta}_i}\Tr{s_i\mathcal{U}_i s_{i-1}\mathcal{U}_i^\dagger}\Tr{s'_i\mathcal{U}_i s'_{i-1}\mathcal{U}_i^\dagger}=0$ for some $i$, leads to $\mathbb{E}_{\bm{\theta}}f(\bm{\theta},s,O,\rho)f(\bm{\theta},s^{\prime},O,\rho)=0$.

By Proposition~\ref{prop:f_ele}, the element corresponding to $i$-th layer of Eq.~\eqref{ap:eq:E_cross_terms} can be written as
\begin{equation}\label{ap:eq:E_cross_i-layer}
\begin{aligned}
\mathbb{E}_{\bm{\theta}_i}&\Tr{s_i\mathcal{U}_i s_{i-1}\mathcal{U}_i^\dagger}\Tr{s'_i\mathcal{U}_i s'_{i-1}\mathcal{U}_i^\dagger}\\
=&\Tr{(s_i s_{i-1})\big|_{I_i}}\Tr{(s'_i s'_{i-1})\big|_{I_i}}
\prod_{k=1}^{C_i}\Tr{\left(s_iV_{i,k} s_{i-1}V_{i,k}^{\dagger}\right)\big|_\g{V_{i,k}}}\Tr{\left(s'_iV_{i,k} s'_{i-1}V_{i,k}^{\dagger}\right)\big|_\g{V_{i,k}}}\\
&\prod_{\sigma_{i,j}\in C(i,s_{i-1})}\Tr{(s_i s_{i-1})\big|_{\g{\sigma_{i,j}}}}\prod_{\sigma_{i,l}\in C(i,s_{i-1}')}\Tr{(s'_i s'_{i-1})\big|_{\g{\sigma_{i,l}}}}\\
&\mathbb{E}_{\bm{\theta}_i}\left\{\prod_{\sigma_{i,j'}\in AC(i,s_{i-1})}\Bigg[ \Tr{(s_i s_{i-1})\big|_{\g{\sigma_{i,j'}}}}\cos{\theta_{i,j'}}
- \Tr{(is_i \sigma_{i,j'}s_{i-1})\big|_{\g{\sigma_{i,j'}}}}\sin{\theta_{i,j'}}\Bigg]\right.\\
&\left.\prod_{\sigma_{i,l'}\in AC(i,s_{i-1}')}\Bigg[ \Tr{(s'_i s'_{i-1})\big|_{\g{\sigma_{i,l'}}}}\cos{\theta_{i,l'}}
- \Tr{(is'_i \sigma_{i,l'}s'_{i-1})\big|_{\g{\sigma_{i,l'}}}}\sin{\theta_{i,l'}}\Bigg]\right\}.
\end{aligned}
\end{equation}

The following proof is divided into two parts. 
\begin{proofpart}
In the first part of the proof, we show that if $\{\overline{\sigma}_{i,j}\}$ can split the Pauli word set $\{\sigma\}$ of $O$, then $s_L= s_L'$, otherwise $\mathbb{E}_{\bm{\theta}}f(\bm{\theta},s,O,\rho)f(\bm{\theta},s^{\prime},O,\rho)=0.$

If there are $s_{i-1}$ and $s_{i-1}'$ have different exchange relation with $\sigma_{i,j}$ at $i$-th layer. 
Without loss of generality, we assume $s_{i-1}$ commutes with $\sigma_{i,j}$, whereas $s_{i-1}'$ anti-commutes with $\sigma_{i,j}$.
By the anti-commutation, $i \sigma_{i,j} s_{i-1}'$ is a normalized Pauli word with a sign factor $\pm$.
As described in the remark of Proposition~\ref{prop:f_ele}, we have $s_{i}'|_{\g{\sigma_{i,j}}}=s_{i-1}'|_{\g{\sigma_{i,j}}}$ with factor $\cos{\theta_{i,j}}$ or $s_{i}'|_{\g{\sigma_{i,j}}}=\left(i\sigma_{i,j} s_{i-1}'\right)|_{\g{\sigma_{i,j}}}$ (up to sign $\pm$) with factor $\sin{\theta_{i,j}}$, otherwise $\Tr{s'_i\mathcal{U}_i s'_{i-1}\mathcal{U}_i^\dagger}=0$.
However, there is still $\mathbb{E}_{\bm{\theta}_i}\Tr{s_i\mathcal{U}_i s_{i-1}\mathcal{U}_i^\dagger}\Tr{s'_i\mathcal{U}_i s'_{i-1}\mathcal{U}_i^\dagger}=0$, because of $\mathbb{E}_{\theta_{i,j}} \sin{\theta_{i,j}}=\mathbb{E}_{\theta_{i,j}} \cos{\theta_{i,j}}=0$.

Thus, for any layer $i=1,\cdots,L$ and $j=1,\cdots,R_i$, Pauli words $s_{i-1}$ and $s_{i-1}'$ have the same exchange relation with $\sigma_{i,j}$. 
In this setting, Eq.~\eqref{ap:eq:E_cross_i-layer} can be written as
\begin{equation}\label{ap:eq:E_cross_i-layer-rewrited}
\begin{aligned}
&\mathbb{E}_{\bm{\theta}_i}\Tr{s_i\mathcal{U}_i s_{i-1}\mathcal{U}_i^\dagger}\Tr{s'_i\mathcal{U}_i s'_{i-1}\mathcal{U}_i^\dagger}\\
=&\Tr{(s_i s_{i-1})\big|_{I_i}}\Tr{(s'_i s'_{i-1})\big|_{I_i}}
\prod_{k=1}^{C_i}\Tr{\left(s_iV_{i,k} s_{i-1}V_{i,k}^{\dagger}\right)\big|_\g{V_{i,k}}}\Tr{\left(s'_iV_{i,k} s'_{i-1}V_{i,k}^{\dagger}\right)\big|_\g{V_{i,k}}}\\
&\prod_{\sigma_{i,j}\in C(i,s_{i-1})}\Tr{(s_i s_{i-1})\big|_{\g{\sigma_{i,j}}}}\Tr{(s'_i s'_{i-1})\big|_{\g{\sigma_{i,j}}}}\\
&\prod_{\sigma_{i,j'}\in AC(i,s_{i-1})}\mathbb{E}_{\theta_{i,j'}}\left\{ \Bigg[ \Tr{(s_i s_{i-1})\big|_{\g{\sigma_{i,j'}}}}\cos{\theta_{i,j'}}
- \Tr{(is_i \sigma_{i,j'}s_{i-1})\big|_{\g{\sigma_{i,j'}}}}\sin{\theta_{i,j'}}\Bigg]\right.\\
&\left.\quad\quad\quad\quad\quad\Bigg[ \Tr{(s'_i s'_{i-1})\big|_{\g{\sigma_{i,j'}}}}\cos{\theta_{i,j'}}
- \Tr{(is'_i \sigma_{i,j'}s'_{i-1})\big|_{\g{\sigma_{i,j'}}}}\sin{\theta_{i,j'}}\Bigg]\right\}\\
=&\Tr{(s_i s_{i-1})\big|_{I_i}}\Tr{(s'_i s'_{i-1})\big|_{I_i}}
\prod_{k=1}^{C_i}\Tr{\left(s_iV_{i,k} s_{i-1}V_{i,k}^{\dagger}\right)\big|_\g{V_{i,k}}}\Tr{\left(s'_iV_{i,k} s'_{i-1}V_{i,k}^{\dagger}\right)\big|_\g{V_{i,k}}}\\
&\prod_{\sigma_{i,j}\in C(i,s_{i-1})}\Tr{(s_i s_{i-1})\big|_{\g{\sigma_{i,j}}}}\Tr{(s'_i s'_{i-1})\big|_{\g{\sigma_{i,j}}}}\\
&\prod_{\sigma_{i,j'}\in AC(i,s_{i-1})}\Bigg[\Tr{(s_i s_{i-1})\big|_{\g{\sigma_{i,j'}}}}\Tr{(s'_i s'_{i-1})\big|_{\g{\sigma_{i,j'}}}}\mathbb{E}_{{\theta}_{i,j'}}\left(\cos{\theta_{i,j'}}\right)^2\\
&\quad\quad\quad\quad\quad +\Tr{(is_i \sigma_{i,j'}s_{i-1})\big|_{\g{\sigma_{i,j'}}}}\Tr{(is'_i \sigma_{i,j'}s'_{i-1})\big|_{\g{\sigma_{i,j'}}}}\mathbb{E}_{{\theta}_{i,j'}}\left(\sin\theta_{i,j'}\right)^2\Bigg],
\end{aligned}
\end{equation}
where the last equality is given by $\mathbb{E}_{\theta_{i,j}} \sin{\theta_{i,j}} \cos{\theta_{i,j}}=0$. 

Similarly, Eq.~\eqref{ap:eq:E_cross_i-layer-rewrited} implies that up to sign $\pm$,
\begin{equation}\label{ap:eq:none_R_ele}
  \begin{gathered}
    s_i\big|_{I_i}=s_{i-1}\big|_{I_i} \;,\; s'_i\big|_{I_i}=s'_{i-1}\big|_{I_i}, \; \mathrm{and}\\
    s_i\big|_{\g{V_{i,k}}}=(V_{i,k} s_{i-1}V_{i,k}^\dagger)\big|_{\g{V_{i,k}}}\;,\;s'_i\big|_{\g{V_{i,k}}}=(V_{i,k} s'_{i-1}V_{i,k}^\dagger)\big|_{\g{V_{i,k}}}
  \end{gathered}
\end{equation}
for $k=1,\cdots,C_i$. If not, then $\Tr{s_i\mathcal{U}_i s_{i-1}\mathcal{U}_i^\dagger}\Tr{s'_i\mathcal{U}_i s'_{i-1}\mathcal{U}_i^\dagger}=0$. 

For $j\in\{1,\cdots,R_i\}$ such that $\sigma_{i,j}\in C(i,s_{i-1})$, we have $s_{i}|_{\g{\sigma_{i,j}}}=s_{i-1}|_{\g{\sigma_{i,j}}}$ and $s'_{i}|_{\g{\sigma_{i,j}}}=s'_{i-1}|_{\g{\sigma_{i,j}}}$; otherwise $\Tr{(s_i s_{i-1})\big|_{\g{\sigma_{i,j}}}}\Tr{(s'_i s'_{i-1})\big|_{\g{\sigma_{i,j}}}}=0$. 
For $j'$ being an index such that $\sigma_{i,j'}\in AC(i,s_{i-1})$, we have two cases: 
\begin{enumerate}
\item $s_{i}|_{\g{\sigma_{i,j'}}}=s_{i-1}|_{\g{\sigma_{i,j'}}}$ and $s'_{i}|_{\g{\sigma_{i,j'}}}=s'_{i-1}|_{\g{\sigma_{i,j'}}}$.
\item $s_{i}|_{\g{\sigma_{i,j'}}}= \left(i\sigma_{i,j'} s_{i-1}\right)|_{\g{\sigma_{i,j'}}}$ and $s_{i}|_{\g{\sigma_{i,j'}}}= \left(i\sigma_{i,j'} s'_{i-1}\right)|_{\g{\sigma_{i,j'}}}$ (up to a sign $\pm$). 
\end{enumerate}
If neither of these two cases holds, the equation Eq.~\eqref{ap:eq:E_cross_i-layer-rewrited} is equal to zero.
We denote the product of these $i\sigma_{i,j'}$ acting on $s_{i-1}$ and $s'_{i-1}$ as operator $\mathcal{P}_i$.
Then combined with Eq.~\eqref{ap:eq:none_R_ele}, we obtain $s_i= \mathcal{P}_i \mathcal{V}_i s_{i-1} \mathcal{V}_i^\dagger$, $s_i'=\mathcal{P}_i \mathcal{V}_i s_{i-1}' \mathcal{V}_i^\dagger$ up to sign $\pm$, for $i=1,\cdots ,L$.

Thus, there must be 
\begin{equation}
  \begin{aligned}
    s_{i-1}&=\mathcal{V}_{i}^\dagger \mathcal{P}_{i}\cdots \mathcal{V}_L^\dagger \mathcal{P}_L s_L \mathcal{V}_L\cdots\mathcal{V}_{i},\\
    s'_{i-1}&= \mathcal{V}_{i}^\dagger \mathcal{P}_{i}\cdots \mathcal{V}_L^\dagger \mathcal{P}_L s'_L \mathcal{V}_L\cdots\mathcal{V}_{i}
  \end{aligned}
\end{equation}
up to sign $\pm$. 

As discussed before, $s_{i-1}$ and $s_{i-1}'$ have the same exchange relation with $\sigma_{i,j}$, leads to $\mathcal{P}_{i} \mathcal{V}_{i+1}^\dagger \cdots \mathcal{V}_L^\dagger \mathcal{P}_L s_L \mathcal{V}_L \cdots\mathcal{V}_{i+1}$ and $\mathcal{P}_{i}\mathcal{V}_{i+1}^\dagger \cdots \mathcal{V}_L^\dagger \mathcal{P}_L s'_L \mathcal{V}_L\cdots \mathcal{V}_{i+1}$ have the same exchange relation with $\mathcal{V}_{i} \sigma_{i,j}\mathcal{V}_{i}^\dagger$.
According to Lemma~\ref{ap:lemma:3terms_exchange}, $\mathcal{V}_{i+1}^\dagger \mathcal{P}_{i+1}\cdots \mathcal{V}_L^\dagger \mathcal{P}_L s_L \mathcal{V}_L\cdots\mathcal{V}_{i+1}$ and $\mathcal{V}_{i+1}^\dagger \mathcal{P}_{i+1}\cdots \mathcal{V}_L^\dagger \mathcal{P}_L s'_L \mathcal{V}_L\cdots\mathcal{V}_{i+1}$ have the same exchange relation with $\mathcal{V}_{i}\sigma_{i,j}\mathcal{V}_{i}^\dagger$. 
Repeating this process, we get $s_L$ and $s'_L$ have the same exchange relation with $\mathcal{V}_{L} \cdots \mathcal{V}_{i} \sigma_{i,j} \mathcal{V}_{i}^\dagger\cdots \mathcal{V}_{L}^\dagger=\overline{\sigma}_{i,j}$.
As a result, $s_L$ and $s_L'$ have the same exchange relation with each element in $\{\overline{\sigma}_{i,j}\}$. On the other hand, $s_L$ and $s'_L$ are contained in the Pauli word set $\{\sigma\}$ of $O$, otherwise $\Tr{Os_L}\Tr{Os'_L}=0$. This leads to conclude that $s_L=s_L'$ is due to the split assumption.
\end{proofpart}

\begin{proofpart}
In the second part of the proof, we demonstrate that $s_L=s_L'$ implies $s_i=s_i'$ for $i=0,\cdots,L$, otherwise $\mathbb{E}_{\bm{\theta}}f(\bm{\theta},s,O,\rho)f(\bm{\theta},s^{\prime},O,\rho)=0$.

To prove this claim, it suffices to show that if $\mathbb{E}_{\bm{\theta}}f(\bm{\theta},s,O,\rho)f(\bm{\theta},s^{\prime},O,\rho)\neq 0$, and $s_i=s_i'$, then $s_{i-1}=s_{i-1}'$.
Given that $s_{i}=s_{i}'$ and by $C(i,s_{i-1})=C(i,s_{i})=C(i,s'_{i-1})$ and $AC(i,s_{i-1})=AC(i,s_{i})=AC(i,s'_{i-1})$, Eq.~\eqref{ap:eq:E_cross_i-layer} can again be rewrited as Eq.~\eqref{ap:eq:E_cross_i-layer-rewrited}.

Thus, we have 
\begin{equation}
\begin{gathered}
  s_{i-1}\big|_{I_i}=s'_{i-1}\big|_{I_i}=s_i\big|_{I_i},\\
  s_{i-1}\big|_{\g{V_{i,k}}}=s'_{i-1}\big|_{\g{V_{i,k}}}=(V_{i,k}^\dagger s_i V_{i,k}) \big|_{\g{V_{i,k}}}
\end{gathered}
\end{equation}
up to sign $\pm$ for $k=1,\cdots,C_i$, otherwise $\Tr{s_i\mathcal{U}_i s_{i-1}\mathcal{U}_i^\dagger}\Tr{s'_i\mathcal{U}_i s'_{i-1}\mathcal{U}_i^\dagger}=0$. 

Suppose $s_{i-1}\big|_{\g{\sigma_{i,j}}}\neq s'_{i-1}\big|_{\g{\sigma_{i,j}}}$ for some $\sigma_{i,j}\in C(i,s_{i})$, then either $\Tr{(s_i s_{i-1})\big|_{\g{\sigma_{i,j}}}}$ or $\Tr{(s_i' s_{i-1}')\big|_{\g{\sigma_{i,j}}}}$ is zero, resulting in Eq.~\eqref{ap:eq:E_cross_i-layer-rewrited} being $0$.

Similarly, suppose there is $s_{i-1}\big|_{\g{\sigma_{i,j'}}}\neq s'_{i-1}\big|_{\g{\sigma_{i,j'}}}$ for some $\sigma_{i,j'}\in AC(i,s_{i})$. 
If this holds, then both equations $\Tr{(s_i s_{i-1})\big|_{\g{\sigma_{i,j'}}}}\Tr{(s'_i s'_{i-1})\big|_{\g{\sigma_{i,j'}}}}$ and $\Tr{(is_i \sigma_{i,j'}s_{i-1})\big|_{\g{\sigma_{i,j'}}}}\Tr{(is'_i \sigma_{i,j'}s'_{i-1})\big|_{\g{\sigma_{i,j'}}}}$ are zero, resulting in Eq.~\eqref{ap:eq:E_cross_i-layer-rewrited} equal to $0$.

Based on the above discussion, when $s_{i}=s_{i}'$ we must have $s_{i-1}=s_{i-1}'$ for $i=1,\cdots,L$, otherwise there is $\mathbb{E}_{\bm{\theta}}f(\bm{\theta},s,O,\rho)f(\bm{\theta},s^{\prime},O,\rho)=0$.
This finished the proof of the second part.
\end{proofpart}
\end{proof}
\end{lemma}

\begin{remark}
Specifically, for the case of $O$ has only one Pauli word $\sigma$, the set of Pauli words $\{\overline{\sigma}_{i,j}\}$ naturally splits the Pauli word set $\{\sigma\}$ of $O$. This means that our requirements for the circuit are fundamentally different from anti-concentration.
\end{remark}

If $\{\overline{\sigma}_{i,j}\}$ can split $\{\mathbb{I},X,Y,Z\}^{\otimes n}$, it obviously can split $\{\sigma\}$.
We will use a lemma to explain the equivalence between $\{\overline{\sigma}_{i,j}\}$ split $\{\mathbb{I},X,Y,Z\}^{\otimes n}$ and $\{\overline{\sigma}_{i,j}\}$ generates $\{\mathbb{I},X,Y,Z\}^{\otimes n}$ up to phase.
Before presenting the lemma, we must clarify the definition of a Pauli set $A$ that generates $\{\mathbb{I},X,Y,Z\}^{\otimes n}$. We say a Pauli set $A$ generates $\{\mathbb{I},X,Y,Z\}^{\otimes n}$ up to phase means that $\langle A \rangle/\left( \langle A \rangle\cap\langle i\mathbb{I}^{\otimes n} \rangle \right) = \mathbb {P}^n$, where $\mathbb{P}_n:=\mathrm{PG}_n/\langle i\mathbb{I}^{\otimes n}\rangle$ and $\mathrm{PG}_n$ is the $n$-qubit Pauli group.
In this expression, the notation $\langle A \rangle$ refers to the Pauli subgroup that is generated by set $A$ (the finite product of elements and their inverses in $A$). And the quotient is used to remove the effect of the phase factor.

Essentially, this representation means that $A$ generates $\{\mathbb{I},X,Y,Z\}^{\otimes n}$ up to phase $\{e^{i\psi}|\psi=0,\frac{\pi}{2},\pi,\frac{3\pi}{2}\}$.

\begin{lemma}\label{ap:lemma:split_generate}
$A$ splits $\{\mathbb{I},X,Y,Z\}^{\otimes n}$ if and only if $\langle A\rangle/\left(\langle A\rangle\cap\langle i\mathbb{I}^{\otimes n}\rangle\right)=\mathbb{P}^n$, where $\mathbb{P}_n:=\mathrm{PG}_n/\langle i\mathbb{I}^{\otimes n}\rangle$.
\end{lemma}
\begin{proof}
Suppose $\langle A\rangle/\left(\langle A\rangle\cap\langle i\mathbb{I}^{\otimes n}\rangle\right)=\mathbb{P}^n$. Take $b_1,b_2\in \{\mathbb{I},X,Y,Z\}^{\otimes n}$ such that $b_1$ and $b_2$ have the same exchange relation with every $a\in A$. Since $\left[b_1b_2,a\right]=0$ for all $a\in A$, we have $\left[b_1b_2,g\right]=0$ for all $g\in\langle i\mathbb{I}^{\otimes n},A\rangle=\mathrm{PG}_n$, which implies $b_1b_2\in\langle i\mathbb{I}^{\otimes n}\rangle$. Combined with $b_1,b_2\in \{\mathbb{I},X,Y,Z\}^{\otimes n}$, we have $b_1=b_2$.

On the other hand, suppose $A$ splits $\{\mathbb{I},X,Y,Z\}^{\otimes n}$, the goal is to prove that $\langle A\rangle/\left(\langle A\rangle\cap\langle i\mathbb{I}^{\otimes n}\rangle\right)=\mathbb{P}^n$.
To prove this claim, it suffices to show that if $\langle A\rangle/\left(\langle A\rangle\cap\langle i\mathbb{I}^{\otimes n}\rangle\right)\neq\mathbb{P}^n$, then there exist a non-identity Pauli word commutes with each element of $A$.

We set $C=\langle A\rangle/\left(\langle A\rangle\cap\langle i\mathbb{I}^{\otimes n}\rangle\right)$ and then consider the $C^*$-algebras $\mathcal{M}$ and $\mathcal{P}$ generated by $C$ and $\mathbb{P}^n$, respectively. 
By Von Neumann bicommutant theorem \cite{kadison1986fundamentals}, we have $\mathcal{M}''=\overline{\mathcal{M}}=\mathcal{M}$. 
Since $C\neq\mathbb{P}^n$, we can conclude that $\mathcal{M}\neq \mathcal{P}$ as Pauli words constitute an orthonormal basis of the matrix algebra, resulting in $\mathcal{M}''\neq \mathcal{P}$. 
This implies the existence of non-identity elements in $\mathcal{M}'$. In other words, there exists a non-trivial element $x=c_1P_1+c_2P_2+\cdots$ ($P_i$ stands for different Pauli words and $c_i\in \mathbb{C}$) which commutes with every element in $A$. This leads to the conclusion that $\{P_1,P_2,\cdots\}$ also commute with every element in $A$ and they cannot all be identical.
\end{proof}

This finished the proof of Lemma~\ref{lemma:cross_items}. In addition, an equivalent proof can be found in \cite{gao2018efficient}.

\section{Proof of Theorem~\ref{thm:main}}\label{ap:thm1}

\begin{proof}

  From Lemma~\ref{lemma:MSE_l}, we have shown that if $M\geq\frac{1}{4\gamma}\ln{\frac{\norm{O}_\infty^2}{\nu}}$, then
 \begin{equation}
  \mathbb{E}_{\bm{\theta}}\abs{\widetilde{\mathcal{L}}-\widehat{\mathcal{L}}}^2 \leq \nu.
 \end{equation}
  
 By Markov's inequality, for any probability $\delta$, we have
 \begin{equation}
  \Pr{\abs{\widetilde{\mathcal{L}}-\widehat{\mathcal{L}}}\geq\frac{1}{\sqrt{\delta}}\sqrt{\mathbb{E}_{\bm{\theta}}\abs{\widetilde{\mathcal{L}}-\widehat{\mathcal{L}}}^2}}=\Pr{\abs{\widetilde{\mathcal{L}}-\widehat{\mathcal{L}}}^2\geq\frac{1}{\delta}{\mathbb{E}_{\bm{\theta}}\abs{\widetilde{\mathcal{L}}-\widehat{\mathcal{L}}}^2}}\leq \delta.
 \end{equation}

Therefore, with probability at least $1-\delta$ over parameters $\bm{\theta
}$ , there is
\begin{equation}
  \abs{\widetilde{\mathcal{L}}-\widehat{\mathcal{L}}} \leq \frac{1}{\sqrt{\delta}}\sqrt{\mathbb{E}_{\bm{\theta}}\abs{\widetilde{\mathcal{L}}-\widehat{\mathcal{L}}}^2}\leq \sqrt{\frac{\nu}{\delta}}.
\end{equation}
Let $\varepsilon$ be the desired error, we can set $\nu=\varepsilon^2 \delta$ to meet the requirements, thus the truncation weight $M$ needs to satisfy
\begin{equation}
  M\geq \frac{1}{2\gamma} \ln{\frac{\norm{O}_\infty}{\varepsilon \sqrt{\delta}}}.
\end{equation}

Using the algorithm illustrated in Supplement Material~\ref{ap:algorithm}, we can obtain the approximate observable value $\widetilde{\mathcal{L}}$ for $M=\frac{1}{2\gamma} \ln{\frac{\norm{O}_\infty}{\varepsilon \sqrt{\delta}}}$.
The time complexity for obtaining $\widetilde{\mathcal{L}}$ for \textbf{Case 1} is:
\begin{equation}
\begin{aligned}
\mathrm{Poly}(n) \order{L} 2^{M}
&=\mathrm{Poly}(n) \order{L} \bigg(\frac{\norm{O}_\infty}{\varepsilon \sqrt{\delta}} \bigg)^{\order{\frac{1}{\gamma}}}\\
&=\mathrm{Poly}\left(n,L,\frac{1}{\varepsilon},\frac{1}{\sqrt{\delta}},\norm{O}_{\infty}\right).
\end{aligned}
\end{equation}

While dealing with \textbf{Case 2}, the required time complexity is about:
\begin{equation}
  \mathrm{Poly}(n)\order{(nL)^{M+1}}
  =
  \mathrm{Poly}(n)  \order{(nL)^{\frac{1}{2\gamma} \ln{\frac{\norm{O}_\infty}{\varepsilon \sqrt{\delta}}}+1}}.
\end{equation}

The relationship among the variables $n$, $L$, $\varepsilon$, $\delta$, and $\norm{O}_\infty$ is non-polynomial. Therefore, polynomial relations are inapplicable in this scenario, and only a quasi-polynomial relationship can be derived. 
Specifically, with a given relative error $\frac{\varepsilon}{\norm{O}_\infty}$ and probability of success $\delta$, the time complexity is limited by $\mathrm{Poly}(n)  (nL)^{\order{\frac{1}{\gamma}}}$, indicating that the computational complexity maintains a polynomial relationship with the circuit size when the required precision is held constant.

In terms of space complexity, by Eq.~\eqref{ap:eq:space_cost} and Proposition~\ref{prop:f_ele}, the required space complexity is $\order{\mathrm{Poly}(n)+nL}$.

This finished the proof of Theorem~\ref{thm:main}.
\end{proof}

\section{Proof of Proposition~\ref{prop:lambda_and_L}}
The Proposition~\ref{prop:lambda_and_L} discussed two situations for $\gamma=\Omega(\frac{1}{\log L})$ and $\gamma=\order{\frac{1}{L}}$ in the Case 1.

\begin{proofcase}
For $\gamma=\Omega(\frac{1}{\log L})$, we need to calculate $\widetilde{\mathcal{L}}$ with the MSE $\mathbb{E}_{\bm{\theta}}\abs{\widetilde{\mathcal{L}}-\widehat{\mathcal{L}}}^2$ less than a sufficiently small constant $c$.
By lemma~\ref{lemma:MSE_l}, we have $\mathbb{E}_{\bm{\theta}}\abs{\widetilde{\mathcal{L}}-\widehat{\mathcal{L}}}^2  \leq c$ can be satisfied when
 \begin{equation}
   M\geq \frac{1}{4\gamma}  \ln{\frac{\norm{O}_\infty^2}{c}} \sim \frac{1}{\gamma}.
 \end{equation}

By setting $M\sim\frac{1}{\gamma}$, the total runtime for obtaining observable $\widetilde{\mathcal{L}}$ is 
\begin{equation}
  \begin{aligned}
    \mathrm{Poly}(n) \order{L} 2^{M}
  =&\mathrm{Poly}(n,L)  2^{\order{\frac{1}{\gamma}}}\\
  =&\mathrm{Poly}(n,L)  L^{\order{1}}\\
  =&\mathrm{Poly}(n,L).
  \end{aligned}
\end{equation}
\end{proofcase}

\begin{proofcase}
For $\gamma=\order{\frac{1}{L}}$, we will construct a specific example, under which our method will have to incur exponential time cost with respect to $L$ in order to achieve a sufficiently small MSE.

We consider a special VQA algorithm, the ans\"atz consists of a layer of $R_Z$ gates acted on each qubit, a layer of $R_X$ gates acted on each qubit and $L-2$ layers $R_X$ gates acted on the first qubit, shown in Fig.~\ref{ap:prop:fig:ansatz}. 
The initial state is set as $\rho=\ketbra{0}{0}^{\otimes n}$, and the observable $O=Z_1+Y_1$, the cost function is defined as Eq.~\eqref{ap:eq:define_L}. The noise channel is set as depolarizing channel $p_x = p_y = p_z = \frac{\lambda}{4}$ and we get $\lambda=\order{\frac{1}{L}}$

\begin{figure}[tbp]
  \includegraphics[width=0.48\textwidth]{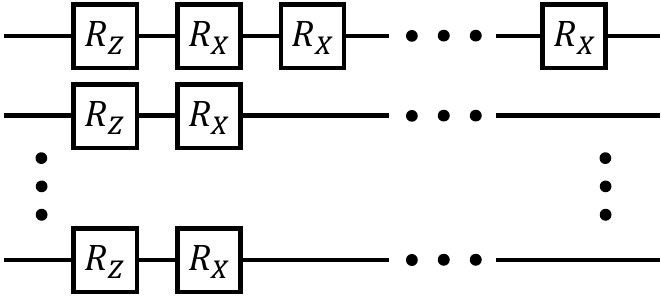}
  \caption{Ans\"atz: There is a layer of $R_Z$ gates acted on each qubit, a layer of $R_X$ gates acted on each qubit and $L-2$ layers $R_X$ gates acted on the first qubit.}\label{ap:prop:fig:ansatz}
  \end{figure}

If we truncate noisy cost function for $\abs{s}\leq L$, the approximate cost function can be expressed as:
\begin{equation}
\widetilde{\mathcal{L'}}(\bm{\theta})=\sum_{\abs{s}\leq L} \hat{f}(\bm{\theta},s,O,\rho)=\sum_{m=0}^L (1-\lambda)^m \sum_{\abs{s}=m} f(\bm{\theta},s,O,\rho).
\end{equation}

Before considering the difference between $\widetilde{\mathcal{L'}}$ and $\widehat{\mathcal{L}}$, we first consider the noiseless situation. The unitary on the first qubit can be expressed as :
\begin{equation}
    U(\mathbb{\theta})|_1=\exp{-i\frac{\theta_{2,1}+\cdots+\theta_{L,1} }{2}X_1}\exp{-i\frac{\theta_{1,1}}{2}Z_1}.
\end{equation}
We denote $\alpha=\frac{\theta_{2,1}+\cdots+\theta_{L,1} }{2}$, the noiseless cost function can be expressed as:
\begin{equation}\label{ap:eq:noiseless_cost}
\begin{aligned}
    \mathcal{L}(\mathbb{\theta})&=\bra{0}\exp{-i\frac{\theta_{1,1}}{2}Z_1}^\dagger\exp{-i \alpha X_1}^\dagger (Z_1+Y_1) \exp{-i \alpha X_1}\exp{-i\frac{\theta_{1,1}}{2}Z_1}\ket{0}\\
    &=\cos{2\alpha}-\sin{2\alpha}.
\end{aligned}
\end{equation}

Note that in Supplement Material~\ref{ap:algorithm}, we have discussed for Pauli path $s$ with non-zero contribution, must have $\abs{s_i} > 0$ for $i=0,\cdots,L$. Thus $\abs{s}\geq L+1$ is required.

Conversely, when $\abs{s}> L+1$, there exist a qubit $k\neq 1$ such that $s_{i'}|_k$ is not identity for some $i'$, leads to $f(\bm{\theta},s,O,\rho)=0$ by $O|_k=\mathbb{I}$. Thus, the Pauli paths $s$ with non-zero contributions to $\mathcal{L}(\theta)$ must obey $\abs{s}=L+1$.
As a result, we can conclude that
\begin{equation}\label{ap:eq:noiseless_cost_2}
\begin{aligned}
\mathcal{L}(\mathbb{\theta})=\sum_{s\in \bm{P}^{L+1}_n} f(\bm{\theta},s,O,\rho)=\sum_{\abs{s}=L+1} f(\bm{\theta},s,O,\rho).
\end{aligned}
\end{equation}

Combine Eq.~\eqref{ap:eq:noiseless_cost} and Eq.~\eqref{ap:eq:noiseless_cost_2}, we know
\begin{equation}
    \mathbb{E}_{\bm{\theta}}\mathcal{L}(\mathbb{\theta})^2=\mathbb{E}_{\bm{\theta}} \left[\sum_{\abs{s}=L+1} f(\bm{\theta},s,O,\rho) \right]^2=\mathbb{E}_{\alpha}(\cos{2\alpha}-\sin{2\alpha})^2=\mathbb{E}_{\alpha}[1-\sin4\alpha]=1.
\end{equation}
The last equality in the above equation can be verified as follows. Since $\alpha$ follows a generalized Irwin-Hall distribution, and its characteristic function $\varphi_{\alpha}(t)=\mathbb{E}[e^{it\alpha}]$ can be expressed as $\left(\frac{e^{i\frac{\pi}{2}t}-e^{-i\frac{\pi}{2}t}}{i\pi t}\right)^{L-1}$, we have
\begin{equation}
\mathbb{E}_{\alpha}[\sin4\alpha]=\textrm{Im }\mathbb{E}[e^{i4\alpha}]=\textrm{Im }\varphi_{\alpha}(4)=\textrm{Im }\left(\frac{e^{2\pi i}-e^{-2\pi i}}{4\pi i}\right)^{L-1}=0.
\end{equation}

The MSE between $\widetilde{\mathcal{L'}}$ and $\widehat{\mathcal{L}}$ can be estimated as 
\begin{equation}
\begin{aligned}
\mathbb{E}_{\bm{\theta}}\abs{\widetilde{\mathcal{L'}}-\widehat{\mathcal{L}}}^2 
&=\mathbb{E}_{\bm{\theta}}\left[\sum_{\abs{s}> L}\hat{f}(\bm{\theta},s,O,\rho)\right]^2\\
&=\mathbb{E}_{\bm{\theta}}\left[\sum_{\abs{s}=L+1}\hat{f}(\bm{\theta},s,O,\rho)\right]^2\\
&= (1-\lambda)^{2(L+1)}\mathbb{E}_{\bm{\theta}} \left[\sum_{\abs{s}=L+1}{f}(\bm{\theta},s,O,\rho)\right]^2\\
&=(1-\lambda)^{2(L+1)}.
\end{aligned}
\end{equation}

By Bernoulli's inequality, for $r\geq 1$ and $x\geq -1$, we have
\begin{equation}
(1+x)^r\geq 1+rx.
\end{equation}

Owing to $\lambda=\order{\frac{1}{L}}$, there is a constant $c$ to make $c\geq \lambda  (L+1)$. Therefore, we have
\begin{equation}
(1-\lambda)^{2(L+1)}= (1-\lambda)^{\frac{2(L+1)}{4c} 4c} \geq \left(1-\lambda \frac{2(L+1)}{4c} \right)^{4c}\geq \left(\frac{1}{2}\right)^{4c},
\end{equation}
leads to
\begin{equation}
    \mathbb{E}_{\bm{\theta}}\abs{\widetilde{\mathcal{L'}}-\widehat{\mathcal{L}}}^2 =\Omega(1).
\end{equation}

So the truncation $\abs{s}\leq L$ is not enough.
While the number of Pauli paths with non-zero contributions and weight $\abs{s}=L+1$ is about $2^{L-1}$ ($s_0=Z_1,s_1=Z_1,s_2=Z_1 \mathrm{or} \;Y_1,\cdots,s_L=Z_1 \mathrm{or} \;Y_1$).
This leads to exponential complexity about $L$.
\end{proofcase}

\section{Numerical: Computational cost analysis}
\label{sec:ap:computationalanalysis}

\subsection{The relationship between computational cost and variables $n$,$L$ and $M$}

The computational cost of our method is positively correlated to the number of Pauli paths that have a non-zero contribution.
We performed a numerical analysis on the quantity of Pauli paths with non-zero contributions, which serves as a measure of its computational complexity.

\begin{figure}[htbp]	\includegraphics[width=0.5\textwidth]{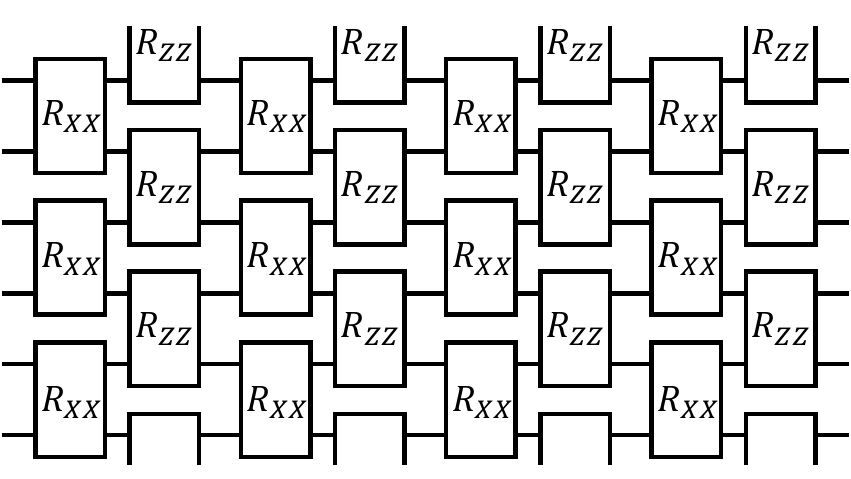}
	\caption{The ans\"atz used in the simulation consists of interleaved $R_{XX}$ and $R_{ZZ}$ Pauli rotating gate layers, where the incomplete $R_{ZZ}$ gates on the first qubit signify that they operate on the first and the last qubit.}\label{fig:XX_ZZ_Ansatz}
\end{figure}

For a given qubit number $n$, we randomly generate a $n\times n$ adjacency $(0,1)$-matrix $A$, in which the probability of each entry being $1$ is $0.5$. The observable $O$ is constructed based on the MaxCut problem and is related to the adjacency matrix as $O=\sum_{A_{i,j}=1} Z_iZ_j$. The initial state is set as $\rho=\ketbra{0}{0}^{\otimes n}$, and the cost function is defined by Eq.~\eqref{ap:eq:define_L}.
The ans\"atz used, shown in Fig.~\ref{fig:XX_ZZ_Ansatz}, comprises Pauli rotation gates including $R_{XX}$ and $R_{ZZ}$.

\begin{figure}[htbp]
  \includegraphics[width=\textwidth]{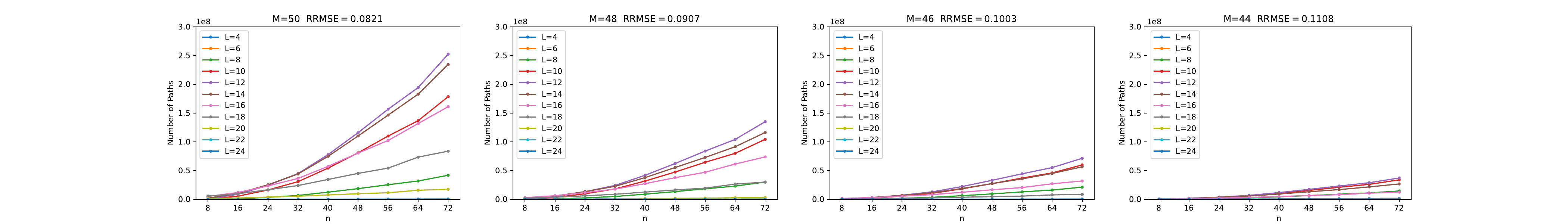}
  \includegraphics[width=\textwidth]{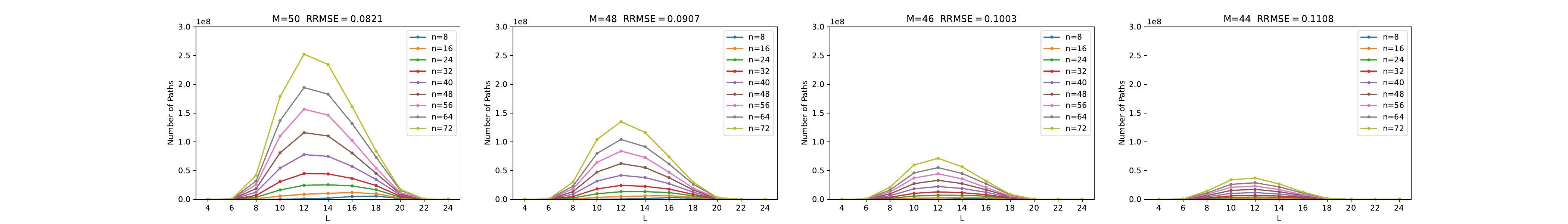}
  \caption{Numerical results for the number of Pauli paths $s$ with non-zero contributions and $\abs{s}\leq M$. The NRMSE bounds are given by relation established in Lemma~\ref{lemma:MSE_l} under noise rate $p_x = p_y = p_z = \frac{\lambda}{4}$, where $\lambda=0.2$. Here, $n$ represents the qubit number and $L$ represents the depth of the circuit. The two figures in each column represent different presentations of the same data. The top and bottom figures show the trend of the number of paths with $n$ and $L$, respectively.}\label{fig:numerical_cost}
\end{figure}

For different $n$ and $L$, we calculate the number of Pauli paths $s$ with weight $\abs{s}\leq M$ and non-zero contribution, shown in Fig.~\ref{fig:numerical_cost}. Our numerical findings reveal that the number of Pauli paths is significantly smaller than the upper bound $\mathrm{Poly}(n)2^M$, as suggested by the theoretical analysis.
In the case of a fixed $M$, as $n$ increases, the number of non-zero contributing Pauli paths will increase in a moderate manner.

By Lemma~\ref{lemma:MSE_l}, it is worth noting that $M$ corresponds to a Normalized Root-Mean-Square Error bound (NRMSE) given by $\frac{\sqrt{\mathbb{E}_{\bm{\theta}}\abs{\widetilde{\mathcal{L}}-\widehat{\mathcal{L}}}^2}}{\norm{O}_\infty}$ under a fixed noise rate. Setting $p_x = p_y = p_z = \frac{\lambda}{4}$ and $\lambda=0.2$, we observe that for a given NRMSE bound, the number of paths first increases and then decreases with the increase of $L$, which aligns with the findings in Ref.~\cite{noh2020efficient}.

On the other hand, $M$ also corresponds to a noise rate $\lambda$ under a fixed NRMSE bound by Eq.~\ref{ap:eq:depolarizing_M}. With NRMSE bound $0.0821$, $M$ takes the values $M=50,48,46,44$ for $\lambda=0.2,0.21,0.22,0.23$ (two significant digits), respectively. From this, it can be seen that the noise rate has a significant impact on computational complexity.

\subsection{The relationship between $M$ and circuit scale}

In this section, we investigate the varying circuit scale requirements for the truncation number $M$ in the context of achieving dependable accuracy. In this section the noise channel is set as depolarizing channel $p_x = p_y = p_z = \frac{\lambda}{4}$.

Similar to the previous section, we also take the observable derived from the MaxCut problem as an example to numerically analyze the relationship between $M$ and the circuit scale. 
For a given qubit number $n$, we randomly generate a $n\times n$ adjacency $(0,1)$-matrix, in which the probability of each entry being $1$ is $0.5$. The observable is $O=\sum_{A_{i,j}=1} Z_iZ_j$. The initial state is set as $\rho=\ketbra{0}{0}^{\otimes n}$.
The ans\"atz used, is an example of hardware-efficient ans\"atz~\cite{kandala2017hardware}, shown in Fig.~\ref{fig:XX_ZZ_Ansatz}. The depth of the circuit is $4\times D+2$, and the number of parameters is $2\times D+2$.

\begin{figure}[htbp]	\includegraphics[width=0.5\textwidth]{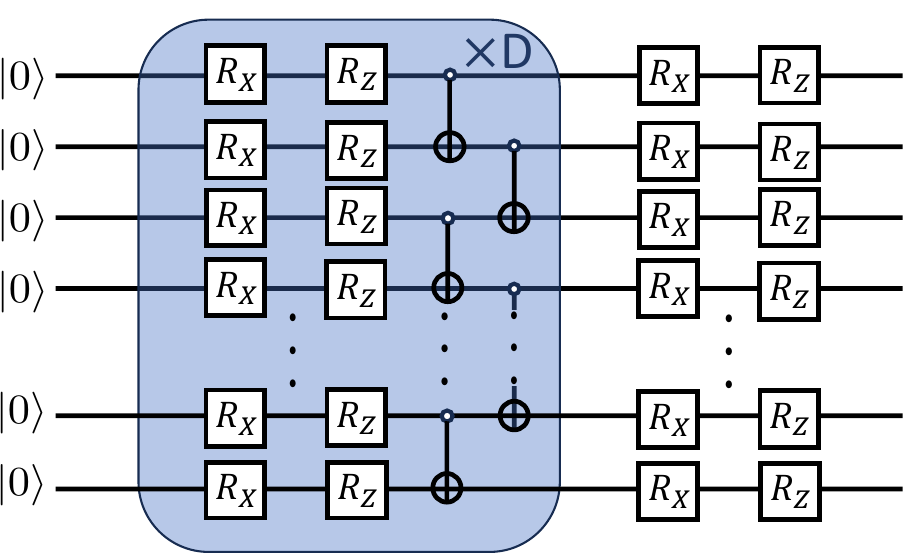}
	\caption{The hardware-efficient ans\"atz employed consists of a series of single-qubit rotation and entangling 2-qubit gates. 
  The variable representing the number of repetitions for the single-qubit rotations and two-qubit gates (CNOT gates) is denoted as $D$.
  Notably, the application of 2-qubit gates is limited to neighboring qubits, a design choice tailored for NISQ devices.}\label{fig:HEA_Ansatz}
\end{figure}

We randomly generate $100$ parameters, denoted as $\bm{\theta^1}\cdots \bm{\theta^{100}}$, and find the minimal truncation number $M$ that satisfies the relative difference between the approximate noisy cost function $\widetilde{\mathcal{L}}(\bm{\theta})$ and the exact noisy cost function $\widehat{\mathcal{L}}(\bm{\theta})$ is below $0.1\%$ for all the $100$ points, formulated as $\max_{i}\abs{\widetilde{\mathcal{L}}(\bm{\theta^i})-\widehat{\mathcal{L}}(\bm{\theta^i})}\leq 0.001 \abs{\widehat{\mathcal{L}}(\bm{\theta^i})}$.

\begin{figure}[htbp]
  \includegraphics[width=0.4\textwidth]{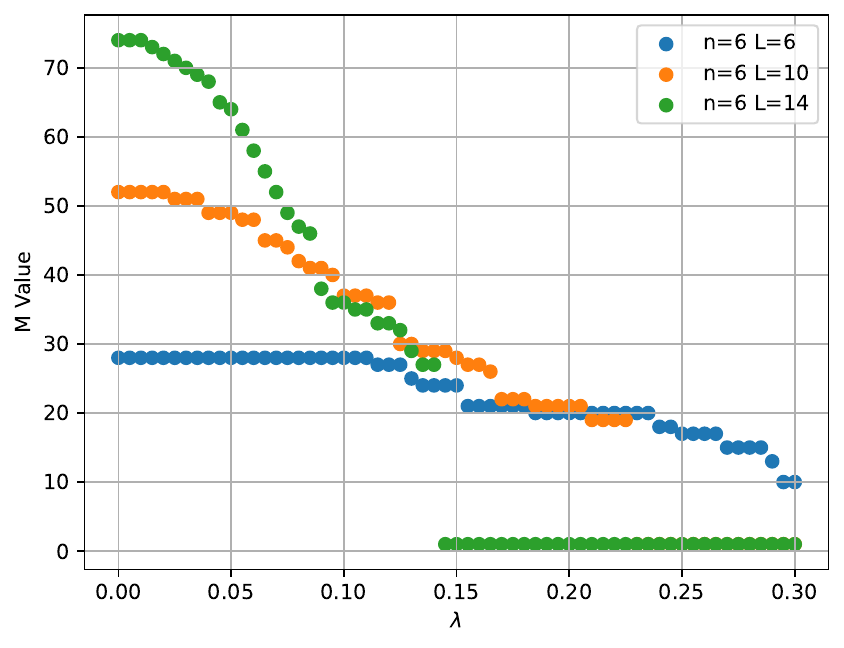}
  \includegraphics[width=0.4\textwidth]{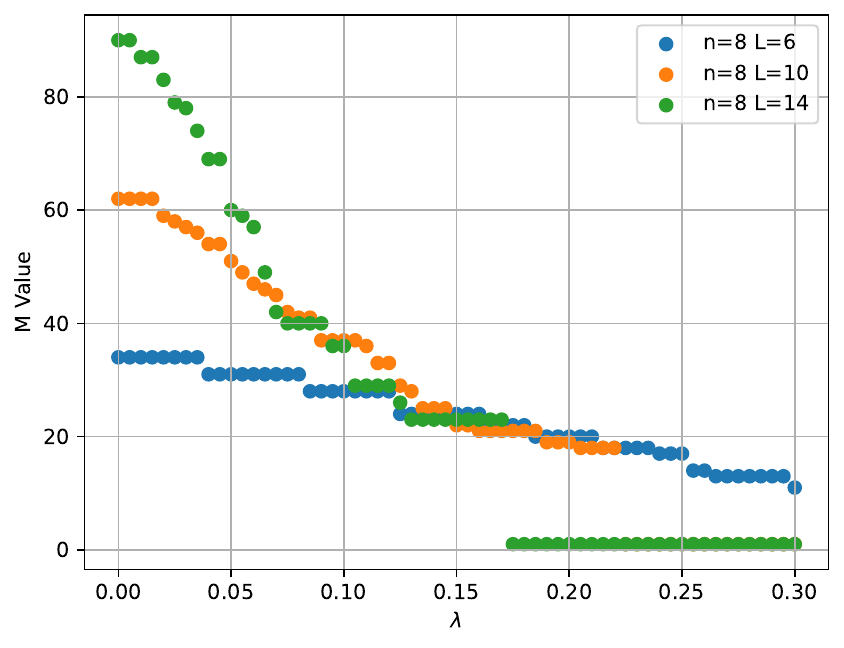}
  \includegraphics[width=0.4\textwidth]{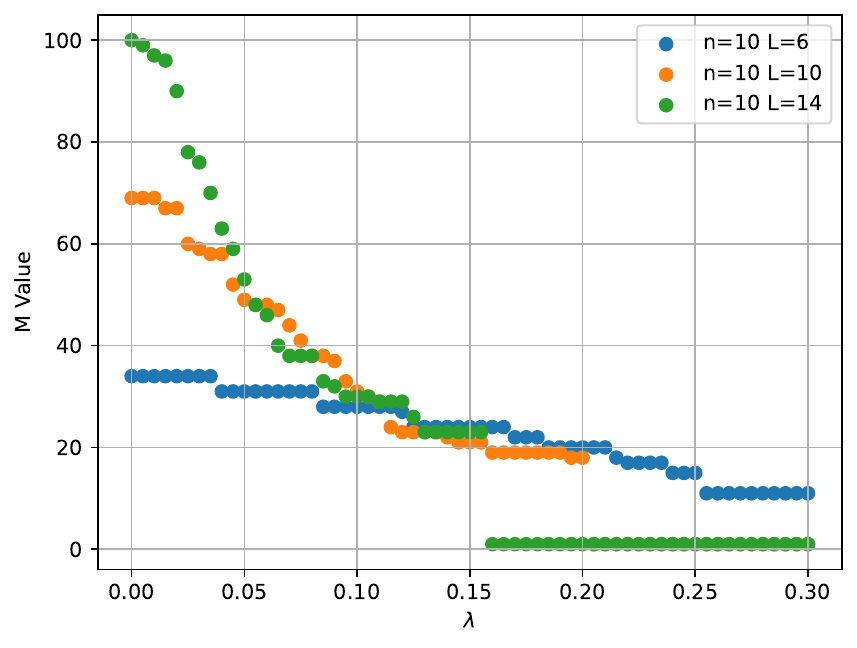}
  \includegraphics[width=0.4\textwidth]{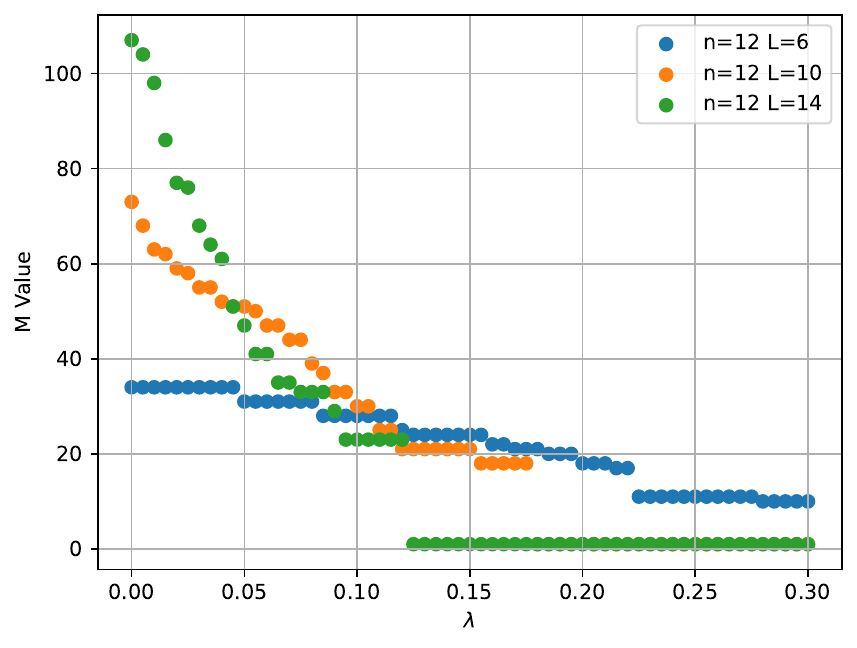}
  \caption{The numerical results for the minimal truncation $M$, satisfying the condition that the difference between approximate and exact expected values $\max_{i}\abs{\widetilde{\mathcal{L}}(\bm{\theta^i})-\widehat{\mathcal{L}}(\bm{\theta^i})}/\abs{\widehat{\mathcal{L}}(\bm{\theta^i})}\leq 0.001$ for $100$ randomly picked parameters $\bm{\theta^1}\dots \bm{\theta^{100}}$. 
  The minimum truncation value $M$ is compared across various circuit scales and noise rates.}\label{fig:numerical_M}
\end{figure}
The results are shown in Fig.~\ref{fig:numerical_M}.
First, we observe that the minimal truncation $M$ decreases with the noise rate $\lambda$ in all cases. This trend is consistent with the theoretical analysis in Lemma~\ref{lemma:MSE_l}.
The numerical results reveal that as the noise rate remains at a low level, the minimum truncation value $M$ escalates linearly with the circuit depth $L$.
This trend arises due to the minor influence of noise on Pauli paths, resulting in an increased propagation distance of Pauli paths within shallow circuits as the circuits become deeper, consequently elevating their Hamming weight. Consequently, a heightened requirement for the truncation value $M$ is observed in deep circuits to uphold the relevant Pauli path in shallow circuits.

When noise levels escalate, the depth of the circuit plays a crucial role in determining the rate at which the minimum truncation value $M$ decreases. This phenomenon occurs due to the accumulation of noise intensifying as the circuit depth increases, leading to a swifter convergence of the expected value $\widehat{\mathcal{L}}(\bm{\theta})$ to $\tr{O}$ (the contribution of the fully-identify travail Pauli path). Consequently, as the circuit depth deepens, the minimum truncation value $M$ converges more rapidly towards $0$.

In addition, as the number of qubits $n$ increases, the minimum truncation value $M$ also tends to increase.
Nevertheless, this escalation is modest, indicating that the OBPPP approach is not unduly constrained by the scale of qubit number $n$.

The investigation also explored the correlation between accuracy and the minimal truncation number $M$ within a specified circuit scale. When $n = 12$, we identified the minimum value of $M$ needed to achieve varying levels of accuracy, where $\max_{i} \sfrac{\left| \widetilde{\mathcal{L}}(\bm{\theta^i}) - \widehat{\mathcal{L}}(\bm{\theta^i}) \right|} {\abs{\widehat{\mathcal{L}}(\bm{\theta^i})}} \leq 0.1, 0.05, 0.01, 0.001, 0.0001$ and $0.00001$.
The results are depicted in Fig.~\ref{fig:numerical_ACC}, which demonstrates that the minimal truncation value $M$ increases with the required accuracy. 

\begin{figure}[htbp]
  \includegraphics[width=0.3\textwidth]{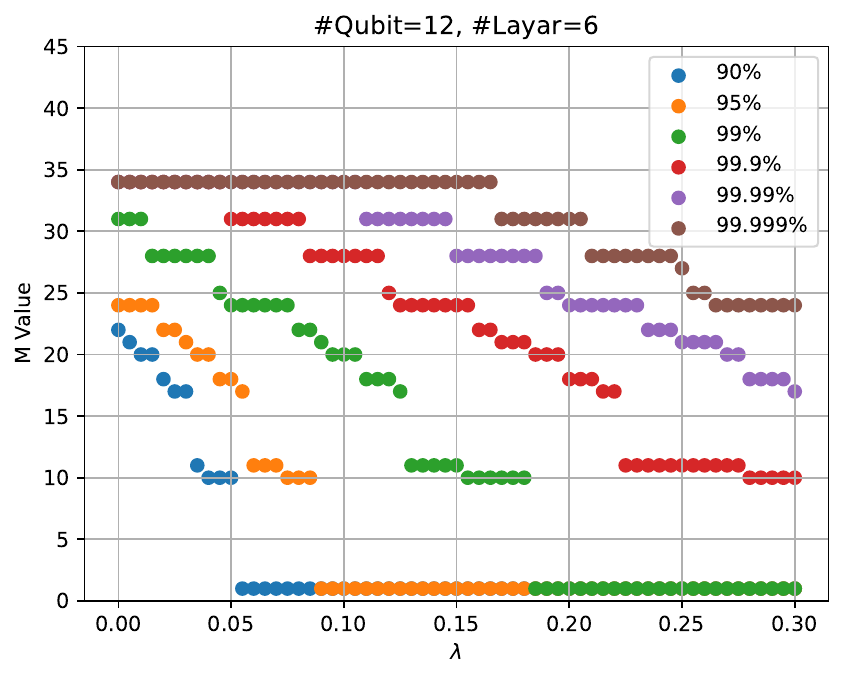}
  \includegraphics[width=0.3\textwidth]{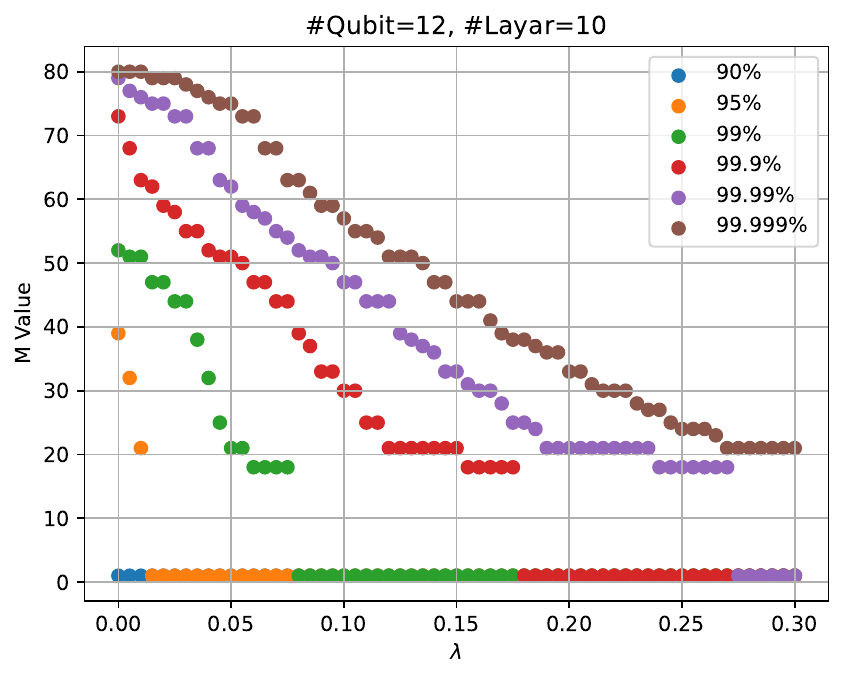}
  \includegraphics[width=0.3\textwidth]{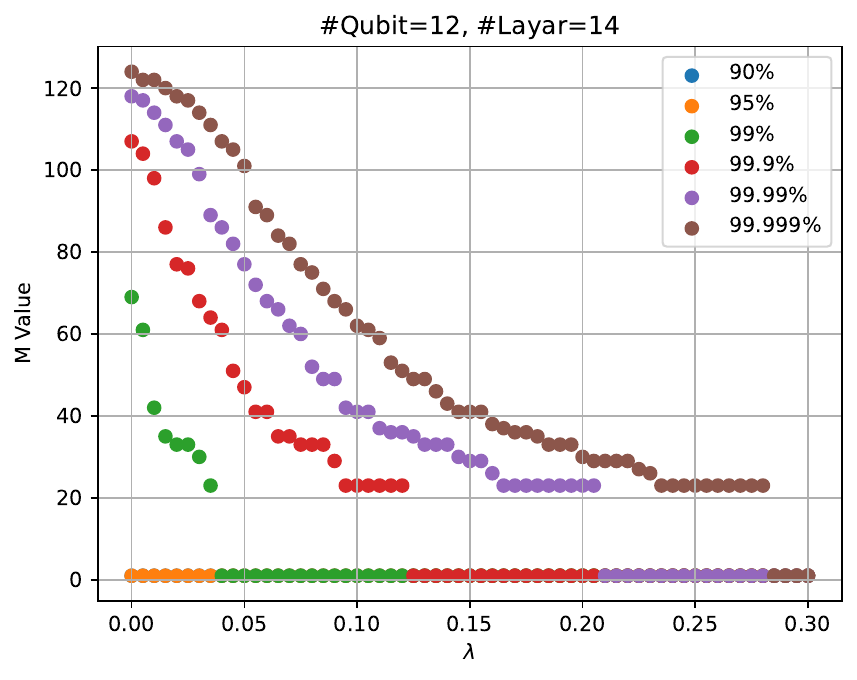}
  \caption{
    The numerical results for the minimum truncation $M$, are presented with respect to varying levels of accuracy. $\max_{i} \abs{\widetilde{\mathcal{L}}(\bm{\theta^i})-\widehat{\mathcal{L}}(\bm{\theta^i})}/\abs{\widehat{\mathcal{L}}(\bm{\theta^i})}\leq 0.1 ,\dots, 0.00001$ corresponding to $90\%,\dots,99.999\%$ for $100$ randomly picked parameters $\bm{\theta^1}\dots \bm{\theta^{100}}$.}\label{fig:numerical_ACC}
\end{figure}
From the numerical results, it can be observed that when the demand for accuracy rapidly increases, the increase of the demand for $M$ is relatively moderate. 
In the scenario where $L=6$, the minimum $M$ needed to achieve an accuracy of $>99.9\%$ is approximately $34$ with $\lambda<0.05$. It is worth noting that the outcomes corresponding to $99.9\%$ and $99.99\%$ accuracy are obscured by the $99.999\%$ result displayed in the Figure.
When the noise rate $\lambda$ approaches $0$, the $M$ required for high accuracy~($\geq 99.9\%$) converges to the nearly same location, suggesting that at this point, $M$ is sufficiently large to encompass an adequate number of Pauli paths in the computation.

When the noise rate $\lambda$ increases, the $M$ required for low accuracy will decrease faster compared to high accuracy, and this trend is consistent with Lemma~\ref{lemma:MSE_l}.

For the minimal $M$, we empirically estimate the mean-square error $\tilde{\nu}$ using the formula:
\begin{equation}\label{ap:eq:emse}
  \tilde{\nu}=\frac{1}{100}\sum_{i=1}^{100}\abs{\widetilde{\mathcal{L}}(\bm{\theta^i})-\widehat{\mathcal{L}}(\bm{\theta^i})}^2,
\end{equation}
and then compare $M$ with the analytical bound given by Eq.~\ref{ap:eq:depolarizing_M}:
\begin{equation}\label{ap:eq:analytical_bound}
  \frac{1}{2\lambda}\ln{\frac{\norm{O}_\infty^2}{\tilde{\nu}}}.
\end{equation}

\begin{figure}[htbp]
  \includegraphics[width=0.3\textwidth]{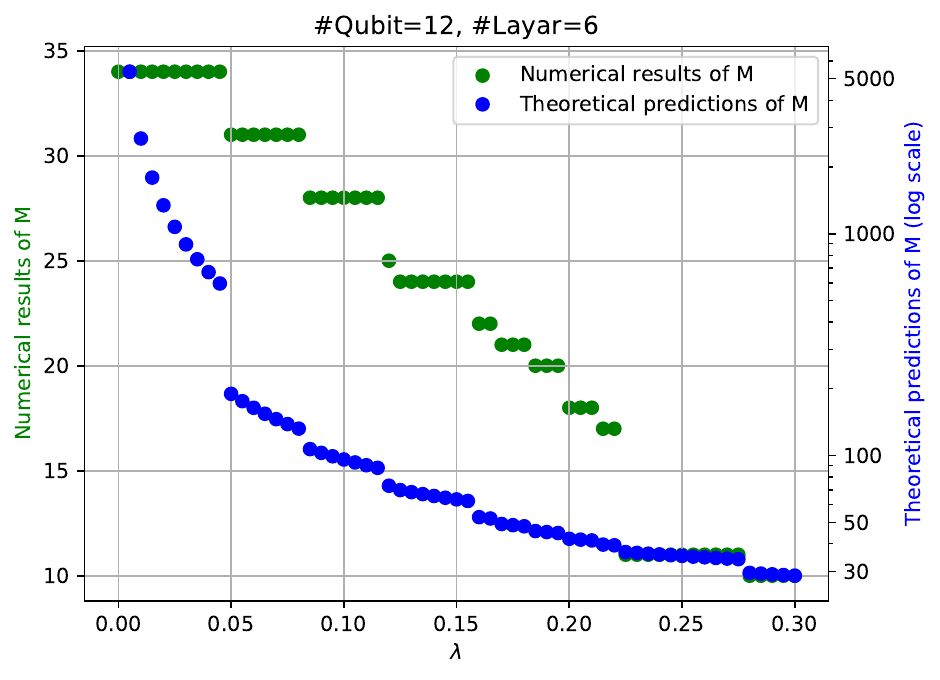}
  \includegraphics[width=0.3\textwidth]{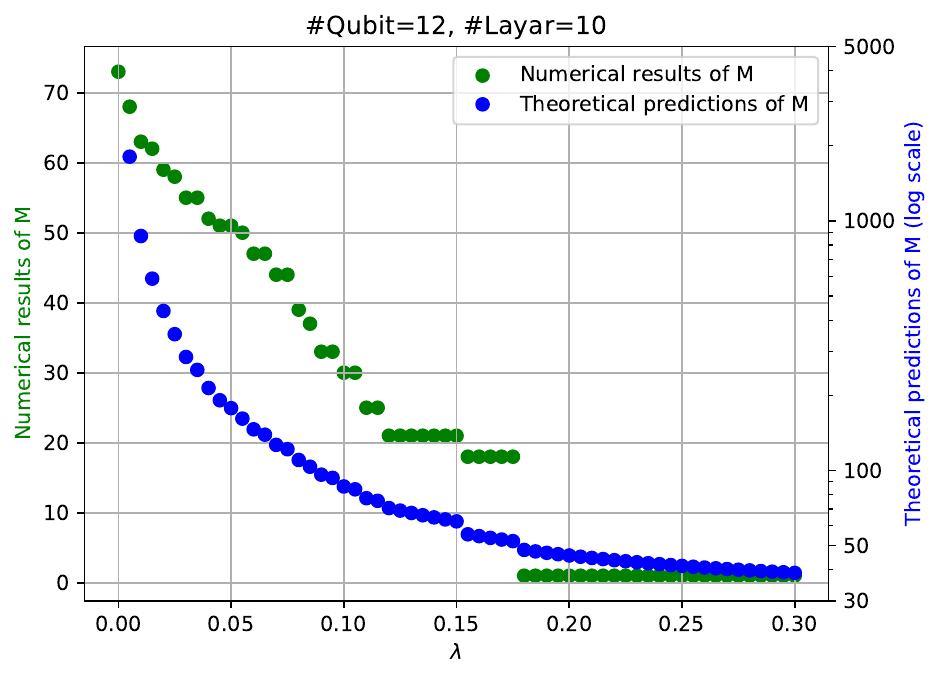}
  \includegraphics[width=0.3\textwidth]{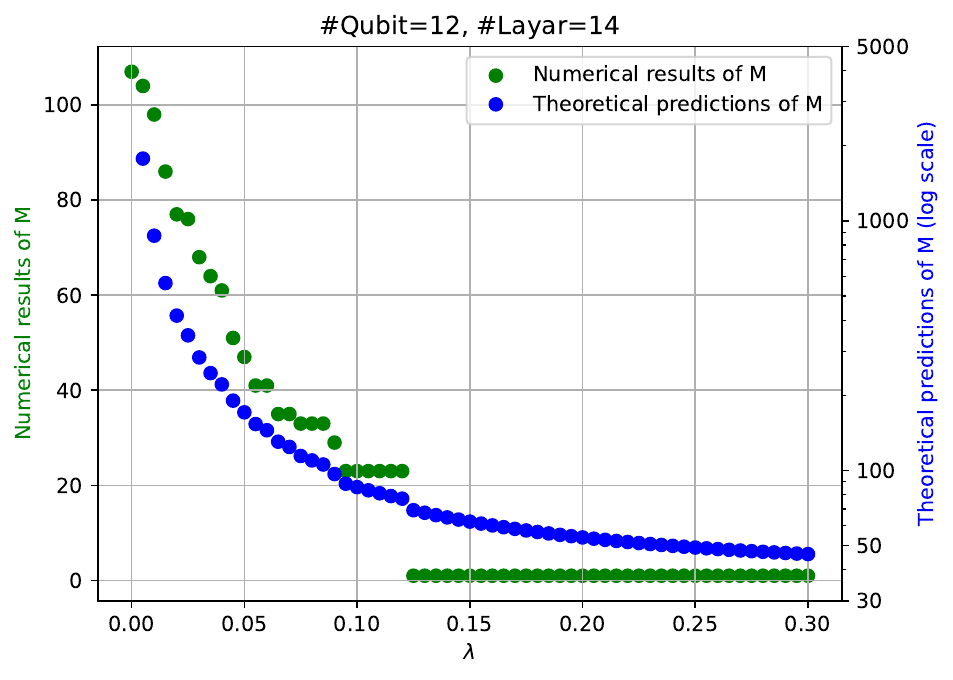}
  \caption{
    The comparison between the minimal truncation $M$ and the analytical bound \eqref{ap:eq:analytical_bound}, referred to as \textit{``Theoretical predictions of M"}. For clarity, logarithmic axes are utilized for the \textit{``Theoretical predictions of M"} data, shown on the right side.}\label{fig:numerical_M2}
\end{figure}
The corresponding results for qubit number $n=12$ are depicted in Fig.~\ref{fig:numerical_M2}~(the other cases are similar). 
To avoid the denominator of Eq.~\eqref{ap:eq:analytical_bound} being $0$, the analytical bound with $\lambda=0$ is removed.
The results suggest that the required value of $M$ to attain acceptable accuracy in practical situations is significantly lower than the theoretical estimate boundary~($\sim 100$ vs $\sim 5000$).

In the end, we demonstrate the changes in the empirical mean squared error $\tilde{\nu}$ as $M$ increases to show the convergence trend, for $n=12$ data, as shown in Fig.~\ref{fig:numerical_MSE}

\begin{figure}[htbp]
  \includegraphics[width=0.3\textwidth]{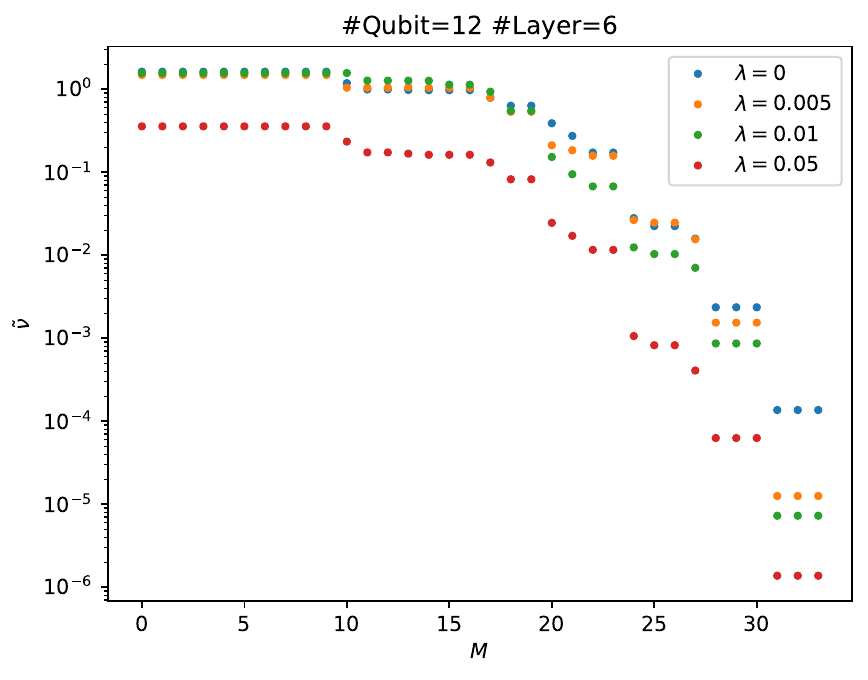}
  \includegraphics[width=0.3\textwidth]{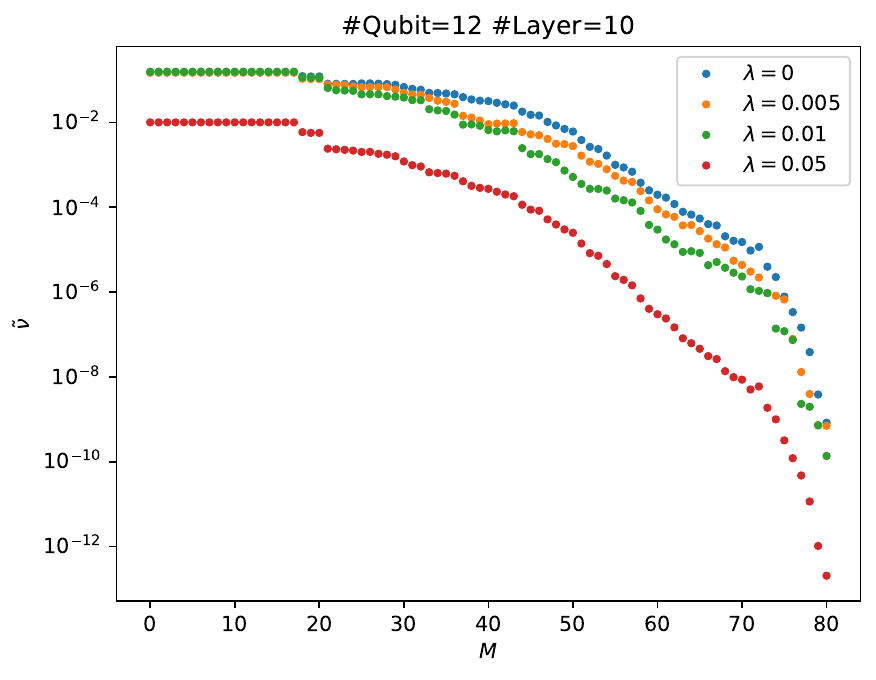}
  \includegraphics[width=0.3\textwidth]{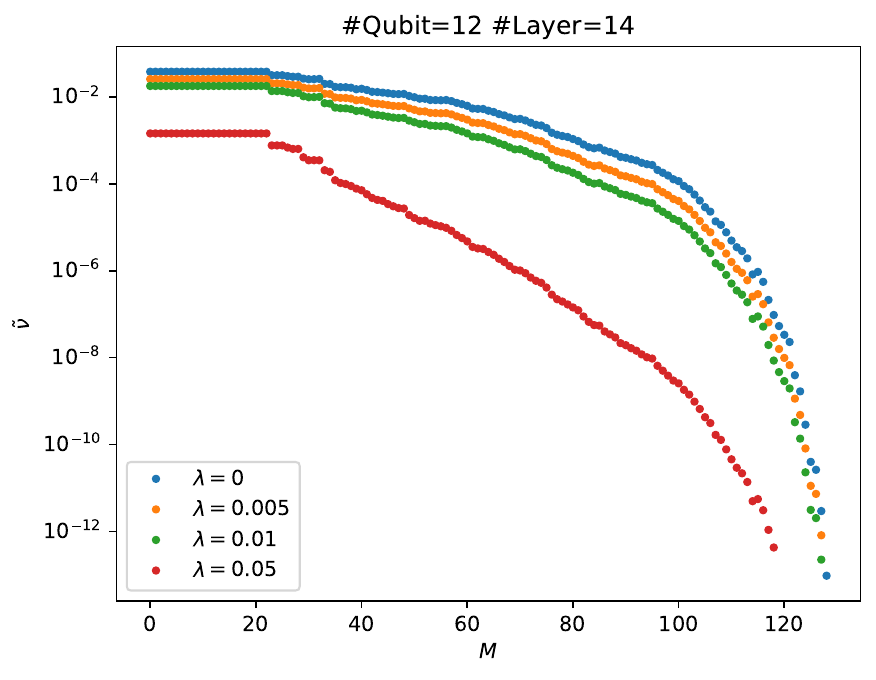}
  \caption{The convergence trend as $M$ increases for various noise rates $\lambda$. The y-axis represents the empirical mean squared error, as indicated in Equation~\ref{ap:eq:emse}, and is displayed on a logarithmic scale. In addition, in the graph \#Layer=6, the result for $M=34$ was omitted because the value of these points dropped below $e^{-20}$.}\label{fig:numerical_MSE}
\end{figure}

\section{Numerical: Details about simulation on IBM's Eagle proccessor}
\label{sec:ap:Numerical:IBM}
\begin{figure}[htbp]	
  \includegraphics[width=1\textwidth]{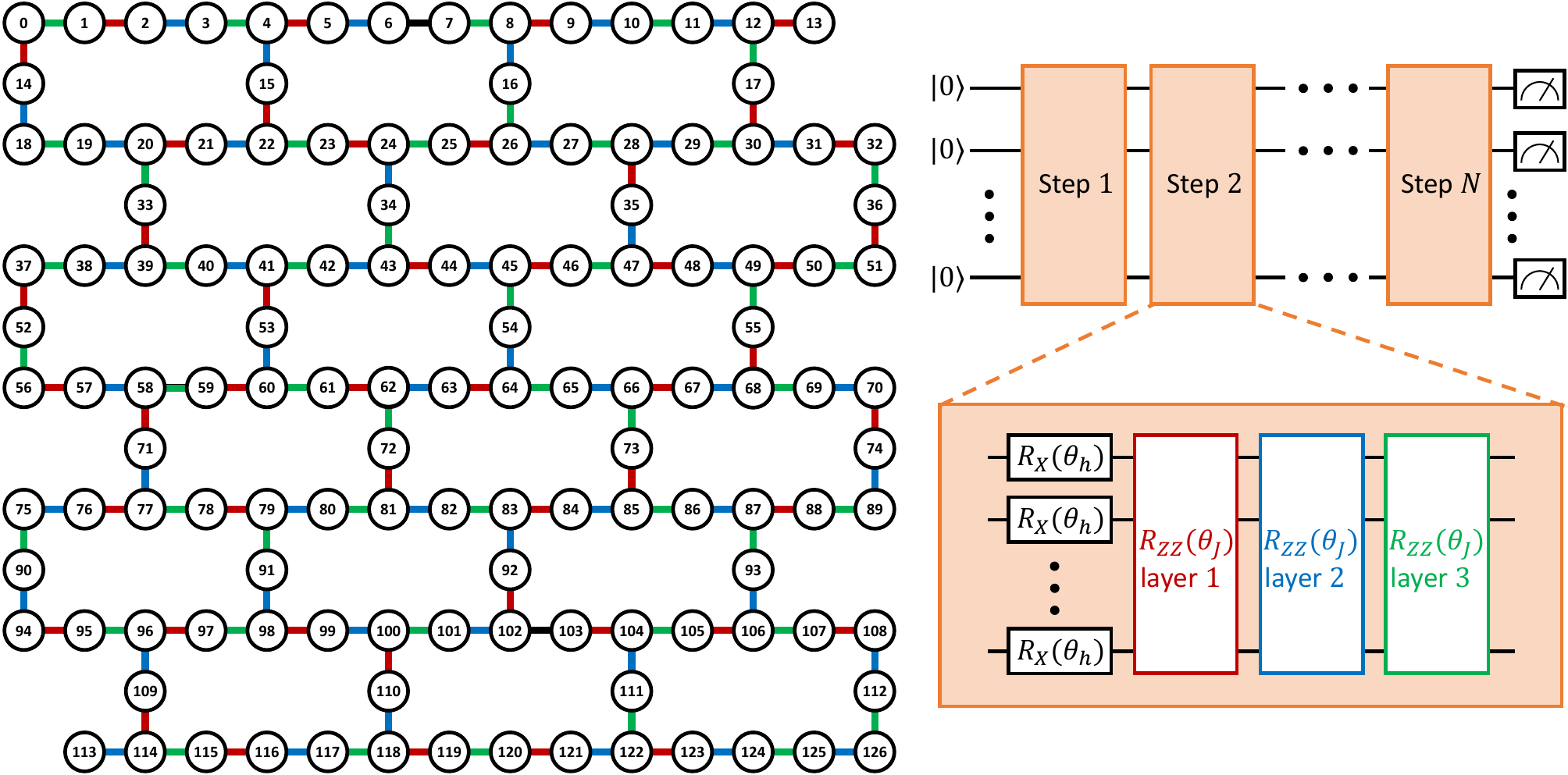}
	\caption{The circuit for first-order Trotterized time evolution in IBM experiments consists of $N$ Trotter steps. Each step contains one layer of $R_X$ gates with a common rotation angle $\theta_h$ applied to each qubit, and three layers of $R_{ZZ}$ gates with a common rotation angle $\theta_J$, whose acting qubits are depicted in the Eagle processor topology on the left and labeled with the same color.}\label{fig:IBM_Ansatz}
\end{figure}

In Ref.~\cite{kim2023evidence}, IBM reported experiments on a 127-qubit Eagle processor and demonstrated the measurement of accurate expectation values.  The benchmark circuits used were constructed from the Trotterized time evolution of a 2D transverse-field Ising model, which was designed to mirror the topology of the Eagle processor.
The time dynamics of the system are governed by the Hamiltonian
\begin{equation}
  H=-J\sum_{\langle i,j\rangle} Z_iZ_j+h\sum_i X_i,
\end{equation}
where $J$ is the coupling strength, $h$ is the transverse field strength, and $\langle i,j\rangle$ denotes the nearest-neighbor qubit pairs.
Spin dynamics can be simulated through the first-order Trotterized time evolution of the Hamiltonian, which is given by
\begin{equation}
  U(\tau)=\exp{-i\tau H}=\prod_{\langle i,j\rangle} \exp{i\tau J Z_iZ_j}\prod_i \exp{-i\tau h X_i}+\order{\tau^2}=\prod_{\langle i,j\rangle}R_{Z_iZ_j}(-2J\tau) \prod_i R_X(2h\tau)+\order{\tau^2},
\end{equation}
in which the evolution time $T$ is discretized into $N$ Trotter steps, with a single step evolution time of $\tau = \frac{T}{N}$.
The Trotterized time evolution is implemented by the ans\"atz shown in Fig.~\ref{fig:IBM_Ansatz}, in which a single step is composed of one layer of $R_X$ gates and three layers of $R_{ZZ}$ gates. The initial state is set as $\rho=\ketbra{0}{0}^{\otimes 127}$. For simplicity, IBM chooses $\theta_J=-2J\tau=-\frac{\pi}{2}$ and considers $\theta_h=2h\tau$ to be in the range $[0,\frac{\pi}{2}]$. In this simulation, we assumed there are depolarizing noises $p_x = p_y = p_z = \frac{\lambda}{4}$ in the hardware.

In the implementation of the simulation algorithm, we initially find all the Pauli paths that satisfy the condition $\abs{s}\leq M$ and have a non-zero contribution, using the back-propagation method described in Supplement Material~\ref{ap:algorithm}. Subsequently, we transform these Pauli paths into trigonometric polynomials according to Eq.~\eqref{ap:eq:f}, Prop.~\ref{prop:f_ele} and Lemma~\ref{lemma:f_noisy}. Finally, we calculate the expectation value by substituting different variables for trigonometric polynomials and summing them.
The code implementation of simulating IBM experimental results is summarized in Algorithm~\ref{ALGORITHM_IBM}.

\begin{algorithm}[H]
\caption{Pseudo-code for estimating analytical expressions of IBM experimental results}\label{ALGORITHM_IBM}
\begin{algorithmic}
  \State 
  \State Set empty list $Pauli\_Path\_mL$ to store Pauli paths which are back-propagated to the $mL$-th layer.
  \State Enumerate $s_L$ as all Pauli words with non-zero coefficient in $O$.
  \For{Pauli path elements $s_{L}$}
  \State{According to the $L$-th layer of the circuit, generate Pauli path elements $s_{L-1}$.}
    \For{Pauli path elements $s_{L-1}$}
      \State{According to the $(L-1)$-th layer of the circuit, generate Pauli path elements $s_{L-2}$.}
      \State{\vdots}
      \For{Pauli path elements $s_{mL}$}
      \State Append pauli path $s=(s_{mL},\cdots,s_L)$ into list $Pauli\_Path\_mL$.
      \EndFor
  \EndFor
  \EndFor
  \Comment{\textbf{Parallelized depth-first search}}
  \State Set empty list $Pauli\_Path$ to store Pauli paths with non-zero contributions and Hamming weights $\leq M$.
  \For{Pauli path elements $s_{mL}$ in list $Pauli\_Path\_mL$}
  \State{According to the $mL$-th layer of the circuit, generate candidates of $s_{mL-1}$.}
  \State{Eliminate cases with $\abs{s_{L}}+\cdots+\abs{s_{mL-1}} > M-(mL-1)$.}
      \State{\vdots}
      \For{Pauli path elements $s_1$}
      \State{According to the $1$-th layer of the circuit, generate candidates of $s_{0}$.}
      \State{Eliminate cases with $\abs{s_{L}}+\cdots+\abs{s_{0}}> M$.}
      \For{Pauli path elements $s_{0}$}
      \If{$X/\sqrt{2}$ or $Y/\sqrt{2}$ in $s_{0}$}
      \State{\textbf{Pass} Because of $\Tr{\ketbra{0}{0}s_0}=0$.}
      \Else
      \State Append pauli path $s=(s_0,\cdots,s_L)$ into list $Pauli\_Path$.
      \EndIf
      \EndFor
  \EndFor
  \EndFor

  \State Calculate the analytical expression by $\widetilde{\mathcal{L}}=\sum_{s \in Pauli\_Path} (1-\lambda)^\abs{s} f(\bm{\theta},s,O,\rho)$.
\end{algorithmic}
\end{algorithm}

In Fig.~\ref{fig:IBM} within the main body, we compare the results of our method with IBM's experiment results both before and after Error mitigation (zero-noise extrapolation).
In order to compare the unmitigated results, we employed a classical optimizer to minimize the distance between experimental dataset $\{(\theta_h,y_{\theta_h})\}$ and our approximate noisy cost function $\widetilde{\mathcal{L}}$, formalized as 
\begin{equation}\label{ap:eq:minimal_lambda}
  \lambda=\arg\min_\lambda  \sqrt{\sum_{(\theta_h,y_{\theta_h})\in \mathrm{data set}}\abs{\widetilde{\mathcal{L}}(\theta_h)-y_{\theta_h}}^2}.
\end{equation}
We utilized the SLSQP optimizer to find $\lambda$, which is integrated within the scipy package~\cite{2020SciPy-NMeth}.

The circuits employed in Figure~\ref{fig:IBM} in the main body are described as follows:
\begin{enumerate}
  \item In Fig.(a)-(c), the Trotter step is set as $N=5$, corresponding to a circuit with depth $L=20$.
  \item In Fig.(d), the Trotter step is set as $N=5$ and there is an additional layer of $R_{X}$ gates applied at the end of the circuit, corresponding to a circuit with depth $L=21$.
  \item In Fig.(e), the Trotter step is set as $N=20$, corresponding to a circuit depth with $L=80$. 
  \item In Fig.(f), we set the rotation angle of $R_{ZZ}$ gates as $\theta_J=-\frac{\pi}{4}$ and the Trotter step as $N=20$, corresponding to a circuit depth with $L=80$. 
\end{enumerate}

Additional informations and the runtimes are presented in Table~\ref{tab:IBM}:
\begin{table}[h!]
  \centering
  \scalebox{0.8}{
\begin{tabular}{|c|c|c|c|c|c|c|c|c|}
  \hline
    Fig. & Qubits &Observable &Step of Trotter & Depth of circuit & M & Runtime (56-core)&$\min_\lambda  \sqrt{\sum\abs{\widetilde{\mathcal{L}}(\theta_h)-y_{\theta_h}}^2}/\mathrm{\#dataset}$\\
  \hline
  (a) &127& $M_Z=\sum_q Z_q/n$ &5 & 20 & 210  & 13s&0.0015892257500055675\\
  \hline
  (b) &127& $X_{13,29, 31}Y_{9, 30}Z_{8, 12, 17, 28, 32}$ & 5& 20 & 210  & 146s&0.001395021411390668\\
  \hline
  (c) &127& $X_{37, 41, 52, 56, 57, 58, 62, 79}Y_{75}Z_{38, 40, 42, 63, 72, 80, 90, 91}$ & 5& 20 & 210  & 29s&0.0007674113478406952\\
  \hline
  (d) &127& $X_{37, 41, 52, 56, 57, 58, 62, 79}Y_{38, 40, 42, 63, 72, 80, 90, 91}Z_{75}$ &5& 21 & 210  & 137s&0.0018301118556860437\\
  \hline
  (e) &127& $Z_{63}$ &20& 80 & 210  & 262s&0.007527499236555272\\
  \hline
  (f) &127& $Z_{63}$ &20& 80 with $\theta_J=-\pi/4$ & 90  & 57s&\\
  \hline
\end{tabular}}
\caption{Comparison for each task and the Runtime of the simulation. The runtime is estimated based on the computation time for the analytical expressions of expected values on 56-core CPU.}\label{tab:IBM}
\end{table}

\section{Numerical: Simulating experimental results in trapped-ion systems}\label{sec:ap:Numerical:ion}

In Sec.~\ref{sec:ap:Numerical:IBM}, we performed a well-fetched simulation of IBM's experiments, which is a typical superconductor qubit system. In this section, we will simulate the experimental results of a trapped-ion system, which are reported in Ref.~\cite{pagano2020quantum}. 

In the experiment, the authors used a 1D array of $^{171}$Yb$^+$ ions to imply a low-depth Quantum Approximate Optimization Algorithm (QAOA). 
The optimization problem, they considered, is encoded in the transverse-field antiferromagnetic Ising Hamiltonian with long-range interactions:

\begin{equation}
    H=\underbrace{\sum_{i<j}J_{ij}X_iX_j}_{H_A}+\underbrace{B\sum_iY_i}_{H_B}.
\end{equation}
Here, $J_{i,j}>0$~(falls off as a power law in the distance between the spins) is the Ising coupling between spins $i$ and $j$; and $B$ denotes the transverse magnetic field.

The state obtained after $p$ layers of the QAOA is:
\begin{equation}
  \ket{\vec{\beta},\vec{\gamma}}=\prod_{k=1}^p e^{-i\beta_k (H_B/J_0)} e^{-i\gamma_k (H_A/J_0)}\ket{\psi_0},
\end{equation}
where $J_0$ is the average nearest-neighbor coupling and the angles $\vec{\beta}$ and $\vec{\gamma}$ are the variational parameters used to minimize the final energy:
\begin{equation}
  E(\vec{\beta},\vec{\gamma})=\bra{\vec{\beta},\vec{\gamma}}H\ket{\vec{\beta},\vec{\gamma}}.
\end{equation}
And the initial state is set as $\ket{\psi_0}=\ket{\uparrow\uparrow\dots\uparrow}_Y$, where $\ket{\uparrow}_Y=(\ket{\uparrow}_Z+i\ket{\downarrow}_Z)/\sqrt{2}$

The authors used the dimensionless quantity:
\begin{equation}\label{ap:eq:qaoa:eta}
  \eta=\frac{E(\vec{\beta},\vec{\gamma})-E_{\max}}{E_{gs}-E_{\max}}
\end{equation}
to measure the performance of the QAOA, where $E_{gs}$ is the ground state energy and $E_{\max}$ is the energy of the highest excited state.

In this section, we will simulate the experimental results of Fig.2C in Ref.~\cite{pagano2020quantum} using our method. 
The number of qubits is $12$, and the number of QAOA layers is $p=1$. In our simulation, the Ising couplings $J_{ij}$ are set as a fitting analytic form given in Supplementary Information of Ref.~\cite{pagano2020quantum}:
\begin{equation}
  J_{ij}\approx\frac{J_0}{r^{\alpha^\prime}}e^{-\beta^\prime(r-1)},
\end{equation}
where $J_0$ is the average nearest-neighbor coupling, $r=\abs{i-j}$ is the ion separation and $\alpha^\prime,\beta^\prime$ are exponential decay variables. As the reference reported, $J_0=0.580$, $\alpha^\prime=0.322$, $\beta^\prime=0.229$ and the transverse magnetic field satisfied $B/J_0=-0.25$.

\begin{figure}[hbp]	
  \includegraphics[width=0.5\textwidth]{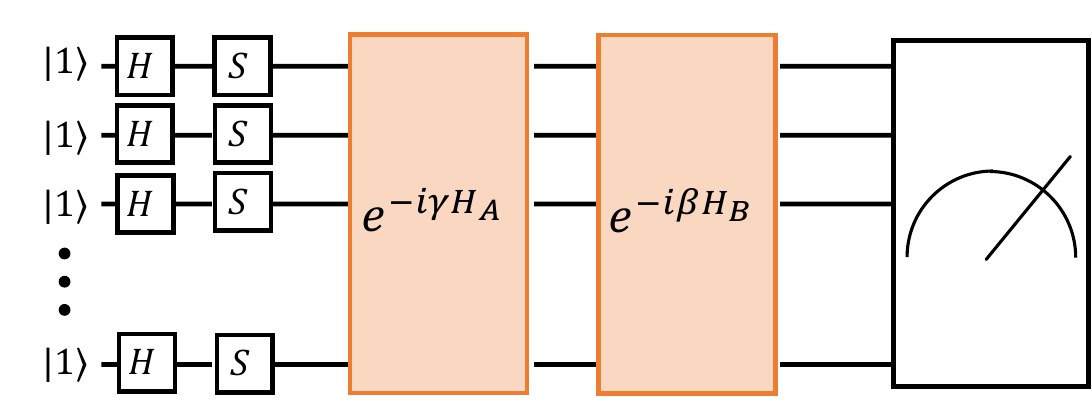}
	\caption{The 12 qubits circuit employed in the simulation. 
  The initial state $\ket{1}$ is first passed through a layer of $H$ gate and a layer of $S$ gate to equivalently generate the initial state $\ket{\psi_0}$ in the experiments.
  The evolution corresponding to $H_A$ in the circuit is implemented by $11$ layers of non-overlapping $R_{XX}$ gates, and the evolution corresponding to $H_B$ is implemented by a layer of $R_Y$ gates acting on each qubit. The total number of layers in the circuit is $14$.}\label{fig:Ion_Ansatz}
\end{figure}

The circuits employed in the simulation are described as shown in Fig.~\ref{fig:Ion_Ansatz}. By equation $SH\ket{1}=S(\ket{0}-\ket{1})/\sqrt{2}=(\ket{0}-i\ket{1})/\sqrt{2}=-i\ket{\uparrow}_Y$, we use the initial state $\ket{1}$ and pass it through a layer of $H$ gate and a layer of $S$ gate to equivalently generate the initial state $\ket{\psi_0}$ in the experiments. 

The evolution corresponding to $H_A$ in the circuit is $e^{-i\gamma (H_A/J_0)}=e^{-i\gamma \sum (J_{ij}/J_0)X_iX_j}$, which can be implemented by a series of $R_{XX}$ gates $\Pi e^{-i\gamma(J_{ij}/J_0)X_iX_j}$.
There are a total of $(^{12}_2)=66$ pairwise interactions~(or $R_{XX}$ gates), and these $R_{XX}$ gates can be arranged in the most compact way using $66/6=11$ layers in the circuit.
And the evolution $e^{-i\beta (H_B/J_0)}=e^{-i\beta \sum(B/J_0)Y_i}$ is implemented by a layer of $R_Y$ gates acting on each qubit.

\begin{figure}[htbp]	
  \includegraphics[width=0.5\textwidth]{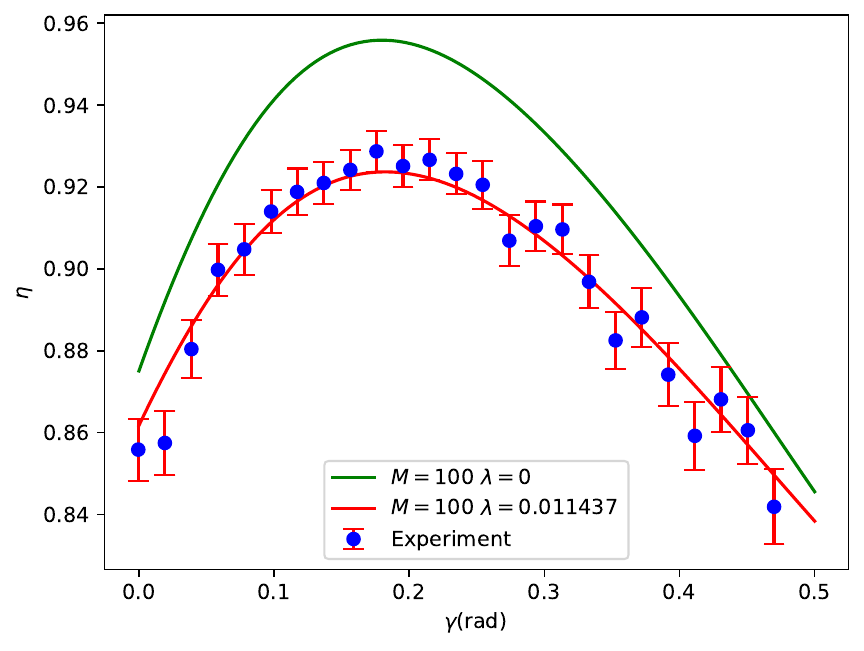}
	\caption{The comparison between experiments and simulations, corresponding to Fig.~2C in Ref.~\cite{pagano2020quantum}. The y-axis $\eta$ is defined in Eq.~\ref{ap:eq:qaoa:eta}, and $\beta^\star=1.12$.}\label{fig:QAOA_result}
\end{figure}

In Fig.~2C of Ref.~\cite{pagano2020quantum}, the authors reported the experimental results of $\eta$ for $\beta^\star=1.12$ as a function of the variational parameters $\gamma$. The simulated results are shown in Fig.~\ref{fig:QAOA_result}. Assessing depolarizing noises $p_x = p_y = p_z = \frac{\lambda}{4}$ happened, similar to the previous section, we use the SLSQP optimizer to optimize Eq.~\ref{ap:eq:minimal_lambda} to obtain an appropriate noise rate $\lambda$ to fit the experimental results. The optimal noise rate is $\lambda=0.011437$.
From the Fig, it can be seen that the appearance of noise affects the performance of the QAOA, and the experimental results obtained are slightly lower compared to the results speculated without noise. 
Moreover, in this trapped-ion system experiment, our method fits the noisy results well with the experimental results.

\section{Discussion: More noise models}
\label{sec:ap:discuss:unital}

\subsection{Variations of circuits under single-qubit unital noise}
In this section, we discussed the impact of single-qubit unital noise.
Unital qubit channels can be considered as variations of Pauli error channels, be written as $\mathcal{N}(\rho) = U \mathcal{P}(V^\dagger \rho V) U^\dagger$, where $U$ and $V$ are unitary operators, and $\mathcal{P}$ represents a Pauli error channel~\cite{choi2023unital}. 

To neutralize $U$ and $V$, we substitute variational gates ${\mathcal{U}}_i$ for $\mathcal{W}_i=V^{\otimes n}\mathcal{U}_i U^{\dagger\otimes n}$. Consequently, the circuit $\mathcal{U}(\bm{\theta})=\mathcal{W}_L(\bm{\theta}_L)\mathcal{W}_{L-1}(\bm{\theta}_{L-1})\cdots \mathcal{W}_1(\bm{\theta}_1)$ is a variation of parameterized quantum circuits and all these kind circuits consist of a universal circuit model.

Consequently, if each Pauli word in sparse~(described in Sec.~\ref{ap:pre_data}) observable $O$ has at most constant $K$ non-identity elements, then Theorem~\ref{thm:main} remains valid for the unital noise $\mathcal{N}$. Formalized as:

\begin{theorem}\label{ap:thm:main:unital}
  Suppose $\langle \{\overline{\sigma}_{i,j}\}\rangle/\left(\langle \{\overline{\sigma}_{i,j}\}\rangle\cap\langle i\mathbb{I}^{\otimes n}\rangle\right)=\{\mathbb{I},X,Y,Z\}^{\otimes n}$ is satisfied and $O,\rho$ are sparse, each Pauli word in $O$ has at most constant $K$ non-identity, for a fixed $\gamma$, given arbitrary truncation error $\varepsilon$, there exists a polynomial-scale classical algorithm to determine the approximated noisy cost function $\widetilde{\mathcal{L}}$, which satisfies $\abs{\widetilde{\mathcal{L}}-\widehat{\mathcal{L}}} \leq \varepsilon$ with a probability of at least $1-\delta$ over all possible parameters $\bm{\theta}$. The time complexity is $\mathrm{Poly}(n) \order{L} \bigg(\frac{\norm{O}_\infty}{\varepsilon \sqrt{\delta}} \bigg)^{\order{\sfrac{1}{\gamma}}}3^K$ for $\mathcal{P}$ in Case 1 and $\mathrm{Poly}(n)  \order{(nL)^{\frac{1}{2\gamma} \ln{\frac{\norm{O}_\infty}{\varepsilon \sqrt{\delta}}}+1}}3^K$ for $\mathcal{P}$ in Case 2. The space complexity is $\order{\mathrm{Poly}(n)+nL}$.
\end{theorem}

Thus, our method can approximately simulate the variation of variational quantum circuits under unital noise $\mathcal{N}$, with time complexity depending on the category of $\mathcal{P}$.

To demonstrate the validity of Theorem~\ref{ap:thm:main:unital}, we define $\mathcal{Q}=U^{\dagger\otimes n}O U^{\otimes n}$. 
We notice that for each single-qubit Pauli operator $p\in\{X,Y,Z\}$ the operator $U^\dagger p U$ is a linear combination of $X,Y,Z$. 
By the assumption that each Pauli word in $O$ has at most constant $K$ non-identity elements, $\mathcal{Q}$ is a linear combination of Pauli words with at most $3^K\mathrm{Poly}(n)$ words. This introduces a $3^K$ overhead in the number of elements in $\mathcal{Q}$.

The noisy expected value can be represented as:
  \begin{equation}\label{eq:unital_noise}
    \begin{aligned}
      \widehat{\mathcal{L}}(\bm{\theta})&=\Tr{O \mathcal{N}^{\otimes n}(\mathcal{W}_L \mathcal{N}^{\otimes n}(\cdots\mathcal{W}_1\mathcal{N}^{\otimes n}(\rho) \mathcal{W}_1^\dagger\cdots) \mathcal{W}_L^\dagger)}\\
      &=\;\begin{tikzpicture}[baseline=(current bounding box.center)]
        \coordinate(l)at(-0.25,0.75){};\coordinate(r)at(8.25,0.75){};
        \node[rectangle,draw] (H) at (0,0) {$O$};
        \node[rectangle,draw,minimum height=40] (Nl) at (1,0.4) {$\mathcal{N}^{\otimes n}$};
        \node[rectangle,draw] (Ul) at (2,0) {$\mathcal{W}_L$};
        \node[rectangle,draw] (Ujt) at (2,0.75) {$\overline{\mathcal{W}}_L$};
        \node[rectangle,draw,minimum height=40] (Nl1) at (3,0.4) {$\mathcal{N}^{\otimes n}$};
        \node[rectangle,draw,minimum height=40] (N1) at (5,0.4) {$\mathcal{N}^{\otimes n}$};
        \node[rectangle,draw,minimum height=40] (N0) at (7,0.4) {$\mathcal{N}^{\otimes n}$};
        \node[rectangle,draw] (U1) at (6,0) {$\mathcal{W}_1$};
        \node[rectangle,draw] (U1t) at (6,0.75) {$\overline{\mathcal{W}}_1$};
        \node[] (cdots_down) at (4,0) {$\cdots$};
        \node[] (cdots_up) at (4,0.75) {$\cdots$};
        \node[rectangle,draw] (rho) at (8,0) {$\rho$};
        \draw [thick] (H)--(H-| Nl.west) (Ul-|Nl.east)--(Ul)--(Ul-| Nl1.west) (cdots_down-|Nl1.east)--(cdots_down)--(cdots_down-| N1.west) (U1-|N1.east)--(U1)--(U1-| N0.west) (rho-|N0.east)--(rho) (l)--(l-| Nl.west) (Ujt-| Nl.east)--(Ujt)--(Ujt-| Nl1.west) (cdots_up-|Nl1.east)-- (cdots_up)--(cdots_up-| N1.west) (U1t-|N1.east)--(U1t)--(U1t-| N0.west) (r-|N0.east)--(r);
        \draw[thick] (-0.25,0.75) arc(90:270:0.75/2);
        \draw[thick] (8.25,0) arc(-90:90:0.75/2);
        \end{tikzpicture}\\
        &=\sum_{s_0,\cdots,s_L \in \bm{P}_n}
        \begin{tikzpicture}[baseline=(current bounding box.center)]
          \coordinate(l)at(-0.25,0.5){};\coordinate(r)at(3.25,0.5){};
          \node[rectangle,draw] (H) at (0,0) {$O$};
          \node[rectangle,draw,minimum height=40] (N) at (1,0.4) {$\mathcal{N}^{\otimes n}$};
          \node[draw,shape=circle,inner sep=1pt] (sL) at (3,0) {$s_L$};
          \node[rectangle,draw] (V) at (2.15,0) {$V^{\otimes n}$};
          \node[rectangle,draw] (Vt) at (2.15,0.75) {$\overline{V}^{\otimes n}$};
          \draw [thick] (H)--(H-| N.west) (V-|N.east)--(V)--(sL) (l)--(l-| N.west) (Vt-|N.east)--(Vt)--(r);
          \draw[thick] (-0.25,0.5) arc(90:270:0.25);
          \draw[thick] (3.25,0) arc(-90:90:0.25);
          \end{tikzpicture}\;
      \dots
      \begin{tikzpicture}[baseline=(current bounding box.center)]
        \coordinate(l)at(-1.25,0.5){};\coordinate(r)at(4,0.5){};
        \node[draw,shape=circle,inner sep=1pt] (sj) at (-1.,0) {$s_i$};
        \node[rectangle,draw] (Uj) at (0.75,0) {$\mathcal{W}_i$};
        \node[rectangle,draw] (V) at (2.8,0) {$V^{\otimes n}$};
        \node[rectangle,draw] (Vt) at (2.8,0.75) {$\overline{V}^{\otimes n}$};
        \node[rectangle,draw] (V2) at (-0.15,0) {$V^{\dagger\otimes n}$};
        \node[rectangle,draw] (Vt2) at (-0.15,0.75) {${V}^{T\otimes n}$};
        \node[draw,shape=circle,inner sep=0pt] (sj1) at (3.75,0) {$s_{i-1}$};
        \node[rectangle,draw] (Ujt) at (0.75,0.75) {$\overline{\mathcal{W}}_i$};
        \node[rectangle,draw,minimum height=40] (N) at (1.75,0.4) {$\mathcal{N}^{\otimes n}$};
        \draw [thick] (sj)--(V2)--(Uj)--(Uj-| N.west) (V-|N.east)--(V)--(sj1) (l) --(Vt2)--(l-|Ujt.west) (Ujt)--(Ujt-| N.west) (Vt-|N.east)--(Vt)--(r);
        \draw[thick] (-1.25,0.5) arc(90:270:0.25);
        \draw[thick] (4,0) arc(-90:90:0.25);
      \end{tikzpicture}
      \dots
        \begin{tikzpicture}[baseline=(current bounding box.center)]
          \coordinate(l)at(-0.25,0.5){};\coordinate(r)at(1.9,0.5){};
          \node[rectangle,draw] (rho) at (1.7,0) {$\rho$};
          \node[rectangle,draw] (V2) at (0.8,0) {$V^{\dagger\otimes n}$};
          \node[rectangle,draw] (Vt2) at (0.8,0.75) {${V}^{T\otimes n}$};
          \node[rectangle,draw,minimum height=40,opacity=0] (N) at (1,0.4) {$\mathcal{N}^{\otimes n}$};
          \node[draw,shape=circle,inner sep=1pt] (s0) at (0.,0) {$s_0$};
          \draw [thick] (rho)--(V2)--(s0) (l)--(Vt2)--(r);
          \draw[thick] (-0.25,0.5) arc(90:270:0.25);
          \draw[thick] (1.9,0) arc(-90:90:0.25);
        \end{tikzpicture}\\
      &=\sum_{s_0,\cdots,s_L \in \bm{P}_n}
        \begin{tikzpicture}[baseline=(current bounding box.center)]
          \coordinate(l)at(-0.25,0.5){};\coordinate(r)at(2.15,0.5){};
          \node[rectangle,draw] (H) at (0,0) {$\mathcal{Q}$};
          \node[rectangle,draw,minimum height=40] (N) at (1,0.4) {$\mathcal{P}^{\otimes n}$};
          \node[draw,shape=circle,inner sep=1pt] (sL) at (1.85,0) {$s_L$};
          \draw [thick] (H)--(H-| N.west) (sL-|N.east)--(sL) (l)--(l-| N.west) (r-|N.east)--(r);
          \draw[thick] (-0.25,0.5) arc(90:270:0.25);
          \draw[thick] (2.15,0) arc(-90:90:0.25);
          \end{tikzpicture}\;
      \dots
      \begin{tikzpicture}[baseline=(current bounding box.center)]
        \coordinate(l)at(-0.25,0.5){};\coordinate(r)at(3,0.5){};
          \node[draw,shape=circle,inner sep=1pt] (sj) at (0.,0) {$s_i$};
          \node[rectangle,draw] (Uj) at (0.75,0) {$\mathcal{U}_i$};
          \node[draw,shape=circle,inner sep=0pt] (sj1) at (2.75,0) {$s_{i-1}$};
          \node[rectangle,draw] (Ujt) at (0.75,0.75) {$\overline{\mathcal{U}}_i$};
          \node[rectangle,draw,minimum height=40] (N) at (1.75,0.4) {$\mathcal{P}^{\otimes n}$};
          \draw [thick] (sj)--(Uj)--(Uj-| N.west) (sj1-|N.east)--(sj1) (l) --(l-|Ujt.west) (Ujt)--(Ujt-| N.west) (r-|N.east)-- (r);
          \draw[thick] (-0.25,0.5) arc(90:270:0.25);
          \draw[thick] (3,0) arc(-90:90:0.25);
        \end{tikzpicture}
      \dots
        \begin{tikzpicture}[baseline=(current bounding box.center)]
          \coordinate(l)at(-0.25,0.5){};\coordinate(r)at(1.9,0.5){};
          \node[rectangle,draw] (rho) at (1.7,0) {$\rho$};
          \node[rectangle,draw] (V2) at (0.8,0) {$V^{\dagger\otimes n}$};
          \node[rectangle,draw] (Vt2) at (0.8,0.75) {${V}^{T\otimes n}$};
          \node[rectangle,draw,minimum height=40,opacity=0] (N) at (1,0.4) {$\mathcal{N}^{\otimes n}$};
          \node[draw,shape=circle,inner sep=1pt] (s0) at (0.,0) {$s_0$};
          \draw [thick] (rho)--(V2)--(s0) (l)--(Vt2)--(r);
          \draw[thick] (-0.25,0.5) arc(90:270:0.25);
          \draw[thick] (1.9,0) arc(-90:90:0.25);
        \end{tikzpicture}\\
      &=\Tr{\mathcal{Q} \mathcal{P}^{\otimes n}(\mathcal{U}_L \mathcal{P}^{\otimes n}(\cdots\mathcal{U}_1\mathcal{P}^{\otimes n}( V^{\dagger\otimes n}\rho V^{\otimes n}) \mathcal{W}_1^\dagger\cdots) \mathcal{W}_L^\dagger)},
    \end{aligned}
  \end{equation}
  where the third equality is obtained by substituting $VV^\dagger$ for $I$ in Eq.~\refeq{ap:eq:tn:cost_function_noisy}, and the fourth equality is held by the local unital noise channel $\mathcal{N}(\rho)=U \mathcal{P}(V^\dagger\rho V)U^\dagger$, $\mathcal{W}_i=V^{\otimes n}\mathcal{U}_i U^{\dagger\otimes n}$ and $\mathcal{Q}=U^{\dagger\otimes n}O U^{\otimes n}$:
  \begin{equation}
    \begin{tikzpicture}[baseline=(current bounding box.center)]
      \coordinate(l)at(-1.25,0.5){};\coordinate(r)at(4,0.5){};
      \node[draw,shape=circle,inner sep=1pt] (sj) at (-1.,0) {$s_i$};
      \node[rectangle,draw] (Uj) at (0.75,0) {$\mathcal{W}_i$};
      \node[rectangle,draw] (V) at (2.8,0) {$V^{\otimes n}$};
      \node[rectangle,draw] (Vt) at (2.8,0.75) {$\overline{V}^{\otimes n}$};
      \node[rectangle,draw] (V2) at (-0.15,0) {$V^{\dagger\otimes n}$};
      \node[rectangle,draw] (Vt2) at (-0.15,0.75) {${V}^{T\otimes n}$};
      \node[draw,shape=circle,inner sep=0pt] (sj1) at (3.75,0) {$s_{i-1}$};
      \node[rectangle,draw] (Ujt) at (0.75,0.75) {$\overline{\mathcal{W}}_i$};
      \node[rectangle,draw,minimum height=40] (N) at (1.75,0.4) {$\mathcal{N}^{\otimes n}$};
      \draw [thick] (sj)--(V2)--(Uj)--(Uj-| N.west) (V-|N.east)--(V)--(sj1) (l) --(Vt2)--(l-|Ujt.west) (Ujt)--(Ujt-| N.west) (Vt-|N.east)--(Vt)--(r);
      \draw[thick] (-1.25,0.5) arc(90:270:0.25);
      \draw[thick] (4,0) arc(-90:90:0.25);
    \end{tikzpicture}
    =
    \begin{tikzpicture}[baseline=(current bounding box.center)]
      \coordinate(l)at(-1.25,0.5){};\coordinate(r)at(4,0.5){};
      \node[draw,shape=circle,inner sep=1pt] (sj) at (-1.,0) {$s_i$};
      \node[rectangle,draw] (Uj) at (0.75,0) {$\mathcal{W}_i$};
      \node[rectangle,draw] (U) at (1.65,0) {$U^{\otimes n}$};
      \node[rectangle,draw] (Ut) at (1.65,0.75) {$\overline{U}^{\otimes n}$};
      \node[rectangle,draw] (V2) at (-0.15,0) {$V^{\dagger\otimes n}$};
      \node[rectangle,draw] (Vt2) at (-0.15,0.75) {${V}^{T\otimes n}$};
      \node[draw,shape=circle,inner sep=0pt] (sj1) at (3.75,0) {$s_{i-1}$};
      \node[rectangle,draw] (Ujt) at (0.75,0.75) {$\overline{\mathcal{W}}_i$};
      \node[rectangle,draw,minimum height=40] (N) at (2.75,0.4) {$\mathcal{P}^{\otimes n}$};
      \draw [thick] (sj)--(V2)--(Uj)--(U)--(U-| N.west) (sj1-|N.east)--(sj1) (l) --(Vt2)--(l-|Ujt.west) (Ujt)--(Ut)--(Ut-| N.west) (r-|N.east)--(r);
      \draw[thick] (-1.25,0.5) arc(90:270:0.25);
      \draw[thick] (4,0) arc(-90:90:0.25);
    \end{tikzpicture}
    =
    \begin{tikzpicture}[baseline=(current bounding box.center)]
      \coordinate(l)at(-0.25,0.5){};\coordinate(r)at(3,0.5){};
        \node[draw,shape=circle,inner sep=1pt] (sj) at (0.,0) {$s_i$};
        \node[rectangle,draw] (Uj) at (0.75,0) {$\mathcal{U}_i$};
        \node[draw,shape=circle,inner sep=0pt] (sj1) at (2.75,0) {$s_{i-1}$};
        \node[rectangle,draw] (Ujt) at (0.75,0.75) {$\overline{\mathcal{U}}_i$};
        \node[rectangle,draw,minimum height=40] (N) at (1.75,0.4) {$\mathcal{P}^{\otimes n}$};
        \draw [thick] (sj)--(Uj)--(Uj-| N.west) (sj1-|N.east)--(sj1) (l) --(l-|Ujt.west) (Ujt)--(Ujt-| N.west) (r-|N.east)-- (r);
        \draw[thick] (-0.25,0.5) arc(90:270:0.25);
        \draw[thick] (3,0) arc(-90:90:0.25);
      \end{tikzpicture}.
  \end{equation}

  By applying Theorem~\ref{thm:main} to estimate $\Tr{\mathcal{Q} \mathcal{P}^{\otimes n}(\mathcal{U}_L \mathcal{P}^{\otimes n}(\cdots\mathcal{U}_1\mathcal{P}^{\otimes n}( V^{\dagger\otimes n}\rho V^{\otimes n}) \mathcal{W}_1^\dagger\cdots) \mathcal{W}_L^\dagger)}$, we can obtain the approximated noisy expected value $\widetilde{\mathcal{L}}$.
  The differences between Eq.\eqref{eq:unital_noise} and the case in Pauli error is the initial state term ( $\begin{tikzpicture}[baseline=(current bounding box.center)]
      \coordinate(l)at(-0.25,0.5){};\coordinate(r)at(1.9,0.5){};
      \node[rectangle,draw] (rho) at (1.7,0) {$\rho$};
      \node[rectangle,draw] (V2) at (0.8,0) {$V^{\dagger\otimes n}$};
      \node[rectangle,draw] (Vt2) at (0.8,0.75) {${V}^{T\otimes n}$};
      \node[rectangle,draw,minimum height=40,opacity=0] (N) at (1,0.4) {$\mathcal{N}^{\otimes n}$};
      \node[draw,shape=circle,inner sep=1pt] (s0) at (0.,0) {$s_0$};
      \draw [thick] (rho)--(V2)--(s0) (l)--(Vt2)--(r);
      \draw[thick] (-0.25,0.5) arc(90:270:0.25);
      \draw[thick] (1.9,0) arc(-90:90:0.25);
    \end{tikzpicture}$ in the fourth line of Eq.\eqref{eq:unital_noise}).
    It introduces a constant factor in the process of calculating the contributions of Pauli paths. 
    
    Combined with the $3^K$ overhead in the observable $\mathcal{Q}$, there is an additional factor of $\order{3^K}$ in the time complexity of the simulation, where $K$ is a constant in many common VQAs. At this point, we have completed the proof of Theorem~\ref{ap:thm:main:unital}.

  Hence, when addressing the local unital noise channel $\mathcal{N}$, it suffices to examine the Pauli error channels $\mathcal{P}$ by excluding $U$ and $V$ rotations, thereby classifying them by number of non-zero noise variable in $\mathcal{P}$, as shown before.

\subsection{Finite amplitude damping noise}
In this section, we examine a circuit that is simultaneously affected by depolarizing noise and discontinuous amplitude damping noise, as depicted in Fig~\ref{reply_unital_fig}.
The circuit is segmented into sets of gates, denoted as $G_A, G_B, \ldots$, where local depolarizing noise occurs in each set and amplitude damping noise acts between sets at most a constant $K$ positions. If the gates in each set satisfy the relationship $\langle \{\overline{\sigma}_{i,j}\}\rangle/\left(\langle \{\overline{\sigma}_{i,j}\}\rangle\cap\langle i\mathbb{I}^{\otimes n}\rangle\right)=\{\mathbb{I},X,Y,Z\}^{\otimes n}$~(or Eq.~\eqref{ap:eq:E_cross_equals_0} holds), our methodology is capable of estimating the circuit's expected outcome.  
This estimation is achieved with a computational complexity characterized by $2^K\mathrm{Poly}\left(n,L,\sfrac{1}{\varepsilon},\sfrac{1}{\sqrt{\delta}},\norm{O}_{\infty}\right)$.

\begin{figure}[htbp]
  \includegraphics[width=0.4\textwidth]{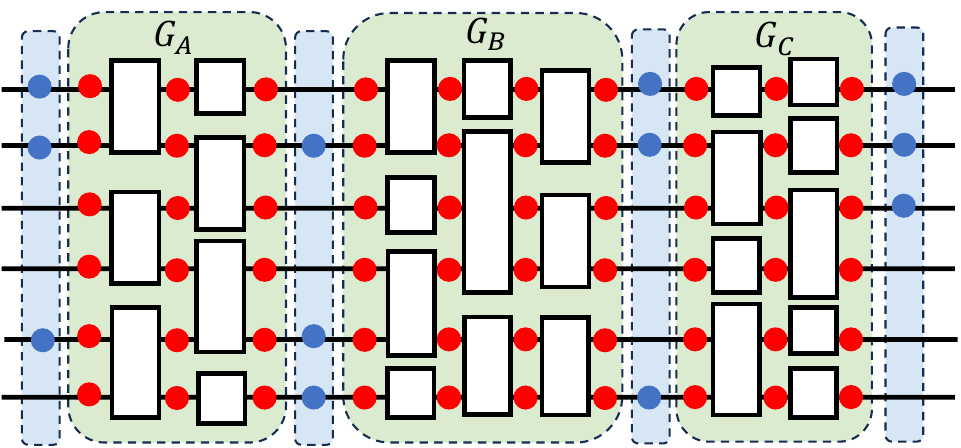}
  \caption{A circuit that contains both depolarizing noise and amplitude damping noise simultaneously. The depolarizing noise and amplitude damping noise are denoted by red and blue dots, respectively. The amplitude damping noise manifests intermittently, impacting a maximum of $K=12$ positions.}\label{reply_unital_fig}
\end{figure}

  This is because for amplitude damping channel $\mathcal{N}$, we have
  \begin{equation}
    \mathcal{N}(I)=I+\gamma Z, \quad \mathcal{N}(X)=\sqrt{1-\gamma}X, \quad \mathcal{N}(Y)=\sqrt{1-\gamma}Y, \quad \mathcal{N}(Z)=(1-\gamma)Z.
  \end{equation}
  When the Pauli path traverses it and the corresponding position is $I$, the Pauli path will bifurcate and create a new path.
  This leads to an additional factor of $2^K$ in the final time complexity, by the assumption that $K$ is a constant, thereby not impacting the outcomes of theorem~\ref{thm:main}. 

\subsection{General non-unital noise}
For general non-unital noise, we temporarily could not provide general classical easiness results like Theorem~\ref{thm:main} because of the drastic properties of non-unital noise. There are some relative works, which provide some insights into the non-unital noise and reveal distinct properties when compared to unital noise:
\begin{itemize}
  \item The work in Ref.~\cite{ben2013quantum} shows the possibility of performing fault-tolerant quantum computation under non-unital noise, with specially constructed circuits. They demonstrated that the non-unital noise like amplitude damping channel can be used as a resource to generate fresh ancilla qubits. This indicates that it may be inherently difficult to classically simulate the general quantum circuits efficiently under non-unital noise without additional assumptions.

  \item The anti-concentration of the output distribution is a critical characteristic in many classical complexity analyses~\cite{bouland2019complexity,boixo2018characterizing,deshpande2022tight}. 
  Research indicates that the output distribution of random quantum circuits demonstrates anti-concentration in noiseless conditions or under unital noise at sufficiently large depths~\cite{deshpande2022tight,dalzell2022random}.

  However, Ref.~\cite{fefferman2023effect} reveals that the output distribution exhibits a lack of anti-concentration under non-unital noise, even with additional unital noise sources, regardless of the circuit depth. This stands in sharp contrast to the behavior of noiseless random quantum circuits or those subject only to unital noise. It also offers a new perspective for studying the impact of non-unital noise on classical simulation hardness.

  \item If the circuit has enough random structure~(local 2-design), then for any single-qubit non-unitary noise (i.e., a local 2-design gate applied after each instance of noise), Ref.~\cite{mele2024noise} found the influence of gates on Pauli expectation values decreases exponentially in their distance from the last layer. Based on the aforementioned observations, they introduce the concept of “Effective depth” and present an algorithm that classically simulates the mean value of observables by computing the final logarithmic depth of the circuit.
  They further explore the impact of non-unitary noise on exponential cost concentration and the barren plateau phenomenon, from the perspective of effective depth, providing novel insights.
  
  This work provides a new perspective on using randomness to truncate noisy circuits to shallow-depth quantum circuits. On the other hand, simulations of shallow circuits~(logarithmic depth) still contain rich content. The reference uses common light cone methods for treating shallow circuits, and algorithmic efficiency highly depends on the geometric structure and gate width~(polynomial computational complexity under the setting of 1D and a maximum of 2-qubit gates).

  Our approach does not impose constraints on strong randomness~(e.g., local 2-design), geometric structure and gate width; however, it does entail more requirements regarding noise and the types of gates. Further comprehensive research is needed to explore the trade-offs inherent in these associations.

  \item Noise-induced barren plateaus (NIBPs) were previously shown to be present in sufficiently deep circuits subjected to unital maps~\cite{wang2021noise}. However, Ref.~\cite{singkanipa2024beyond} indicates that VQA circuits may not necessarily exhibit NIBPs under Hilbert-Schmidt~(HS)-contractive maps.
  Moreover, it was observed that in both unital and HS-contractive scenarios, a noise-induced fixed point (NIFP) emerges, causing the cost function to converge to a fixed value for circuits with a depth exceeding logarithmic levels.
\end{itemize}

\end{document}